\documentclass[11pt]{imsart}
\usepackage[utf8]{inputenc}
\usepackage[dvipsnames]{xcolor}
\usepackage{amsmath,bm}
\usepackage{array}
\usepackage{amsthm,amssymb,dsfont}
\usepackage{url}		
\usepackage{xspace}		
\usepackage{booktabs}
\usepackage{ulem}
\usepackage{rotating}	
\usepackage{pdflscape}  
\usepackage{fancyvrb}   
\usepackage{epstopdf}
\usepackage{bbm}	    
\usepackage{graphicx}
\usepackage{amsfonts}
\usepackage{latexsym}
\usepackage{amssymb}
\usepackage{a4wide}
\usepackage{caption}
\usepackage{subcaption}
\usepackage{natbib}

\usepackage[pdftex,                %
bookmarks         = true,
bookmarksnumbered = true,
pdfpagemode       = None,
pdfstartview      = FitH,
pdfpagelayout     = SinglePage,
colorlinks        = true,
linkcolor		  = magenta,
citecolor		  = blue,
urlcolor          = orange,
pdfborder         = {0 0 0}
]{hyperref}

\makeatletter
\newcommand{\addresseshere}{%
  \enddoc@text\let\enddoc@text\relax
}
\makeatother

\renewcommand{\emph}[1]{{\it #1}}

\newcommand{\mbf}{\mathbf}
\newcommand{\E}{\mathbf{E}}

\renewcommand{\P}{\mathbf{P}}
\newcommand{\cH}{\mathcal{H}}

\newcommand{\ind}[1]{{\mbf{1}\{#1\}}}

\newcommand{\mFDR}{\textnormal{mFDR}}

\newcommand{\FWER}{\textnormal{FWER}}

\newcommand{\Abonf}{\mathcal{A}^{\mbox{\tiny OB}}}
\newcommand{\AAbonf}{\mathcal{A}^{\mbox{\tiny AOB}}}
\newcommand{\alphabonf}{\alpha^{\mbox{\tiny OB}}}
\newcommand{\alphaAbonf}{\alpha^{\mbox{\tiny AOB}}}
\newcommand{\alphaLORD}{\alpha^{\mbox{\tiny LORD}}}
\newcommand{\ALORD}{\mathcal{A}^{\mbox{\tiny LORD}}}
\newcommand{\alphaALORD}{\alpha^{\mbox{\tiny ALORD}}}
\newcommand{\AALORD}{\mathcal{A}^{\mbox{\tiny ALORD}}}

\newcommand{\AOBSURE}{\mathcal{A}^{\mbox{\tiny $\rho$OB}}}
\newcommand{\AAOBSURE}{\mathcal{A}^{\mbox{\tiny $\rho$AOB}}}
\newcommand{\alphaOBSURE}{\alpha^{\mbox{\tiny $\rho$OB}}}
\newcommand{\alphaAOBSURE}{\alpha^{\mbox{\tiny $\rho$AOB}}}
\newcommand{\ALORDSURE}{\mathcal{A}^{\mbox{\tiny $\rho$LORD}}}
\newcommand{\AALORDSURE}{\mathcal{A}^{\mbox{\tiny $\rho$ALORD}}}
\newcommand{\alphaLORDSURE}{\alpha^{\mbox{\tiny $\rho$LORD}}}
\newcommand{\alphaALORDSURE}{\alpha^{\mbox{\tiny $\rho$ALORD}}}

\newcommand{\alphaDelay}{\alpha^{\mbox{\tiny DS}}}

\newcommand{\alphaHyb}{\alpha^{\mbox{\tiny Hyb}}}

\newcommand{\delay}{\mathcal{C}}

\newtheorem{theorem}{Theorem}[section]

\newtheorem{proposition}{Proposition}[section]

\newtheorem{definition}{Definition}[section]
\newtheorem{remark}{Remark}[section]
\newtheorem{example}{Example}[section]
\newtheorem{lemma}{Lemma}[section]

\begin{document}
\title{Online multiple testing with super-uniformity reward}
\author{Sebastian D\"ohler, Iqraa Meah and Etienne Roquain}
\runtitle{Online multiple testing with super-uniformity reward}
\date{\today}
\begin{abstract}
	Valid online inference is an important problem in contemporary multiple testing research,
    to which various solutions have been proposed recently. 
    It is well-known that these existing methods can suffer from a significant loss of power if the null $p$-values are conservative. 
    In this work, we extend the previously introduced methodology to obtain more powerful procedures for the case of super-uniformly distributed $p$-values. These types of $p$-values arise in important settings, e.g. when discrete 	 hypothesis tests are performed or when the $p$-values are weighted. 
     To this end,  we introduce the method of super-uniformity reward (SUR)
    that incorporates information about the individual null cumulative distribution functions.
    Our approach yields several new 'rewarded' procedures that offer uniform power improvements over known procedures and come with mathematical guarantees for controlling online error criteria based either on the family-wise 	 error rate (FWER) or the marginal false discovery rate (mFDR). We illustrate the benefit of super-uniform rewarding in real-data analyses and simulation studies. 
    While discrete tests serve as our leading example, we also show how our method can be applied to weighted $p$-values.

\end{abstract}

\maketitle


\section{Introduction}
\subsection{Background}
Multiple testing is a well-established statistical paradigm for the analysis of complex and large-scale data sets, 
in which each hypothesis typically corresponds to a scientific question. 
In the classical situation, the set of hypotheses should be pre-specified before running the statistical inference. 
 However, in contrast to the former 'offline' setting, in many contemporary applications questions arise  sequentially. 
 A first instance of such sequential application is when testing a {\it single} null hypothesis repeatedly as new data are collected, as for continuous monitoring of A/B tests in the information technology industry or marketing research, see \cite{Kohavi2013,johari2019valid} and references therein, or \cite{howard2021time} for recent developments.  A second situation is when the null hypotheses are (potentially) different and arise in a continuous stream,  and accordingly decisions have to be made one at a time and prior to the termination of the stream. 
 This is generally referred to as the \emph{online multiple testing} (OMT) framework and is the focus of this paper, see, e.g., \cite{Lark2017,Robertson2019,kohavi2020online} for application examples. 
This second situation also occurs in combination with the first one to form a `doubly-sequential' experiment \citep{ramdastuto}.
 
\subsection{Existing literature on online multiple testing}

The literature aiming at control of various error rates in OMT has grown rapidly in the last few years. 
As a starting point, the family-wise error rate (FWER) is the probability of making at least one error in the past discoveries, 
and a typical aim is to control it at each time of the stream 
(for a formal definition of this and other error rates, see  Section~\ref{sec:ErrorRatesPower}). 
Since controlling FWER at a given level $\alpha$ is a strong constraint, 
it requires employing a procedure that is conservative, thus generally leading to few discoveries. 
The typical strategy is to distribute over time the initial \emph{wealth} $\alpha$, e.g., 
testing the $i$-th test at level $\alpha \gamma_i$ for a sequence $\{\gamma_i\}_{i\geq 1}$  summing to $1$. 
This approach is generally referred to as \emph{$\alpha$-spending} in the literature (\citealp{FosterStineAlphainvest}). 

A less stringent criterion is the false discovery rate (FDR), which corresponds to the expected proportion of false discoveries. 
This versatile criterion allows many more discoveries than the FWER and has known a huge success 
in offline multiple testing literature since its introduction by \cite{BenjaminiHochberg95}, 
both from a theoretical and practical point of view. 
In their seminal work on OMT, \cite{FosterStineAlphainvest} extended the FDR in an online setting 
by considering the expected proportion of errors 
among the past discoveries 
(actually, considering rather the marginal FDR, denoted below by mFDR, which is defined as the ratio of the expectations, 
rather than the expectation of the ratio). 
The novel strategy in \cite{FosterStineAlphainvest}, which is called \emph{$\alpha$-investing}, 
is based on the idea that an mFDR controlling procedure is allowed 
to recover some $\alpha$-wealth after each rejection, which slows down the natural decrease of the individual test levels. 
In subsequent papers, many further improvements of this method have been proposed:
first, the $\alpha$-investing rule has been generalized by \cite{AharoniRossetGAI}, while maintaining marginal FDR control. 
Later, \cite{JM2018} establish the (non-marginal) FDR control of these rules, 
including the LORD (Levels based On Recent Discovery) procedure. 
Then, a uniform improvement of LORD, called LORD++, has been proposed by \cite{ramdas2017online}, 
that maintains FDR/mFDR control while extending the theory in several directions (weighting, penalties, decaying memory).

Extensions to other specific frameworks have been proposed, including rules that allow asynchronous online testing 
\citep{zrnic2021asynchronous}, maintain privacy \citep{zhang2020paprika}, and
accommodate a high-dimensional regression model \citep{johnson2020fitting}. 
Other online error criteria have also been explored, with false discovery exceedance \citep{JM2018, xuramdas2021dynamic}, 
post hoc false discovery proportion bounds \citep{KR2020}, or confidence intervals with false coverage rate control \citep{pmlr-v119-weinstein20a}. 

Since the online framework is more constrained than the offline framework, 
the employed procedures are generally less powerful in that context. Hence, another important branch of the literature aims at proposing improved rules that gain more discoveries: 
first, following the classical 'adaptive' offline strategy,  procedures can be made less conservative 
by implicitly estimating the amount of true null hypotheses, see the SAFFRON procedure 
for FDR and the adaptive-spending procedure for FWER. 
Second, under an assumption on the null distribution, increasing the number of discoveries is possible by 'discarding'  
tests with a too large $p$-value (\citealp{ramdas2019saffron, tian_onlinefwer_2020, tian_onlinefdr_2019}). 

A power enhancement can also be obtained by combining online procedures with other methods.  
A natural idea is to use more sophisticated individual tests in the first place, e.g., 
based on multi-armed bandits \citep{yang2017framework}, or so-called 'always valid $p$-values', 
see \cite{johari2019valid} and references therein. 
Another idea is to combine offline procedures to form 'mini-batch' rules, see \cite{pmlr-v108-zrnic20a}. 
Further improvements are also possible by incorporating contextual information as done by \cite{pmlr-v108-chen20b} 
or using local FDR-like approach, see \cite{gang2020structureadaptive}.
Lastly, performance boundaries have been derived by \cite{ChenArias2021}.

\subsection{Super-uniformity}\label{sec:superunif}
	
This paper considers OMT in the setting of super-uniformly distributed $p$-values (defined in detail in Section \ref{sec:setting}). 
Super-uniformity may originate from various sources. The first main example we have in mind, and which has been extensively investigated in the statistical literature, is super-uniformity arising from discrete $p$-values (described in detail in Section~\ref{sec:discrete}). 
Additionally, we show that super-uniformity can also be used in a more indirect way as a device for dealing with online $p$-value weighting. In the offline setting, this is a powerful and extensively studied approach, which has, however, in the online case, received little attention so far (described in detail in Section~\ref{sec:weighting}).

{Discrete tests often originate when the tests are based on counts or contingency tables, for example: 
	\begin{itemize}
		\item in clinical studies, the efficiency or safety of drugs are compared by counting 
		patients who survive a certain period after being treated, or who experience a certain type of adverse drug reaction;
		\item in biology, the genotype effect on the phenotype can be tested by knocking out genes sequentially in time. 
	\end{itemize}
	The latter case is met for instance with the data from the International Mouse Phenotyping Consortium (IMPC, see \cite{munozfuentes:hal-02361601}), which
	contains many categorical variables, and thus are described with counts and contingency tables. 
	While this data set is frequently used (see e.g., \citealp{tian_onlinefwer_2020, xuramdas2021dynamic, karp2017prevalence}), the classical OMT procedures do not exploit the discrete nature of the tests, and it turns out that much more powerful procedures can be developed, see Section~\ref{sec:real_data_appli}. }
	
	{In the literature, different solutions have been proposed for dealing with the conservatism of discrete tests, e.g., by modifying directly the $p$-values, either by randomization (see Habiger, 2015 and references therein), or by shrinking them to build so-called mid $p$-values (see Heller and Gur, 2011 and 	references therein). While randomized approaches possess attractive theoretical
	properties, they are often criticized for their lack of reproducibility (see, e.g.,
	Berger, 1996 and Ripamonti et al., 2017). An active research area explores this phenomenon in the offline multiple testing setting, with the seminal works of \cite{Tar1990, WestWolf1997, Gilbert05} 
	and the subsequent studies of \cite{Heyse2011, Heller2012, Dickhaus2012, Habiger2015, CDS2015, Doehler2016, CDH2018, DDR2018, DDR2019}, 
	see also references therein. 	The present work shows that such an improvement is also possible in the online setting, as far as FWER or mFDR control is concerned.  }


\begin{table}[h!]
	\centering
	\begin{tabular}{rlll}
		\multicolumn{1}{c}{Error rate} & \multicolumn{1}{c}{Procedure} & \multicolumn{1}{c}{Critical values} & \multicolumn{1}{c}{Results} \\
		\toprule
		\multicolumn{1}{c}{FWER} & OB &    $
		\alphabonf_T = \alpha \gamma_T
		$    &   \cite{tian_onlinefwer_2020}  \\
		& AOB &    {$\begin{aligned}[t]
		\alphaAbonf_T	 = \alpha (1-\lambda) \gamma_{\mathcal{T}(T)}
			\end{aligned}$}   &    \cite{tian_onlinefwer_2020}      \\
		\addlinespace
		\multicolumn{1}{c}{mFDR} &  LORD     &   $\begin{aligned}[t] \alphaLORD_T =& W_0 \gamma_T 
			+ (\alpha-W_0) \gamma_{T-\tau_1} \\
			& + \alpha\sum_{j\geq 2} \gamma_{T-\tau_j} \end{aligned}$     &      $\begin{aligned}[t] &\mbox{\cite{JM2018}}\\&\mbox{and \cite{ ramdas2017online}} \end{aligned}$  \\
		& ALORD &   $ \begin{aligned}[t] \alphaALORD_T = 
		(1-\lambda) \cdot  & \Big( W_0 \gamma_{\mathcal{T}_0(T)} 
		+ (\alpha-W_0) \gamma_{\mathcal{T}_1(T)} \\
		&  + \alpha\sum_{j\geq 2} \gamma_{\mathcal{T}_j(T)}\Big) \end{aligned}$     &  $\substack{\mbox{\cite{ramdas2019saffron}}\\ \mbox{(slightly improved)}}$ \\
		\bottomrule \\
	\end{tabular}%
\caption{Overview of the critical values of the base procedures for some choice of level $\alpha\in (0,1)$, 
        adaptivity parameter $\lambda\in [0,1)$, initial wealth $W_0\in (0,\alpha)$, 
        and spending sequence $(\gamma_j)_{j\geq 1}$. 
        The quantities $\mathcal{T}(\cdot)$, $\tau_j$, $\mathcal{T}_j(\cdot)$ are given by 
        \eqref{Tronde_def}, \eqref{eqn:tauj}, \eqref{Tronde_def_gen}, respectively.}
	\label{table:OverviewBaseProcedures}%
\end{table}

Finally, weighting $p$-values is a well-established and popular approach for improving the
	performance of offline multiple testing procedures. It can be traced back to \cite{Holm1979}
	and has been further developed, in, e.g., \cite{GRW2006,WR2006,RDV2006,BR2008EJS,RW2009,HZZ2010,ZZ2014,Ign2016,Durand2017,ramdas2019unified}	with weights that can be driven for instance by
	sample size, groups, or more generally by some covariates. 
	By approaching the problem from the perspective of super-uniformity, our general method also allows seamless and flexible integration of such weighting schemes in an online context.

\begin{table}[h!]
	\centering
	\begin{tabular}{rlll}
		\multicolumn{1}{c}{Error rate} & \multicolumn{1}{c}{Procedure} & \multicolumn{1}{c}{Critical values} & \multicolumn{1}{c}{Results} \\
		\toprule
		\multicolumn{1}{c}{FWER} & $\rho$OB &    $\begin{aligned}[t]
		\alphaOBSURE_T &= \alphabonf_T + \sum_{t=1}^{T-1} \gamma'_{T-t} \rho_{t} \end{aligned}
		$    &    Theorem \ref{th:OBSURE} \\
		& $\rho$AOB &    {$\begin{aligned}[t]
		\alphaAOBSURE_T	 &=  \alphaAbonf_T +   
		             \sum_{1\leq t \leq T-1 \atop p_t > \lambda} \gamma'_{T-t} \rho_t 
		              + \varepsilon_{T-1}
			\end{aligned}$}   &       Theorem \ref{th:AOBSURE}   \\
		\addlinespace
		\multicolumn{1}{c}{mFDR} &  $\rho$LORD    &   $\begin{aligned}[t]\alphaLORDSURE_T =& \alphaLORD_T  + \sum_{t=1}^{T-1} \gamma'_{T-t} \rho_{t} \end{aligned}$     &     Theorem   \ref{th:LORDSURE}   \\
		& $\rho$ALORD &   $ \begin{aligned}[t] \alphaALORDSURE_T & = \alphaALORD_T
		+   		\sum_{1\leq t \leq T-1 \atop p_t > \lambda} \gamma'_{T-t} \rho_t 		+ \varepsilon_{T-1}  \end{aligned}$  &    Theorem \ref{th:ALORDSURE} \\
		\bottomrule \\
	\end{tabular}%
	\caption{Overview of the critical values of the rewarded procedures denoted as the corresponding base procedures, 
    with an additional symbol ``$\rho$'' in the name. 
    Here, $\alphabonf_T, \alphaAbonf_T, \alphaLORD_T, \alphaALORD_T$ are the base procedures from Table \ref{table:OverviewBaseProcedures} 
    (with the adaptivity parameter $\lambda$ defined there), 
    $\rho_{t}$ is the super-uniformity reward at time $t$ given by \eqref{eqn:superunifreward}, 
	$\gamma'$ is the SURE spending sequence defined in Section~\ref{sec:WealthSpending}
	and $ \varepsilon_{T} = \ind{ p_{T} < \lambda }(\alpha_{T} - \alpha^0_{T})$ is an additional adaptivity reward, 
    for either $(\alpha^0_{T}, \alpha_{T})=(\alphaAbonf_T,\alphaAOBSURE_T)$, 
    or $(\alpha^0_{T}, \alpha_{T})=(\alphaALORD_T, \alphaALORDSURE_T)$, depending on the case.}
	\label{table:OverviewSUREProcedures}
\end{table}

\subsection{Contributions of the paper}

In this paper, we propose uniform improvements of the classical base procedures listed in Table~\ref{table:OverviewBaseProcedures},
and prove control of the corresponding error rates. 
A distinguishing feature of our work is that we assume that a (non-trivial) upper bound for the null 
cumulative distribution function's (c.d.f.),  
called the \emph{null bounding family},  is \emph{known} {(see Section \ref{sec:setting})}. 
By combining this information with base procedures, we construct more efficient OMT procedures 
(see Table~\ref{table:OverviewSUREProcedures}).  
The key quantity involved in this construction can be interpreted as a reward 
(more details will be provided in Section \ref{sec:sure}) induced by the super-uniformity of the null bounding family. 
Therefore, we use the acronym SUR (Super-Uniform-Reward) to refer to these new procedures. 
When we use the uniform null bounding family (i.e., in the classical framework), our SUR procedures reduce to their base counterparts.
Our main contributions are as follows: 
\begin{itemize}
    \item We propose {two} new SUR procedures for online FWER control in Section~\ref{sec:FWERcontrolSU}: 
        the first {one} ($\rho$OB) uniformly improves upon the Online Bonferroni procedure (OB), 
        while the second ($\rho$AOB) uniformly improves upon the adaptive spending procedure of \cite{tian_onlinefwer_2020} (AOB). 
    \item We propose two new SUR procedures for online mFDR control in Section~\ref{sec:mFDRcontrolSU}: 
        the first one ($\rho$LORD) uniformly improves upon the LORD++ procedures of \cite{JM2018, ramdas2017online} (LORD), 
        while the second one ($\rho$ALORD) uniformly improves upon the SAFFRON procedure of \cite{ramdas2019saffron} (ALORD). 
    \item We present a general and simple way of constructing SUR procedures for any base procedure satisfying some mild conditions, 
        see Section~\ref{sec:genresultFWER} for FWER and Section~\ref{sec:genresultmFDR} for mFDR.
        This allows us to obtain concise proofs for all our results, which are deferred to the supplement, see Section~\ref{sec:proofs}.
    \item Application to discrete data: we evaluate the performances of the new SUR procedures on discrete data, with simulated experiments (Section~\ref{sec:numerical_results}) 
    and for a classical real data set (Section~\ref{sec:real_data_appli}), where each hypothesis is tested using a (discrete) Fisher exact test. The gain in power is shown to be substantial.
   \item Application to $p$-value weighting:  our new SUR procedures can be used to derive weighted online FWER and mFDR controlling procedures. The $p$-value weighting is carried out by rescaling in a certain way the 'raw' weights so that the weighted $p$-value distributions become super-uniform and our methodology can be applied.
The new online procedures are shown to outperform existing ones both on simulated and real data (Section~\ref{sec:weighting}).
\end{itemize}

For easier readability of the paper, a succinct overview of our work is presented in 
Tables~\ref{table:OverviewBaseProcedures}~and~\ref{table:OverviewSUREProcedures}. 
It lists the base and SUR procedures and provides links to definitions and results for error rate control.
All our numerical experiments (simulations and application) are reproducible 
from the code provided in the repository 
\url{https://github.com/iqm15/SUREOMT}.


\subsection{Relation to adaptive discarding}  \label{sec:previous_work}
As \cite{tian_onlinefdr_2019} pointed out, online multiple testing procedures frequently suffer from significant power loss
if the null $p$-values are too conservative. 
In  \cite{tian_onlinefwer_2020} (FWER control) and \cite{tian_onlinefdr_2019} (mFDR control),
the authors propose adaptive discarding (ADDIS) approaches as improved methods. 
In particular, an idea is to use a { discarding} rule, that avoids testing a null when the corresponding $p$-value exceeds a given threshold. 
For the particular type of super-uniformity induced by discrete tests, we show that the discarding rule is less efficient than the SUR method, at least in the settings of Sections~\ref{sec:numerical_results}~and~\ref{sec:real_data_appli}.

\section{Preliminaries} \label{sec:preliminaries}
\subsection{Setting, procedure and assumptions}\label{sec:setting}
Let $X = (X_t, t \in \{1, 2, \dots\})$ be a process composed of random variables.
We denote the distribution of $X$ by $P$, 
which is assumed to belong to some distribution set $\mathcal{P}$.
We consider an online testing problem where,
at each time $t\geq 1$, the user only observes variable $X_t$ and should test a new null hypothesis $H_t$, 
which corresponds to some subset of $\mathcal{P}$, typically defined from the distribution of $X_t$.  
We let $\cH_0 =\cH_0(P) = \{t\geq 1\::\: \mbox{$H_t$ is satisfied by $P$}\}$ 
the set of (unknown) times where the corresponding null hypothesis is true.
Throughout the manuscript, we focus on decisions based upon $p$-values. 
Hence, we suppose that at each time $t$,  
we have at hand a $p$-value $p_t = p_t(X)\in [0,1]$ (typically depending only on $X_t$ although this is not necessary) for testing $H_t$, and 
we consider online multiple testing procedures based on $p$-value thresholding.
This means that each null $H_t$ is rejected whenever $p_t(X) \leq \alpha_t$, 
where $\alpha_t \in [0,\infty)$ is a nonnegative threshold, called {\it a critical value}, 
that is allowed to depend on the past decisions. More precisely, 
we denote $R_t=\ind{p_t(X)\leq \alpha_t}$, $C_t=\ind{p_t(X)\geq \lambda}$ for all $t\geq 1$ and assume that each 
 $\alpha_t$ is measurable with respect to the $\sigma$-field $\mathcal{F}_{t-1} = \sigma(R_1,\dots,R_{t-1}, C_1,\dots,C_{t-1})$.
Here, $\lambda\in [0,1]$ is a parameter that is used for designing adaptive procedures. The particular non-adaptive case is obtained by setting $\lambda=0$, in which case $\mathcal{F}_{t-1} = \sigma(R_1,\dots,R_{t-1})$.
 
In the literature, this property is referred to as predictability, see \cite{ramdas2017online}. 
Throughout the manuscript, an online multiple testing procedure is identified with a family 
$\mathcal{A} = \{\alpha_t, t \geq 1\}$ 
of such predictable critical values. 
Let us now state the assumptions used in what follows. 
First, recall the classical {\it super-uniformity} assumption:
\begin{equation}
    \P_{X \sim P} (p_t(X) \leq u) \leq u \: \text{ for all } u \in[0,1], 
    \mbox{ and }P \in \mathcal{P} \mbox{ with } t \in \mathcal{H}_{0},
    \label{eqn:superunif0}
\end{equation}
which means that each test rejecting $H_{0, t}$ when $p_t(X)$ is smaller than or equal to $u$ is of level $u$. 
Here, we typically consider a setting where these tests may have a more stringent level. 
Formally, at each time $t$, there is a {\it known}  null function $F_t:[0,1] \to [0,1]$ satisfying 
\begin{align}   
    \P_{X \sim P} \left(p_t(X) \leq u \right)\leq F_{t}(u) \leq u,
    \: \text{for all } u \in[0,1], \mbox{ and }P \in \mathcal{P} \mbox{ with } t \in \mathcal{H}_{0}.
    \label{eqn:superunif}
\end{align}
Note that we will sometimes also consider $F_t(u)$ for $u \geq 1$, in which it is to be understood as $F_t(u \wedge 1)$. 
The family $\mathcal{F} = \{F_t, t \geq 1 \}$ will be referred to as the {\it null bounding family}. 
Note that \eqref{eqn:superunif} reduces to \eqref{eqn:superunif0} when choosing $F_t(u) = u$ for all $u$, 
but encompasses other cases by choosing differently the null bounding family. 
Typically, for discrete tests, it is well-known that $F_t(u)$ can be (much) smaller than $u$, 
see Example~\ref{ex:leading} for more details. 
Second, another important assumption is the {\it online independence} within the $p$-value process:
\begin{align} 
    \mbox{$p_t(X)$ is independent of the past decisions $\mathcal{F}_{t-1}$ for all $t \in \cH_0$ and $P\in \mathcal{P}$.}
    \label{indep}
\end{align}
For instance, Assumption \eqref{indep} holds in the case where $p_t(X)$ only depends on $X_t$ and the variables in $(X_t, t \geq 1)$ are all mutually independent, 
which means that the data are collected independently at each time.
\begin{remark}
    In this manuscript, results are often based on assumptions \eqref{eqn:superunif} and \eqref{indep}. 
    In all these results, these two assumptions can be replaced by the weaker condition
        \begin{equation}
            \P_{X \sim P} (p_t(X) \leq u\:|\: \mathcal{F}_{t-1}) \leq F_t(u) \leq u \: 
            \text{ a.s. for all } u \in[0,1], 
            \mbox{ for all  }t \in \mathcal{H}_{0} \mbox{ and }  P \in \mathcal{P} .
            \label{eqn:condsuperunif}
        \end{equation}
    When choosing the null bounding family $F_t(u)=u$ for all $u$, 
    the latter condition is sometimes referred to as SuperCoAD (super-uniformity conditionally on all discoveries), 
    see \cite{ramdas2017online}.
\end{remark} 

{Throughout the paper, we investigate the two following prototypical examples of super-uniformity. }

\begin{example}\label{ex:leading}
    Our leading example is the case where a discrete test statistic is used for inference in each individual test. 
    Typical instances include tests for analyzing counts represented by contingency tables, such as Fisher's exact test, 
    see Section~\ref{sec:numerical_results}. 
    In discrete testing, each $p$-value $p_t(X)$ has its own support $\mathcal{S}_t$ (known and not depending on $P$), 
    that is a finite set (or, in full generality, a countable set with $0$ as the only possible accumulation point). 
    A null bounding family satisfying \eqref{eqn:superunif} can easily be derived by considering
    $F_t$, the right-continuous step function that jumps at each point of $\mathcal{S}_t$, 
    see Figure~\ref{fig:superunifreward} below. 
    Note that the support $\mathcal{S}_t$ depends on $t$ so that discrete testing also induces heterogeneity over time.
\end{example}

\begin{example}\label{ex:weighting}
    Our secondary example is $p$-value weighting, where we start from continuous $p$-values (uniform under the null), which are weighted using  external a priori information in order to increase power, see Section~\ref{sec:weighting}.
\end{example}

\subsection{Error rates and power} \label{sec:ErrorRatesPower}
Let us define the criteria that we use to measure the quality of a given procedure 
$\mathcal{A} = \{\alpha_t, t \geq 1\}$.  
For each $T \geq 1$, let $\mathcal{R}(T) = \{t \in \{1,\dots,T\}\::\: p_t(X)\leq \alpha_t\}$
denote the set of rejection times of the procedure $\mathcal{A}$, up to time $T$.
We consider the two following classical online criteria for type I error rates: 
\begin{align}
    \FWER(\mathcal{A}, P) &:= \sup_{T \geq 1} \{\FWER(T,\mathcal{A},P)\}, \:
    \FWER(T, \mathcal{A}, P) := \P_{X \sim P}\Big( 
    |\cH_0 \cap \mathcal{R}(T)| 
    \geq 1 \Big); \label{eq:DefFWER}\\
    \mFDR(\mathcal{A}, P) &:= \sup_{T \geq 1} \{\mFDR(T, \mathcal{A}, P)\}, \: 
    \mFDR(T, \mathcal{A}, P) := \frac{\E_{X \sim P}\left(|\cH_0 \cap \mathcal{R}(T)|\right)}{\E_{X \sim P}\left(1\vee |\mathcal{R}(T)|\right)}, \label{eq:DefmFDRT}
\end{align} 
with the convention $0/0=0$. 
In words, when controlling the online $\FWER$ at level $\alpha$, one has the guarantee that, at each fixed time $T$, the 
probability of making at least one false discovery before time $T$ is below $\alpha$. 
Since FWER control does not tolerate any false discovery (with high probability),
it is generally considered a stringent criterion.
By contrast, when controlling the online $\mFDR$, at each time $T$, 
the expected number of false discoveries before time $T$ can be non-zero, 
but in an amount controlled by the expected number of discoveries. 
While online FWER has been investigated in \cite{tian_onlinefwer_2020}, 
online mFDR control is generally less conservative (that is, allows more discoveries), 
and is widely used in an online context, see \cite{FosterStineAlphainvest, ramdas2017online, ramdas2019saffron}. 
The false discovery rate (FDR) is close to the mFDR: it is defined by using the expectation of the ratio, 
instead of the ratio of the expectations as in \eqref{eq:DefmFDRT}.
Controlling the FDR generally requires more assumptions, 
while mFDR is particularly useful in an online context 
(we refer the reader to  Section~1.1 of \cite{zrnic2021asynchronous} for more discussions on this).
For a given error rate, we aim at deriving procedures that maximize power.
For any procedure $\mathcal{A}$, we define the power as the expected proportion of signal the procedure can detect, that is,
\begin{align} \label{eq:Defpower}
    \text{Power}(T, \mathcal{A}, P):=  \frac{\E_{X \sim P}\left(|\cH_1 \cap \mathcal{R}(T)|\right)}{1\vee \lvert \cH_1 \lvert},
\end{align}
where $\cH_1$ is the set of times of false nulls, that is, the complement of $\cH_0$ in $\{1,2,\dots\}$.

While this power notion will be used in our numerical experiments to compare procedures, 
our theoretical results will use a stricter comparison criterion. 
For two procedures $\mathcal{A} = \{\alpha_t,t\geq 1\}$ and $\mathcal{A}'=\{\alpha'_t,t\geq 1\}$, 
we say that $\mathcal{A}'$ {\it uniformly dominates} $\mathcal{A}$ when $\alpha'_t \geq \alpha_t$ for all $t \geq 1$ (almost surely). 
This implies that, almost surely, $\mathcal{A}'$ makes more discoveries than $\mathcal{A}$, 
in the sense that the set of discoveries of $\mathcal{A}$ is contained in the one of $\mathcal{A}'$,
that is, $\mathcal{R}(T) \subset \mathcal{R}'(T)$ for all $T \geq 1$ (a.s.). 
In particular, this implies the same domination for the true discovery sets 
and thus in particular $\text{Power}(T, \mathcal{A}, P) \leq \text{Power}(T, \mathcal{A}', P)$ for all $T \geq 1$.
With this terminology, we can restate the aim of this work as follows: construct valid OMT procedures that uniformly dominate their base procedures by 
incorporating the null bounding family $F_t$ given in \eqref{eqn:superunif}.

\begin{remark}
There is no consensus regarding the most adequate definition of power in online testing literature. The concept of uniform domination that we use in this paper is much stronger than, e.g., the asymptotic  power considered by \cite{JM2018}. It may, however, not be particularly appropriate if the base procedure $\mathcal{A}$ is chosen poorly. Since the base procedures given in Table \ref{table:OverviewBaseProcedures} are standard in our settting, the domination criterion seems to be reasonable.
\end{remark}

\subsection{Wealth and super-uniformity reward}\label{sec:sure}
In the Generalized Alpha-Investing (GAI) paradigm (see \cite{xuramdas2021dynamic} and the references given therein), 
the nominal level $\alpha$, at which one wants to control the type I error rate, can be seen as an overall error budget 
--  or \emph{wealth}  -- that may be spent on testing hypotheses in the course of an online experiment. 
For a given OMT procedure $\mathcal{A} $, it is possible to define a suitable \emph{wealth function} $W(T) = W(T, \mathcal{A}, P)$, 
such that $W(T)$ represents the wealth available at time $T$ for further testing. 
As a case in point, \cite{xuramdas2021dynamic} define the (nominal) wealth function for the online Bonferroni 
procedure by $ W^{\textnormal{nom}}(T)= \alpha - \sum_{t=1}^{T} \alpha  \gamma_t$. 
Generalizing this expression for arbitrary null distributions we obtain the 'true' or 'effective' wealth 
$W^{\textnormal{eff}}(T)= \alpha - \sum_{t=1}^{T} F_t(\alpha  \gamma_t)$,
where $F_t$ is a null-bounding function. 
In the super-uniform setting, assumption \eqref{eqn:superunif} implies $W^{\textnormal{nom}}(T) \leq W^{\textnormal{eff}}(T)$, 
and as the two orange curves in Figure~\ref{fig:nom_vs_eff_wealth_fwerproc} illustrate, the discrepancy can be quite large. 
\begin{figure}[h!]
  \centering
  \makebox{\includegraphics[width=0.6\textwidth]{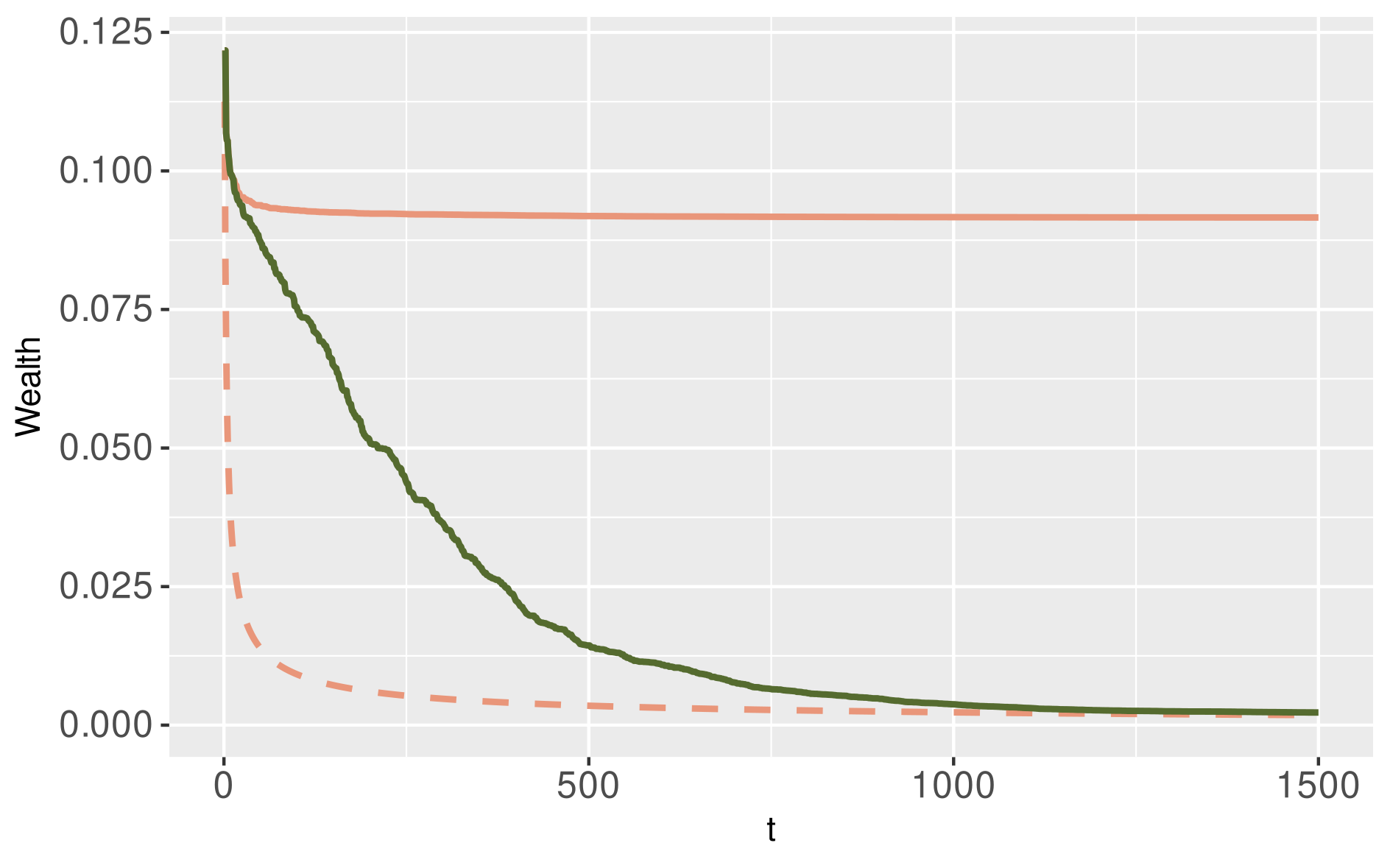}}
\vspace{-0.5cm}
  \caption{Nominal wealth for OB (dashed orange curve), effective wealth for OB (solid orange curve) 
  and effective wealth for $\rho$OB (solid green curve) for the male mice from the IMPC data (see Section \ref{sec:real_data_appli} for more details). }
  \label{fig:nom_vs_eff_wealth_fwerproc}
\end{figure}

However, while the user thinks the procedure is spending the budget over time according to the nominal wealth 
given by the dashed orange curve, 
in reality, the procedure is under-utilizing wealth, as the solid orange true wealth curve indicates.
This unnecessarily austere spending behaviour makes the online Bonferroni procedure sub-optimal.
In addition, this  phenomenon extends to the other procedures 
and error rates listed in Table~\ref{table:OverviewBaseProcedures} as well. 
Our proposed solution incorporates super-uniformity so that its wealth function behaves more like the targeted nominal wealth, 
as depicted by the  green curve in Figure~\ref{fig:nom_vs_eff_wealth_fwerproc}.

For incorporating super-uniformity, we introduce the {\it super-uniformity reward} (SUR), a key quantity in our work.
For any procedure $\mathcal{A} = \{\alpha_t, t \geq 1\}$ 
and  null bounding family $\mathcal{F} = \{F_t, t \geq 1\}$, 
the super-uniformity reward $\rho_t$ at time $t$ is defined by 
\begin{align} 
    \rho_t =\rho_t (\alpha_t, F_t):= \alpha_t - F_t(\alpha_t),\quad t \geq 1.
    \label{eqn:superunifreward}
\end{align} 
Note that \eqref{eqn:superunif} always implies $\rho_t \geq 0$ for all $t \geq 1$.
In the case of discrete testing (Example~\ref{ex:leading}), 
we have $F_t(\alpha_t)=0$ when $\alpha_t$ is below the infimum of the support $\mathcal{S}_t$. 
This produces the maximum possible super-uniformity reward at time $t$, that is, $\rho_t=\alpha_t$. 
Conversely, when $\alpha_t \in \mathcal{S}_t$, we have $F_t(\alpha_t)=\alpha_t$ 
and we have no super-uniformity reward at time $t$, that is, $\rho_t=0$.
In general, we have $\rho_t \in [0, \alpha_t]$, {its actual value depending on}
the discreteness of the test (that is on the steps of $F_t$) and of the value of $\alpha_t$. 
The super-uniformity reward is illustrated
in Figure~\ref{fig:superunifreward} for a single distribution $F_t$ and value $\alpha_t$. 
\begin{figure}[h!]
    \centering
    \makebox{\includegraphics[width=0.7\textwidth]{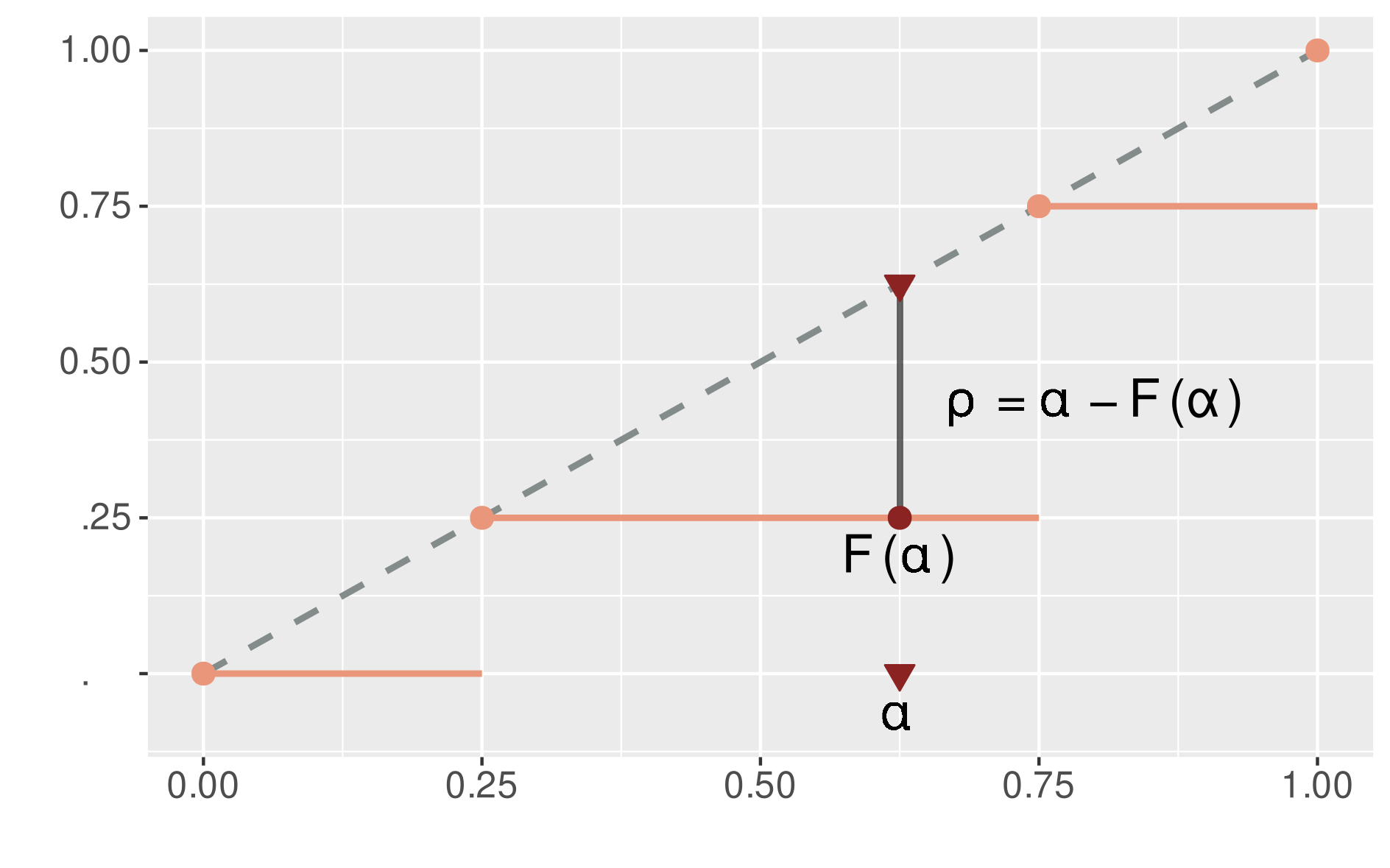}}
    \vspace{-0.5cm}
    \caption{Super-uniformity reward $\rho_t$ at time $t$ (length of the vertical line) 
    as defined by \eqref{eqn:superunifreward} for a given function $F_t$ (orange step function) 
    and a critical value $\alpha_t$ (triangle). The dashed line is the identity function $x \in [0,1]\mapsto x$.}
    \label{fig:superunifreward}
\end{figure}
Mathematically, $\rho_t$ is simply the difference between the nominal significance level $\alpha_t$
and the truly achieved significance level $F_t(\alpha_t)$.  
In terms of wealth, $\rho_t$ can be interpreted as the fraction  
of nominal significance level which the OMT procedure was unable to 'spend' due to super-uniformity. 
Intuitively, it seems clear that this amount can be put aside and be re-allocated to the subsequent tests
to increase the future critical values $(\alpha_{T}, T \geq t+1)$. 
In Sections~\ref{sec:FWERcontrolSU} and~\ref{sec:mFDRcontrolSU},  
we show in detail how this can be done without sacrificing type I error control.

\subsection{Spending sequences} \label{sec:WealthSpending}
As Table \ref{table:OverviewBaseProcedures} displays, 
the base procedures we use are parametrized by a sequence $\gamma = (\gamma_t)_{t \geq 1}$ of non-negative values, 
such that $\sum_{t \geq 1} \gamma_t \leq 1$, which we refer to as the {\it spending sequence}. 
The spending sequence controls the rate at which the wealth is spent in the course of the online experiment 
(for instance, see  \eqref{eqn:OB} for the online Bonferroni procedure).
However, finding suitable spending sequences is not trivial: there is a trade-off between saving wealth for large values of $T$ 
and the ability to make discoveries in the not-too-distant future. 
Typical choices for $\gamma$ in the literature are:
    \begin{itemize}
    \item $\gamma_t \propto t^{-q}$ for all $t$ for some $q>1$, see \cite{tian_onlinefwer_2020};
    \item $\gamma_t \propto (t+1)^{-1} \log^{-q}(t+1)$  for all $t$, for some $q>1$, see \cite{tian_onlinefwer_2020};
    \item $\gamma_t \propto \frac{\log((t+1) \vee 2)}{(t+1)\exp (\sqrt{\log(t+1)})}$, see \cite{JM2018}.
\end{itemize}
Throughout the paper, we choose $\gamma_t \propto t^{-q}$ with $q=1.6$, as suggested by previous literature.
In the base procedures listed in Table \ref{table:OverviewBaseProcedures}, there are two  potential sources of wealth: 
the initial wealth invested at $T=0$, 
and the \emph{rejection reward} that can be earned by rejections for investing procedures (i.e., mFDR controlling procedures). 
When one can use super-uniformity reward as described in Section \ref{sec:sure}, an additional source of wealth comes into play. 
Indeed, our approach is to use an additional \emph{SUR spending sequence} $\gamma'$ 
to smoothly incorporate all the rewards collected up to time $T$ to compute the new critical value $\alpha_T$. 
This SUR spending sequence could be chosen for instance from one of the smoothing sequences listed above. 
{Here, we focus on the following choice:}
\begin{equation}\label{eqn:kernel}
    \gamma'_{t} = \gamma'_{t}(h) = \ind{t \leq h}/h,\:\: \:\: t \geq  1,
\end{equation}
where $h \geq 1$ is a suitably chosen integer. 
Since this leads to procedures that spread rewards uniformly over a finite horizon of length $h$, 
we refer to \eqref{eqn:kernel} --  by analogy with non-parametric density estimation -- 
as a \emph{rectangular kernel} with \emph{bandwidth} $h$.
Finally, another idea introduced by \cite{ramdas2019saffron, tian_onlinefwer_2020} 
in order to slow down the natural decay in the $\alpha_t$ sequence 
is to consider $\gamma_{\mathcal{T}(t)}$ where $\mathcal{T}(t)$ is a slowed down clock, 
see \eqref{Tronde_def} and \eqref{Tronde_def_gen} below. 
As we will see in Section \ref{sec:adaptFWER} and Section \ref{sec:adaptmFDR}, 
this technique can also be combined with a suitable super-uniformity reward.

\section{Online FWER control} \label{sec:FWERcontrolSU}
In this section, we aim at finding procedures $\mathcal{A}$ 
such that $\FWER(\mathcal{A},P) \leq \alpha$ for some targeted level  $\alpha\in (0,1)$.
We begin with a simple application of our approach to improve the online Bonferroni procedure 
with a 'greedy' super-uniformity reward,
and then turn to a smoother spending of the super-uniformity reward (Theorem~\ref{th:OBSURE}). 
This approach is then applied in combination with the adaptive online procedure 
introduced by \cite{tian_onlinefwer_2020} (Theorem~\ref{th:AOBSURE}). 
Finally, a general result is provided (Theorem~\ref{th:genFWER}) that allows to reward 
any procedure controlling the online FWER in some specific way. 
This allows unifying all results obtained in this section 
while further extending the scope of our methodology.

\subsection{Warming-up: online Bonferroni procedure and a first greedy reward} \label{sec:warmingupFWER}
For any given spending sequence  sequence $\gamma = (\gamma_t)_{t \geq 1}$, 
a well-known online FWER controlling procedure is the online Bonferroni procedure, 
$\Abonf = \Abonf(\alpha, \gamma):= \{\alphabonf_t, t \geq 1\}$, defined by
\begin{align} 
    \alphabonf_T:= \alpha \gamma_T, \quad T \geq 1.
    \label{eqn:OB}
\end{align}
It is also called Alpha-Spending rule (\citealp{FosterStineAlphainvest}) in the context of online FWER control, 
see \cite{tian_onlinefwer_2020}.
It is straightforward to check that $\Abonf$ controls the FWER under the classical super-uniformity condition \eqref{eqn:superunif0}:
by the Markov inequality, for all $T \geq 1$,
\begin{align}
    \FWER(T,\Abonf,P) 
    &\leq  \E_{X \sim P} \Big(\sum_{t=1}^T \ind{t \in \cH_0, p_t \leq \alpha \gamma_t}\Big) \label{equ:FWERcontrol} \\
    &\leq \sum_{t \in \cH_0} \P_{X\sim P}(p_t \leq \alpha \gamma_t) 
    \leq \sum_{t \in \cH_0}\alpha \gamma_t \leq \alpha. \label{equ:FWERcontrol2}
\end{align}
Let us now present the rationale behind our approach in this simple case. 
Assume more generally that we have at hand a null bounding family 
$\mathcal{F} = \{F_t, t \geq 1\}$ satisfying \eqref{eqn:superunif}. 
The above reasoning leads  to the following valid bound 
for any procedure $\mathcal{A} = \{\alpha_t, t \geq 1\}$ (with deterministic $\alpha_t$):
\begin{align}
    \FWER(T, \mathcal{A}, P)
    &\leq \sum_{t=1}^T F_t(\alpha_t) \leq \alpha_T + \sum_{t=1}^{T-1} F_t(\alpha_t) 
    = \alpha \sum_{t=1}^{T} \gamma_t \leq \alpha,
    \label{eqn:FWERcontrol3}
\end{align}
by choosing $\alpha_T = \sum_{t=1}^{T} \alpha \gamma_t - \sum_{t=1}^{T-1} F_t(\alpha_t)$.  
The latter is a recursive relation that allows to define a new procedure 
$\mathcal{A} = \{\alpha_t, t \geq 1\}$ controlling the FWER. 
Since $\alpha_1 = \alpha \gamma_1$ and for $T \geq 2$, 
$\alpha_T - \alpha_{T-1} = \alpha \gamma_T - F_{T-1}(\alpha_{T-1})$,
this leads to the simple rule 
\begin{align} 
    \alpha_T = \alpha \gamma_T + \rho_{T-1}, \quad T \geq 1,
    \label{eqn:greedyrhoOB}
\end{align}
where $\rho_{T-1} = \alpha_{T-1} - F_{T-1}( \alpha_{T-1})$ 
is the super-uniformity reward \eqref{eqn:superunifreward} at time $T-1$ (with the convention $\rho_0=0$). 
In addition, from \eqref{eqn:superunif}, we have $\rho_{T-1} \geq 0$,
and the critical values \eqref{eqn:greedyrhoOB} uniformly dominate the online Bonferroni critical values \eqref{eqn:OB} (the obtained critical values are in particular nonnegative, thus defining a valid OMT procedure).
The approach behind critical values \eqref{eqn:greedyrhoOB} is said here to be 'greedy',  
because it spends the complete super-uniformity reward 
$\rho_{T-1}$ obtained at step $T-1$ for increasing the next critical value $\alpha_T$. 

\subsection{Smoothing out the super-uniformity reward} \label{sec:smoothSUFWER}
The greedy policy described in the previous section is not always appropriate 
when time is considered on a potentially large period,
because the sequence of critical values might fall too abruptly.
Instead, we can smooth this effect over time, by distributing the reward collected at time $T-1$ 
over all times following $T$. 
To formalize this idea, we introduce a {\it SUR spending sequence} (see also Section~\ref{sec:WealthSpending}), 
which is defined as a non-negative sequence $\gamma' = (\gamma'_t)_{t \geq 1}$ such that $\sum_{t \geq 1} \gamma'_t \leq 1$. 
While this definition is mathematically the same as the definition of a spending sequence, 
the role of the SUR spending sequence is different, so we use a different name for it.
\begin{definition}
    For any spending sequence $\gamma$ and any SUR spending sequence $\gamma'$, 
    the online Bonferroni procedure with super-uniformity reward, 
    denoted by $\AOBSURE = \{ \alphaOBSURE_t, t \geq 1\}$, is defined by the recursion
    \begin{equation}
        \alphaOBSURE_T = \alpha \gamma_T + \sum_{t=1}^{T-1} \gamma'_{T-t} \rho_{t}, \quad T \geq 1,
        \label{eqn:OBSURE}
    \end{equation}
    where $\rho_{t} = \alphaOBSURE_t - F_t(\alphaOBSURE_t)$ denotes the super-uniformity reward at time $t$ for that procedure. 
\end{definition}
Note that taking $\gamma' = (1, 0, \dots 0)$
recovers the 'greedy' critical values \eqref{eqn:greedyrhoOB}. 
For the rectangular kernel SUR spending sequence given by \eqref{eqn:kernel}, 
we have $\sum_{t=1}^{T-1} \gamma'_{T-t} \rho_{t} =h^{-1} \sum_{t=1\vee (T-h)}^{T-1} \rho_{t}$, 
which we interpret as a uniform spending of the SUR reward over the last $h$ time points. 
As shown in Figure~\ref{fig:smoothedcvs}, the corresponding sequence of critical values (green line) is more 'stable' 
than the one using the greedy approach (blue line), 
allowing for some additional discoveries (on this simulated data).
The following result provides FWER control of the new rewarded critical values \eqref{eqn:OBSURE}, for a general SUR spending sequence.
\begin{theorem} 
    \label{th:OBSURE}
    Consider the setting of Section~\ref{sec:setting},
    where a null bounding family $\mathcal{F} = \{F_t, t \geq 1\}$ satisfying \eqref{eqn:superunif} is at hand.
    For any spending sequence  $\gamma$ and any SUR spending sequence $\gamma'$, 
    consider the online Bonferroni procedure $\Abonf = \{\alphabonf_t, t \geq 1\}$ \eqref{eqn:OB},
    and the online Bonferroni with super-uniformity rewards $\AOBSURE=\{\alphaOBSURE_t, t \geq 1\}$ \eqref{eqn:OBSURE}. 
    Then we have $\FWER(\AOBSURE,P) \leq \alpha$ for all $P \in \mathcal{P}$,
    while $\AOBSURE$ uniformly dominates $\Abonf$.
\end{theorem}
This result will be a consequence of a more general result, see Section~\ref{sec:genresultFWER}.
\begin{figure}[h!]
    \centering
    \makebox{\includegraphics[width=0.7\textwidth]{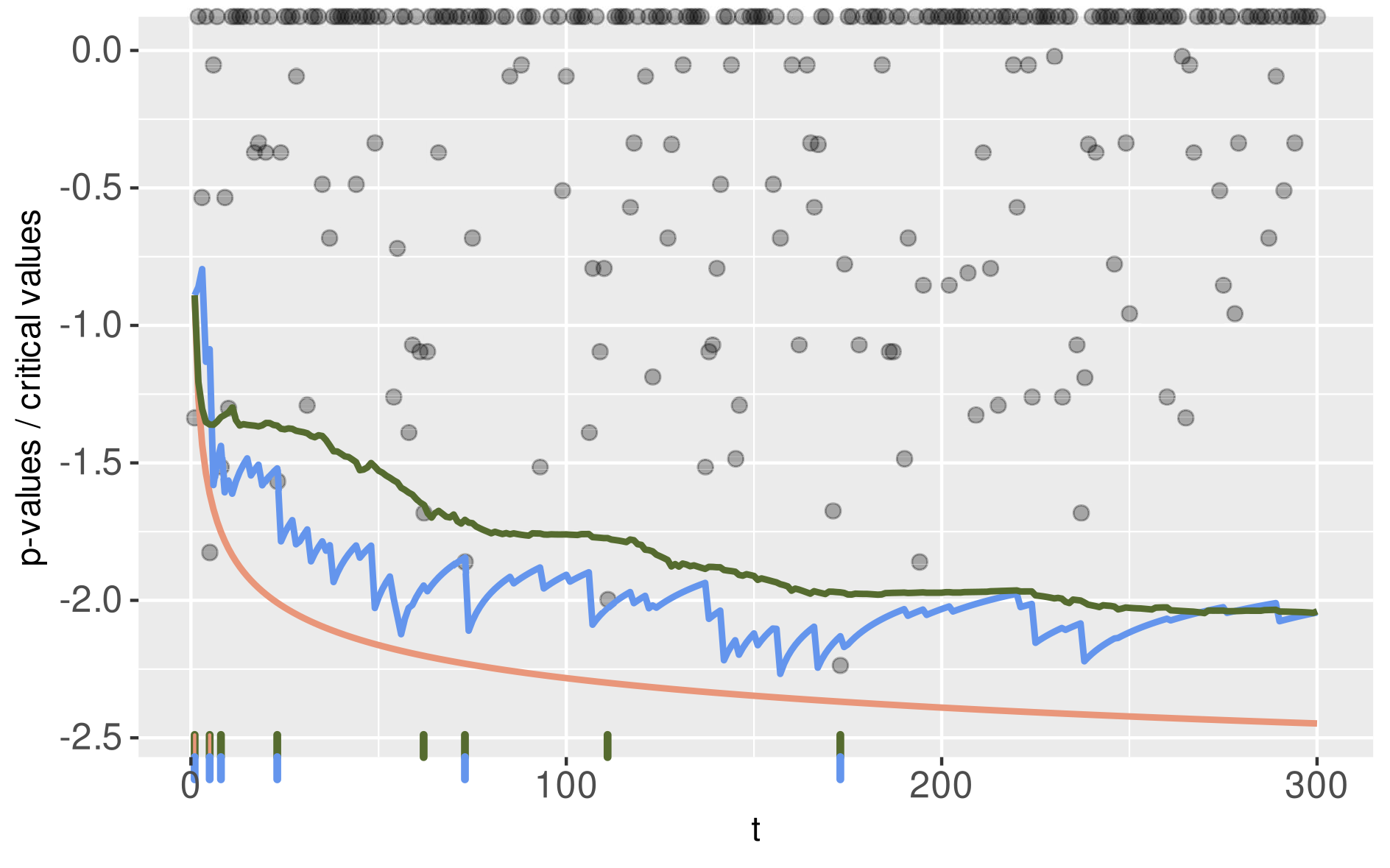}}
        \vspace{-0.5cm}
    \caption{Sequences of critical values for Bonferroni procedures with different rewards 
            over time $1 \leq t \leq T=300$ (simulated data): 
            base Bonferroni critical values \eqref{eqn:OB} (orange line), 
            rewarded with the greedy approach \eqref{eqn:greedyrhoOB} (blue line), 
            and with the rectangular kernel SUR spending sequence  \eqref{eqn:OBSURE} ($h=100$, green line). 
            The rug plots display the time of discoveries for each procedure with the corresponding color.
            The $Y$-axis has been transformed by $y \mapsto -\log(-\log(y))$.  
            The grey dots denote the $p$-value sequence {(those equal to $1$ are displayed at the top of the picture)}.
            The spending sequence is $\gamma_t \propto t^{-1.6}$.
            }
    \label{fig:smoothedcvs}
\end{figure}

\subsection{Rewarded Adaptive Online Bonferroni} \label{sec:adaptFWER}
It is apparent from \eqref{equ:FWERcontrol}-\eqref{equ:FWERcontrol2}
that there is some looseness when upper-bounding 
$\sum_{t \in \cH_0} \gamma_t$ by $ \sum_{t\geq 1} \gamma_t$
which may lead to unnecessarily conservative procedures. 
We may attempt to avoid this loss in efficiency by considering a spending sequence  $\gamma$ satisfying the condition 
$\sum_{t \in \cH_0} \gamma_t \leq 1$ which is more liberal than $\sum_{t \geq 1} \gamma_t \leq 1$. 
In words, this means that the index $t$ in the sequence $\{\gamma_t, t \geq 1\}$ 
should only be incremented when we are testing an hypothesis $H_t$ with $t \in \cH_0$.
Since $\cH_0$ is unknown, such a modification cannot be implemented directly in the $\gamma$ sequence. 
Nevertheless, an approach proposed by \cite{tian_onlinefwer_2020} works by replacing the unknown set $\cH_0$ 
by an estimate $\{1\} \cup \{t \geq 2 \::\: p_{t-1} \geq \lambda\}$ 
for some parameter $\lambda \in (0,1)$, 
and to correct the introduced error in the thresholds $\alpha_t$ to maintain the FWER control.
More formally, we follow \cite{tian_onlinefwer_2020} by introducing the re-indexation functional 
$\mathcal{T}:\{1,\dots\} \to \{1,\dots\}$ defined by
\begin{align} 
    \mathcal{T}(T) = 1 + \sum_{t=2}^{T} \ind{p_{t-1} \geq \lambda}
    , \quad T \geq 1.
    \label{Tronde_def}
\end{align}
Since a large $p$-value is more likely to be linked to a true null, 
$\mathcal{T}(T)$ is used to account for the number of true nulls before time $T$ 
(note that this estimate is nevertheless biased). 
From an intuitive point of view, $\mathcal{T}(T)$ slows down the time by only incrementing time 
when the preceding $p$-value is large enough.
This idea leads to the adaptive online Bonferroni procedure  introduced by \cite{tian_onlinefwer_2020} (called there 'Adaptive spending'
\footnote{The so-called 'discarding' part of the method proposed by \cite{tian_onlinefwer_2020} 
cannot be implemented in our setting because the $F_t$ are not convex, as discussed in Section~\ref{sec:previous_work}.}),
with spending sequence $\gamma$ and \emph{adaptivity parameter} $\lambda \in [0,1)$, 
denoted here by $\AAbonf = \{\alphaAbonf_t, t \geq 1\}$,
and given by
\begin{align} 
    \alphaAbonf_T = \alpha (1-\lambda) \gamma_{\mathcal{T}(T)}, \quad T \geq 1.
    \label{eqn:AOB}
\end{align}
It recovers the standard online Bonferroni procedure when $\lambda = 0$ 
(because $\mathcal{T}(T)=T$ for $T \geq 1$ in that case), 
but leads to different thresholds when $\lambda > 0$. 
Comparing $\AAbonf$ to $\Abonf$, no procedure uniformly dominates the other.
An improvement of $\AAbonf$ over $\Abonf$ is expected to hold 
when there are many false null hypotheses in the data, 
and increasingly so if the signal occurs early in the time sequence, 
see the numerical experiments in Section~\ref{sec:numerical_results}. 
In addition, note that the critical value $\alphaAbonf_T$ depends on the data $X_1, \dots, X_{T-1}$ and thus is random.
As a result, the adaptive approach requires additional distributional assumptions compared with the online Bonferroni procedure. 
In \cite{tian_onlinefwer_2020}, $\AAbonf$ is proved to control the FWER under \eqref{eqn:superunif0} and \eqref{indep} 
(actually under the slightly more general condition \eqref{eqn:condsuperunif} with $F_t$ equal to identity).
Let us now use this approach in combination with the super-uniformity reward. 
\begin{definition}
    For any spending sequence $\gamma$, any SUR spending sequence $\gamma'$, 
    and $\lambda \in [0,1)$, 
    the adaptive online Bonferroni procedure with super-uniformity reward, 
    denoted by $\AAOBSURE=\{\alphaAOBSURE_t, t \geq 1\}$, is defined by 
    \begin{equation}
        \alphaAOBSURE_T = \alpha (1 - \lambda) \gamma_{\mathcal{T}(T)} 
        + \sum_{1 \leq t \leq T-1 \atop p_t \geq \lambda} \gamma'_{T-t} \rho_t + \varepsilon_{T-1}, \quad T \geq 1,
        \label{eqn:AOBSURE}
    \end{equation}
    where $\rho_{t} = \alphaAOBSURE_t - F_t(\alphaAOBSURE_t)$ denotes the super-uniformity reward a time $t$, and 
    $\varepsilon_{T-1} = \ind{ p_{T-1} < \lambda }(\alpha_{T-1} - \alpha (1 - \lambda) \gamma_{\mathcal{T}(T-1)})$
    is an additional {'adaptive' reward} (convention $\varepsilon_0 = 0$).
\end{definition}
This class of procedures reduces to the class of procedures \eqref{eqn:OBSURE}
introduced in the previous section by setting $\lambda = 0$. 
However, when $\lambda > 0$ the class is different 
since the term $\alpha (1 - \lambda) \gamma_{\mathcal{T}(T)}$, which comes from $\alphaAbonf_T$, makes the threshold random. 
Also, the super-uniformity reward is only collected at time $t \leq T-1$ where $p_t \geq \lambda$. 
The latter is well expected from the motivation of the adaptive approach described above: 
when $p_t < \lambda$, no testing is performed so no reward could be obtained from $\rho_t$. 
Nevertheless, note that the additional term $\varepsilon_{T-1}$ allows to collect some reward 
at time $T-1$ in the case where $p_{T-1} < \lambda$. Since this term only appears in critical values of adaptive procedures, 
we call it the 'adaptive' reward. It is linked to the super-uniformity reward in that no adaptive reward can be obtained 
if no super-uniformity reward has been collected in the past. 
The following result shows that this approach is valid from the FWER control perspective. 
\begin{theorem}
    \label{th:AOBSURE}
    Consider the setting of Section~\ref{sec:setting} 
    where a null bounding family $\mathcal{F} = \{F_t, t \geq 1\}$ satisfying \eqref{eqn:superunif} is at hand.
    For any spending sequence $\gamma$, any SUR spending sequence $\gamma'$ and $\lambda\in [0,1)$, 
    consider the adaptive online Bonferroni procedure $\AAbonf = \{\alphaAbonf_t, t \geq 1\}$ \eqref{eqn:AOB} 
    and the adaptive online Bonferroni with super-uniformity rewards $\AAOBSURE = \{\alphaAOBSURE_t, t \geq 1\}$ \eqref{eqn:AOBSURE}.
    Then, assuming that the model $\mathcal{P}$ is such that \eqref{indep} holds, 
    we have $\FWER(\AAOBSURE,P) \leq \alpha$  for all $P \in \mathcal{P}$,
    while $\AAOBSURE$ uniformly dominates $\AAbonf$.
\end{theorem} 
Theorem~\ref{th:AOBSURE} relies on a more general result (Theorem~\ref{th:genFWER} below).
Note that, contrary to Theorem~\ref{th:OBSURE}, 
Theorem~\ref{th:AOBSURE} needs an independence assumption. 
This was already the case without the super-uniformity reward 
since this is due to the adaptive methodology that makes the critical values random. 
If this independence assumption holds, we show in Section~\ref{sec:numerical_results} 
that $\AAOBSURE$ can indeed improve $\AOBSURE$, 
while it always improves the procedure $\AAbonf$ 
of \cite{tian_onlinefwer_2020} (as guaranteed by the above theorem).

\subsection{Rewarded version for base FWER controlling procedures} \label{sec:genresultFWER}
In this section we present a general result stating that any procedure ensuring 
online FWER control (in a specific way) 
can be rewarded using super-uniformity while maintaining the FWER control.  
\begin{theorem} \label{th:genFWER}
    Assuming that \eqref{eqn:superunif} holds, 
    consider any procedure $\mathcal{A}^0=(\alpha^0_t, t \geq 1)$ satisfying almost surely, 
    for some $\lambda \in [ 0,1)$ 
    and for all $T \geq 1$,
    \begin{align} 
        \alpha^0_T + \sum_{1\leq t\leq T-1, \atop p_t \geq \lambda } \alpha^0_t \leq (1-\lambda) \alpha.
        \label{eqn:conditionalpha0}
    \end{align}
    Then the following holds:
    \begin{itemize}
        \item[(i)] $\mathcal{A}^0$ controls the online FWER, 
        that is, $\FWER(\mathcal{A}^0,P) \leq \alpha$  for all $P \in \mathcal{P}$, 
        either if the $\alpha^0_T$ are deterministic for all $T \geq 1$, or if \eqref{indep} holds;
        \item[(ii)] for any SUR spending sequence $\gamma'=(\gamma'_t,t\geq 1)$, 
        the procedure $\mathcal{A}=(\alpha_t,t\geq 1)$, 
        corresponding to the rewarded $\mathcal{A}^0$, and defined by 
    \begin{align} 
        {\alpha}_T = \alpha^0_T
        + \sum_{1\leq t \leq T-1 \atop p_t \geq \lambda} \gamma'_{T-t} ({\alpha}_t - F_t({\alpha}_t)) 
        +  \ind{ p_{T-1} < \lambda }({\alpha}_{T-1} - \alpha^0_{T-1}) , \quad T\geq 1,
        \label{eqn:generalreward}
    \end{align}
    controls the online FWER, 
    that is, $\FWER(\mathcal{A},P) \leq \alpha$  for all $P \in \mathcal{P}$, 
    either if the $\alpha_T$ are deterministic for all $T \geq 1$, or if \eqref{indep} holds.
    \end{itemize}
\end{theorem}
Theorem~\ref{th:genFWER} is proved  in Section~\ref{sec:proofFWER}. 
Condition \eqref{eqn:conditionalpha0} is essentially the same as Condition~(20) 
derived in \cite{tian_onlinefwer_2020}.
It is satisfied by the online Bonferroni procedure ($\mathcal{A}^0=\Abonf$), 
and the online adaptive Bonferroni procedure ($\mathcal{A}^0=\AAbonf$). 
While this is obvious for $\Abonf$, the case of $\AAbonf$ requires to carefully check how the functional 
$\mathcal{T}(\cdot)$ \eqref{Tronde_def} slows down the time, which is done in Lemma~\ref{lem:Tronde}. 
Statement (i) of Theorem~\ref{th:genFWER} thus proves the online FWER control for these procedures.
Statement (ii) of Theorem~\ref{th:genFWER} is our main contribution and reduces to Theorems~\ref{th:OBSURE}~and~\ref{th:AOBSURE}, 
when choosing $\mathcal{A}^0=\Abonf$ and $\mathcal{A}^0=\AAbonf$, respectively. 
This recovers the rewarded procedures $\AOBSURE$ and $\AAOBSURE$ discussed in the previous sections: 
compare \eqref{eqn:generalreward}  to \eqref{eqn:OBSURE} (with $\lambda=0$),
and \eqref{eqn:generalreward} to \eqref{eqn:AOBSURE}.
Nevertheless, other choices for $\mathcal{A}^0$ satisfying \eqref{eqn:conditionalpha0} are possible. 
According to our general result, any such choice is compatible with our reward methodology.

\section{Online mFDR control} \label{sec:mFDRcontrolSU}
In this section, we aim at finding procedures $\mathcal{A}$ 
such that $\mFDR(\mathcal{A},P) \leq \alpha$ for some targeted level  $\alpha \in (0,1)$.
We follow the same route as for the FWER: 
we start with an application of the super-uniformity reward to the classical LORD++ procedure 
(\citealp{ramdas2017online}, called just LORD hereafter for short), 
and then turn to adaptive counterparts. 
Finally, we propose a general result encompassing all these cases. 
In this section, we follow the notation of \cite{ramdas2017online} for online mFDR control. 
For any procedure $\mathcal{A}=\{\alpha_t, t \geq 1\}$ and realization of the $p$-value process, 
let us denote  
\begin{equation}
    R(T)=\sum_{t=1}^T \ind{p_t(X) \leq \alpha_t}
    \label{eqn:RT}
\end{equation}
the number of rejections of the procedure up to time $T$,
and 
\begin{equation}
    \tau_j = \min\{t \geq 1\::\:R(t) \geq j\}\:\:\:\mbox{ ($\tau_j = + \infty$ if the set is empty)},
    \label{eqn:tauj}
\end{equation}
the first time that the procedure makes $j$ rejections, for any $j \geq 1$.

\subsection{Warming up: LORD procedure and a first greedy reward} \label{sec:warmingupmFDR}
While a sufficient condition for online FWER control is $\sum_{t \geq 1} \alpha_t \leq \alpha$ 
(see the previous section and in particular \eqref{eqn:conditionalpha0}), 
the mFDR control is ensured when $\sum_{t \geq 1} \alpha_t \leq \alpha  (1\vee R(T))$, 
as proved in Theorem~2 of \cite{ramdas2017online} 
(applicable, e.g., under assumptions \eqref{eqn:superunif} and \eqref{indep}).
Consequently, for each rejection we earn back wealth $\alpha$ with which we are allowed to increase $\alpha_t$; 
typically by starting a new online Bonferroni critical value process. 
This idea is referred to as $\alpha$-investing in the literature, 
see \cite{FosterStineAlphainvest, AharoniRossetGAI, JM2018}.
This idea leads to the LORD (Levels based On Recent Discovery) procedure (\citealp{JM2018}), 
with the improvement given by \cite{ramdas2017online}:
\begin{align} 
    \alphaLORD_T = W_0 \gamma_T + (\alpha-W_0) \gamma_{T-\tau_1} + \alpha\sum_{j\geq 2} \gamma_{T-\tau_j},\quad T \geq 1,
    \label{eqn:LORD}
\end{align}
where by convention $\gamma_t = 0$ at any time $t \leq 0$ and where $\gamma$ is an arbitrary spending sequence. 
Note that the test level at time $T$ splits the initial $\alpha$-wealth between the cases where $R(T)=0$ and $R(T)=1$,
because the bound is equal to $\alpha (1 \vee R(T)) = \alpha$ in both cases 
so the first rejection does not provide an extra room for false discoveries.
The resulting additional parameter $W_0 \in (0,\alpha)$  
balances the initial $\alpha$-wealth between these two cases to maintain the mFDR control. 
The procedure $\ALORD = \{\alphaLORD_t, t \geq 1\}$ 
controls the mFDR under \eqref{eqn:superunif0} and \eqref{indep}, 
because $\sum_{t\geq 1} \alphaLORD_t\leq \alpha (1\vee R(T))$ (see Section~\ref{sec:proofmFDR} for a proof). 
Now, let us consider our more general framework 
where we have at hand a null bounding family $\mathcal{F}=\{F_t, t \geq 1\}$ 
satisfying \eqref{eqn:superunif}.
In that case, we can prove that a sufficient condition on the critical values for mFDR control is that, almost surely,
\begin{align*}
    \sum_{t=1}^T F_t(\alpha_t)\leq \alpha_T + \sum_{t=1}^{T-1} F_t(\alpha_t)\leq \alpha (1\vee R(T)),
\end{align*}
see the general condition \eqref{conditionalpha0mFDR} below. This can be achieved by choosing 
\begin{align*}
    \alpha_T = \sum_{t=1}^T \alphaLORD_t - \sum_{t=1}^{T-1} F_t(\alpha_t), \quad T \geq 1.
\end{align*}
This leads to the thresholds 
\begin{align}
    \alpha_T = \alphaLORD_T + \rho_{T-1}, \quad T \geq 1,
    \label{eqn:greedyrhoLORD}
\end{align}
where $\rho_{T-1} = \alpha_{T-1} - F_{T-1}(\alpha_{T-1})$ 
is the super-uniformity reward \eqref{eqn:superunifreward} at time $T-1$ (with the convention $\rho_0=0$). 
Since $\rho_t \geq 0$ for all $t$ by \eqref{eqn:superunif}, 
this procedure uniformly dominates the procedure $\ALORD$. 
Furthermore, depending on the magnitude 
of the super-uniformity reward, this new procedure is potentially much more powerful.
\begin{figure}[h!]
    \centering
    \makebox{\includegraphics[width=0.7\textwidth]{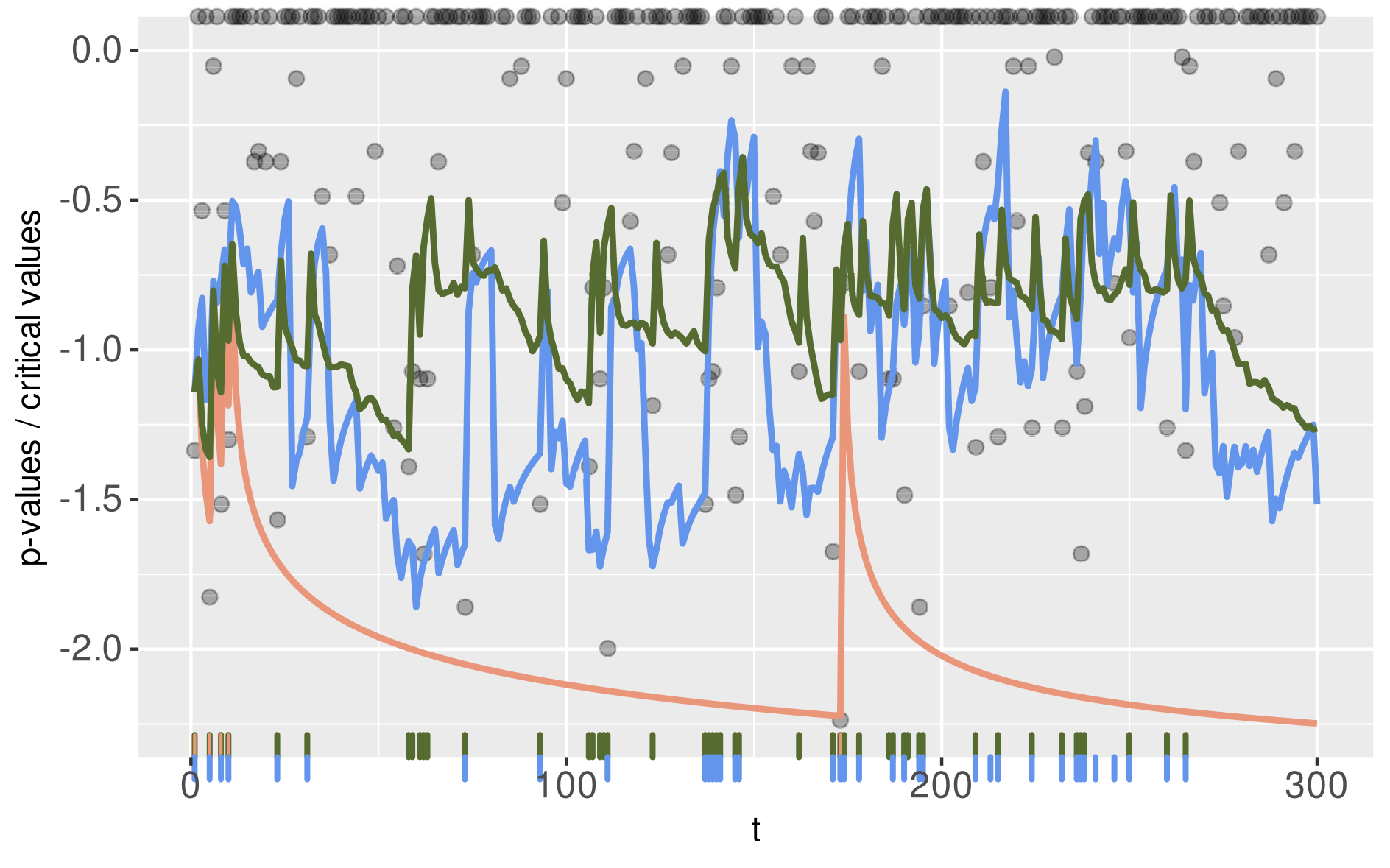}}
        \vspace{-0.5cm}
    \caption{Sequences of critical values of LORD procedure with different rewards over time $1 \leq t \leq T=300$ (simulated data): 
            base LORD critical values \eqref{eqn:LORD}(orange line), 
            rewarded with the greedy approach \eqref{eqn:greedyrhoLORD} (blue line), 
            and with the rectangular kernel SUR spending sequence \eqref{eqn:LORDSURE} ($h=10$, green line). 
            The rug plots display the time of discoveries for each procedure with the corresponding color.
            The $y$-axis has been transformed by $y \mapsto -\log(-\log(y))$. The grey dots denote the $p$-value sequence {(those equal to $1$ are displayed at the top of the picture)}.
            The spending sequence is $\gamma_t \propto t^{-1.6}$.}
    \label{fig:smoothedcvsmfdr}
\end{figure}

\subsection{Smoothing out the super-uniformity reward} \label{sec:smoothSUmFDR}
As discussed for FWER control (see Section~\ref{sec:smoothSUFWER}), 
the preliminary procedure \eqref{eqn:greedyrhoLORD} spends immediately at time $T$ 
all of the super-uniformity reward collected at time $T-1$. 
However, it is more advantageous to redistribute this reward over subsequent times $T, T+1, \dots$, 
by using a SUR spending sequence $\gamma'=(\gamma'_t)_{t\geq 1}$.
This gives rise to the following more general class of online procedures.
\begin{definition}
    For a spending sequence $\gamma$ and a SUR spending sequence $\gamma'$, 
    the LORD procedure with super-uniformity reward, 
    denoted by $\ALORDSURE=\{ \alphaLORDSURE_t, t \geq 1\}$, is defined by the recursion
    \begin{equation}
        \alphaLORDSURE_T = \alphaLORD_T + \sum_{t=1}^{T-1} \gamma'_{T-t} \rho_{t} \quad T\geq 1,
    \label{eqn:LORDSURE}
    \end{equation}
    where $\alphaLORD_T$ is given by \eqref{eqn:LORD} and 
    $\rho_{t} = \alphaLORDSURE_t - F_t(\alphaLORDSURE_t)$ denotes the super-uniformity reward at time $t$.
\end{definition}
Figure~\ref{fig:smoothedcvsmfdr} displays
the critical values of the LORD procedure,
and of those rewarded with the greedy SUR spending sequence $\gamma'=(1,0,\dots)$ 
or rewarded with the rectangular kernel SUR spending sequence \eqref{eqn:LORDSURE} ($h=10$). 
First, the reward given by the $\alpha$-investing, which is possible for mFDR control, 
is visible at each discovery for which all critical value curves 'jump'. 
Second, the effect of the super-uniformity reward is visible between these jumps, 
and the kernel sequence is able to better smooth the critical value sequence.  
As a result, the corresponding procedure is likely to make more discoveries 
(as it is the case on the simulated data presented in Figure~\ref{fig:smoothedcvsmfdr}).
The following result establishes the mFDR control of this new class of rewarded procedures.
\begin{theorem}
    \label{th:LORDSURE}
    Consider the setting of Section~\ref{sec:setting} where a null bounding family $\mathcal{F}=\{F_t,t\geq 1\}$ 
    satisfying \eqref{eqn:superunif} is at hand.
    For any spending sequence $\gamma$ and any SUR spending sequence $\gamma'$, 
    consider the LORD procedure $\ALORD = \{\alphaLORD_t, t \geq 1\}$ \eqref{eqn:LORD}  
    and the LORD procedure with super-uniformity rewards 
    $\ALORDSURE = \{\alphaLORDSURE_t, t \geq 1\}$ \eqref{eqn:LORDSURE}. 
    Then, assuming that the model $\mathcal{P}$ is such that \eqref{indep} holds, 
    we have  $\mFDR(\ALORDSURE,P) \leq \alpha$  for all $P \in \mathcal{P}$ 
    while $\ALORDSURE$ uniformly dominates $\ALORD$.
\end{theorem}
This theorem is proved in Section~\ref{proof:ALORD}, 
as a corollary of a more general result (Theorem~\ref{th:genmFDR} below). 
As shown in the numerical experiments (Section~\ref{sec:numerical_results}), 
the improvement of  $\ALORDSURE$  with respect to  $\ALORD$ can be substantial.

\begin{remark}
$\ALORDSURE$ can be also expressed by using the paradigm of generalized $\alpha$ investing (GAI) rules, as introduced in \cite{FosterStineAlphainvest,AharoniRossetGAI,ramdas2017online}, see Section~\ref{sec:ALORSisGAI}.
\end{remark}

\subsection{Rewarded Adaptive LORD} \label{sec:adaptmFDR}
In this section, we apply the re-indexation trick 
of the $\gamma$ sequence presented in Section~\ref{sec:adaptFWER} 
to improve the performance of the procedures $\ALORD$ and $\ALORDSURE$.
For this, we follow essentially the reasoning used 
by \cite{ramdas2019saffron} for deriving the SAFFRON procedure, 
with a slight modification, as explained below. 
To start, let us define, for some parameter $\lambda \in [0,1)$, 
\begin{align} \label{Tronde_def_gen}
    \mathcal{T}_j(T) = \left\{\begin{array}{ll} 1 + 
    \sum_{t=\tau_{j}+2}^{T} \ind{p_{t-1} \geq \lambda} 
    & \mbox{ if $T \geq \tau_{j}+1$ }\\0&\mbox{ if $T\leq \tau_j$} \end{array}\right. , \quad j\geq 1,
\end{align}
with $\mathcal{T}_0(T) = \mathcal{T}(T)$ given by \eqref{Tronde_def} by convention. 
From an intuitive point of view, $\mathcal{T}_j(T)$ is like a  'stopwatch' starting after $\tau_j$ 
and suspended at each time $t$ for which $p_{t-1} < \lambda$. 
Hence, having $p_t < \lambda$ allows to delay the natural dissipation of $\alpha$-wealth due to online testing.
Then, the SAFFRON procedure \citep{ramdas2019saffron} is defined by the threshold
\begin{align}
    \label{eqn:SAFFRON}
    \alpha_T = 
    \min \left(\lambda, (1-\lambda) \left(W_0 \gamma_{\mathcal{T}_0(T)} 
    + (\alpha-W_0) \gamma_{\mathcal{T}_1(T)} 
    + \alpha\sum_{j\geq 2} \gamma_{\mathcal{T}_j(T)}\right) \right).
\end{align}
This procedure controls the mFDR under \eqref{eqn:superunif0} and \eqref{indep} 
as proved by \cite{ramdas2019saffron}. 
However, examining the proof in \cite{ramdas2019saffron},
it turns out that the capping with $\lambda$ is not necessary. 
The capping prevents the critical values from exceeding $\lambda$, 
thus avoiding to get $p_{t} \geq \lambda$ when $p_t \leq \alpha_t$. 
However, to our knowledge, the latter does not play any role in the mFDR control,
and we work with the (uniformly dominating) procedure
\begin{equation}
    \label{eqn:ALORD}
    \alphaALORD_T
    = (1-\lambda) \left(W_0 \gamma_{\mathcal{T}_0(T)} 
    + (\alpha-W_0) \gamma_{\mathcal{T}_1(T)} 
    + \alpha\sum_{j\geq 2} \gamma_{\mathcal{T}_j(T)}\right) .
\end{equation}
With the capping \eqref{eqn:SAFFRON}, an mFDR control is provided in Theorem~1 in \cite{ramdas2019saffron}. 
For our version \eqref{eqn:ALORD}, the mFDR control follows as a special case of Theorem~\ref{th:ALORDSURE} 
below with $F_t(u)=u$ for all $t,u$.
Also note that $\AALORD$ reduces to $\ALORD$ \eqref{eqn:LORD} when $\lambda = 0$, 
because $\mathcal{T}_j(T)=0 \vee (T-\tau_j)$ in that case. 
Now, we generalize this method to our present framework.
\begin{definition}
    For a spending sequences $\gamma$, a SUR spending sequence $\gamma'$ and $\lambda\in [0,1)$, 
    the adaptive LORD procedure with super-uniformity reward 
    denoted by $\AALORDSURE = \{ \alphaALORDSURE_t, t \geq 1\}$, is defined by 
    \begin{equation}
        \label{eqn:ALORDSURE}
        \alphaALORDSURE_T = 
        \alphaALORD_T + 
        \sum_{1\leq t \leq T-1 \atop p_t \geq \lambda} \gamma'_{T-t} \rho_t 
        + \varepsilon_{T-1}, \quad T\geq 1,
    \end{equation}
    where $\alphaALORD_T$ is defined by \eqref{eqn:ALORD}, 
    $\rho_{t}=\alphaALORDSURE_t - F_t(\alphaALORDSURE_t)$ denotes the super-uniformity reward a time $t$ and 
    $ \varepsilon_{T-1} = \ind{ p_{T-1} < \lambda }(\alphaALORDSURE_{T-1} - \alphaALORD_{T-1})$ 
    is an additional 'adaptive' reward at time $T-1$ (convention $\varepsilon_0=0$).
\end{definition}
Note that  $\AALORDSURE$ reduces to $\ALORDSURE$ \eqref{eqn:LORDSURE} when $\lambda=0$, 
and to $\AALORD$ when $F_t(u)=u$ for all $u,t$. 
The following result shows that this class of procedures controls the mFDR.
\begin{theorem}
    \label{th:ALORDSURE}
    Consider the setting of Section~\ref{sec:setting} 
    where a null bounding family $\mathcal{F}=\{F_t,t\geq 1\}$ satisfying \eqref{eqn:superunif} is at hand.
    For any spending sequence $\gamma$ and any SUR spending sequence $\gamma'$, 
    consider the adaptive LORD procedure $\AALORD = \{\alphaALORD_t, t \geq 1\}$ \eqref{eqn:ALORD},
    and the adaptive LORD procedure with super-uniformity rewards 
    $\AALORDSURE = \{\alphaALORDSURE_t, t \geq 1\}$ \eqref{eqn:ALORDSURE}. 
    Then, assuming that the model $\mathcal{P}$ is such that \eqref{indep} holds, 
    we have $\mFDR(\AALORDSURE,P) \leq \alpha$  for all $P \in \mathcal{P}$ 
    while $\AALORDSURE$ uniformly dominates $\AALORD$ 
    and thus also the SAFFRON procedure of \cite{ramdas2019saffron}.
\end{theorem}
Theorem~\ref{th:ALORDSURE} follows from Theorem~\ref{th:genmFDR} below.
Let us underline that $\AALORDSURE$ both incorporates $\alpha$-investing and super-uniformity reward.
Thus, it is expected to be the most powerful among the procedures considered in the present paper.
This is supported both by the numerical experiments of Section~\ref{sec:numerical_results} and the real data analysis in Section~\ref{sec:real_data_appli}.

\begin{remark}\label{rem:exceed1}
Note that the critical values of ALORD and $\rho$-ALORD can exceed $1$ (e.g., when all $p$-values are zero). Since the rejection decision is the same for a critical value larger than $1$ or equal to $1$, this {may} appear at first sight as wasted wealth. While this is indeed the case for ALORD, we emphasize that this is not the case for $\rho$-ALORD, because 
the super-uniformity reward allows to reuse the exceeding amount of wealth engaged in $\alphaALORDSURE_t$; namely $\rho_t =  \alphaALORDSURE_t -1$ when $\alphaALORDSURE_t\geq 1$.
\end{remark}

\subsection{Rewarded version for base mFDR controlling procedures} \label{sec:genresultmFDR}
The following result establishes that any base online mFDR controlling procedure 
(of a specific type) can be rewarded with super-uniformity.
\begin{theorem} \label{th:genmFDR}
    Assuming that both \eqref{eqn:superunif} and  \eqref{indep}  hold, 
    consider any procedure $\mathcal{A}^0=(\alpha^0_t,t\geq 1)$ 
    satisfying almost surely, for some $\lambda \in [ 0,1)$ 
    and for all $T \geq 1$, 
    \begin{align}
        \alpha^0_T + \sum_{1\leq t\leq T-1, \atop p_t\geq\lambda} \alpha^0_t \leq (1-\lambda) \alpha\:( 1\vee R(T)),
        \label{conditionalpha0mFDR}
    \end{align}
    where $R(T)$ denotes the number of rejections up to time $T$ for this procedure, see \eqref{eqn:RT}.
    Then the following holds
    \begin{itemize}
        \item[(i)] $\mathcal{A}^0$ controls the online mFDR, 
        that is, $\mFDR(\mathcal{A}^0,P) \leq \alpha$  for all $P \in \mathcal{P}$;
        \item[(ii)] for any SUR spending sequence $\gamma'=(\gamma'_t,t\geq 1)$, 
        the procedure $\mathcal{A}=(\alpha_t, t \geq 1)$, corresponding to the rewarded $\mathcal{A}^0$, 
        and defined by \eqref{eqn:generalreward}, 
        controls the online mFDR, that is, $\mFDR(\mathcal{A},P) \leq \alpha$  for all $P \in \mathcal{P}$.
    \end{itemize}
\end{theorem}
Theorem~\ref{th:genmFDR} is proved  in Section~\ref{sec:proofmFDR}. 
Condition \eqref{conditionalpha0mFDR} is essentially the same as the condition found in Theorem~1 of \cite{ramdas2019saffron}. 
Our main contribution is thus in statement (ii), showing that the super-uniformity reward can be used with 
any base procedure $\mathcal{A}^0$ satisfying \eqref{conditionalpha0mFDR}. 
Since the latter condition holds for the LORD procedure $\mathcal{A}^0=\ALORD$,
 and the adaptive LORD procedure $\mathcal{A}^0=\AALORD$ 
(see Lemma~\ref{lem:Tronde}), Theorem~\ref{th:genmFDR} entails Theorem~\ref{th:LORDSURE} and Theorem~\ref{th:ALORDSURE}, respectively. 
Finally, let us emphasize the similarity between Theorem~\ref{th:genFWER} (FWER) and Theorem~\ref{th:genmFDR} (mFDR). 
Strikingly, the reward takes exactly the same form \eqref{eqn:generalreward}, 
which makes the range of improvement comparable for these two criteria.

\section{SUR procedures for discrete tests} \label{sec:discrete}

In this section, we study the performances of our newly derived SUR procedures in discrete online multiple testing problems for simulated and real data. 
We defer some of the numerical results to Appendix~\ref{sec:addnum}.

\subsection{Considered procedures}\label{sec:discreteproc}

The considered procedures are
the base (non-rewarded) procedures 
$\Abonf$ \eqref{eqn:OB}, $\AAbonf$ \eqref{eqn:AOB}, 
$\ALORD$ \eqref{eqn:LORD}, and $\AALORD$ \eqref{eqn:ALORD}, 
and their rewarded counterparts 
$\AOBSURE$ \eqref{eqn:AOBSURE}, $\AAOBSURE$ \eqref{eqn:AOBSURE},
$\ALORDSURE$ \eqref{eqn:ALORDSURE}, and $\AALORDSURE$ \eqref{eqn:ALORDSURE}, respectively.
As mentioned in Section \ref{sec:previous_work}, 
we also consider the ADDIS-spending and ADDIS procedures 
(see \citealp{tian_onlinefwer_2020, tian_onlinefdr_2019}) although 
the type I error rate control is not guaranteed for these two procedures,
in our (discrete) setting.
The parameters of the OMT procedures are set to
$\alpha = 0.2$, $W_0 = \alpha/2$ and $\lambda = 0.5$.  
For ADDIS and ADDIS-spending, we use the default values 
$W_0 = \frac{\alpha  \lambda  \tau }{2}$,  with 
$\lambda = 0.25$ and $\tau = 0.5$ 
(the latter being the discarding parameter, see \citealp{tian_onlinefwer_2020, tian_onlinefdr_2019}). 
Following \cite{tian_onlinefdr_2019}, we set $\gamma_t \propto t^{-1.6}$ 
with a normalizing constant chosen such that $\sum_{t=1}^{+\infty} \gamma_t = 1$.
For the SUR spending sequence $(\gamma'_t)_{t\geq 1}$ 
we use a rectangular kernel with bandwidth $h$, 
as defined by \eqref{eqn:kernel}, with $h=100$ for FWER and $h=10$ for mFDR.
We discuss different choices for tuning parameters in the SUR procedures 
(adaptivity parameter $\lambda$ and the rectangular kernel bandwidth $h$) in Appendices~\ref{apenadapt} and \ref{apenband}.

\subsection{Application to simulated data}\label{sec:numerical_results}

\subsubsection{Simulation setting}
We simulate $m$ experiments in which the goal is to detect differences between two groups 
by counting the number of successes/failures in each group.
More specifically, we follow \cite{Gilbert2005}, \cite{HellerGur2011} and \cite{DDR2018} 
by simulating a two-sample problem in which a vector of $m$ independent binary responses  is observed for $N$ subjects in both groups. 
The goal is to  test the $m$ null hypotheses $H_{0i}$: '$p_{1i} = p_{2i}$', $i = 1,...,m$ in an online fashion,
where $p_{1i}$ and $p_{2i}$ are the success probabilities for the $i^{th}$ binary response in group A and B respectively. 
Thus, for each hypothesis $i$, the data can be summarized by a $2 \times 2$ contingency table,
and we use (two-sided) Fisher's exact test for testing $H_{0i}$. 
The  $m$ hypotheses are split in three groups of size 
$m_1$, $m_2$, and $m_3$ such that $m = m_1 + m_2 + m_3$.
Then, the binary responses are generated as i.i.d Bernoulli of probability 0.01 ($\mathcal{B}(0.01)$) at $m_1$ positions for both groups,
i.i.d $\mathcal{B}(0.10$) at $m_2$ positions for both groups, 
and i.i.d $\mathcal{B}(0.10)$ at $m_3$ positions for one group
and i.i.d $\mathcal{B}(p_3)$ at $m_3$ positions for the other group.
Thus, the null hypotheses are true for $m_1 + m_2$ positions (set $\cH_0$),
while the null hypotheses are false for $m_3$ positions (set $\cH_1$). 
Therefore, we interpret $p_3$ as the strength of the signal while $\pi_{A} = \frac{m_3}{m}$, 
corresponds to the proportion of signal. 
Also, $m_1$ and $m_2$ are both taken equal to $\frac{m - m_3}{2}$.
In these experiments, we fix $m = 500$, and vary each one of the parameters 
$\cH_1$ (Section~\ref{sec:varyH1}), $\pi_A$ (Section~\ref{sec:varypiA}), 
$N$ (Section~\ref{sec:varyN}), $p_3$ (Section~\ref{sec:varyp3}) while keeping the others fixed.  
The default values are $\pi_A = 0.3$, $N = 25$, $p_3 = 0.4$ and $\cH_1 \subset \{1,\dots,m\}$ chosen randomly for each simulation run.
We estimate the different criteria (FWER \eqref{eq:DefFWER}, mFDR \eqref{eq:DefmFDRT}, power \eqref{eq:Defpower}) 
using empirical mean over 10 000 independent simulation trials. 

\subsubsection{Position of signal}\label{sec:varyH1}
 
We start by studying how the position of the signal can affect the performances of the procedures (it is well-known to be critical, see \citealp*{FosterStineAlphainvest, ramdas2017online}).
We investigate different positioning schemes in which the signal can be clustered at the beginning of the stream,
or at the end, or clustered between the two, as described in the caption of Figure \ref{fig:varypos}.
\begin{figure}[h!]
    \centering
    \begin{tabular}{cc}
    \makebox{\includegraphics[width=.49\linewidth]{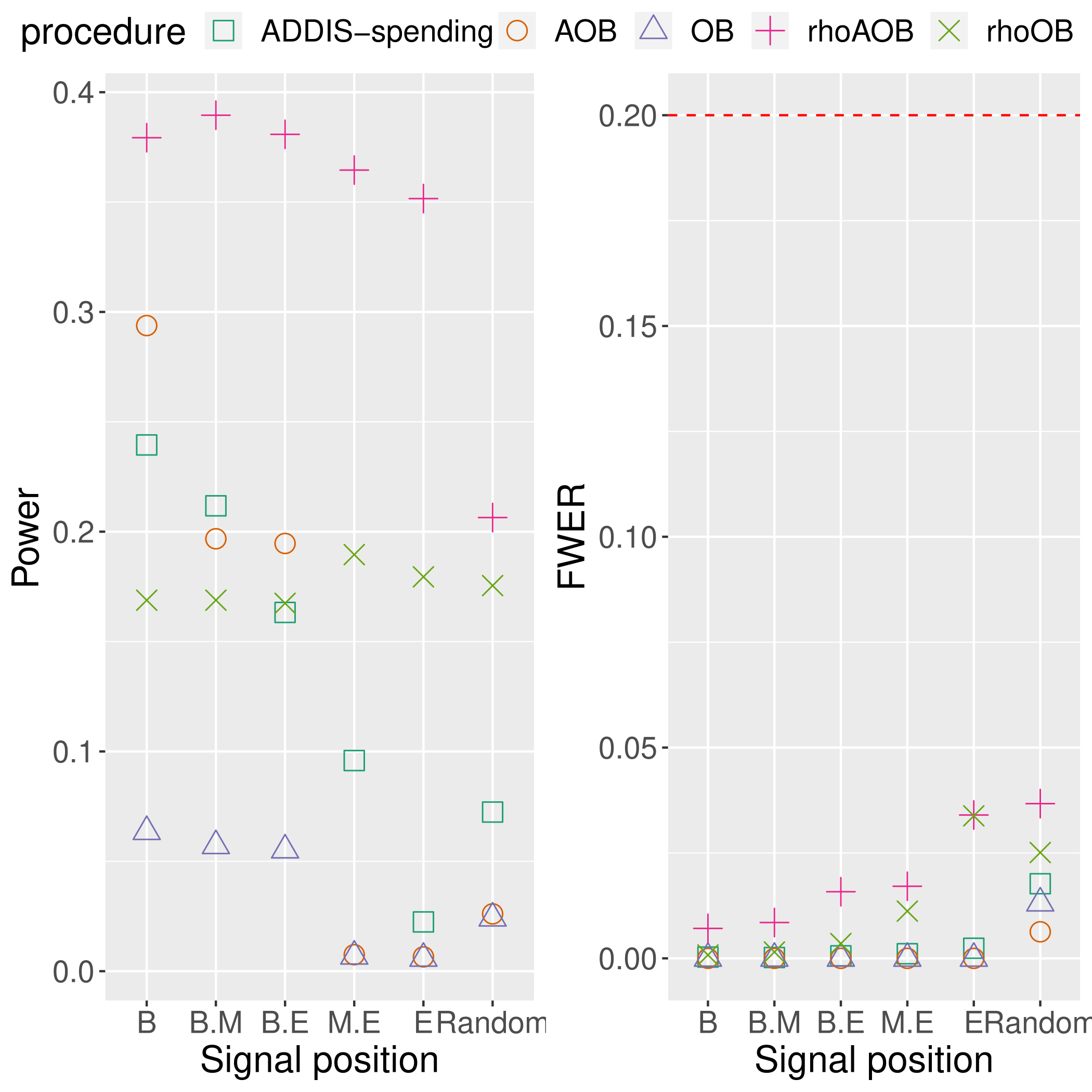}} & 
    \makebox{\includegraphics[width=.49\linewidth]{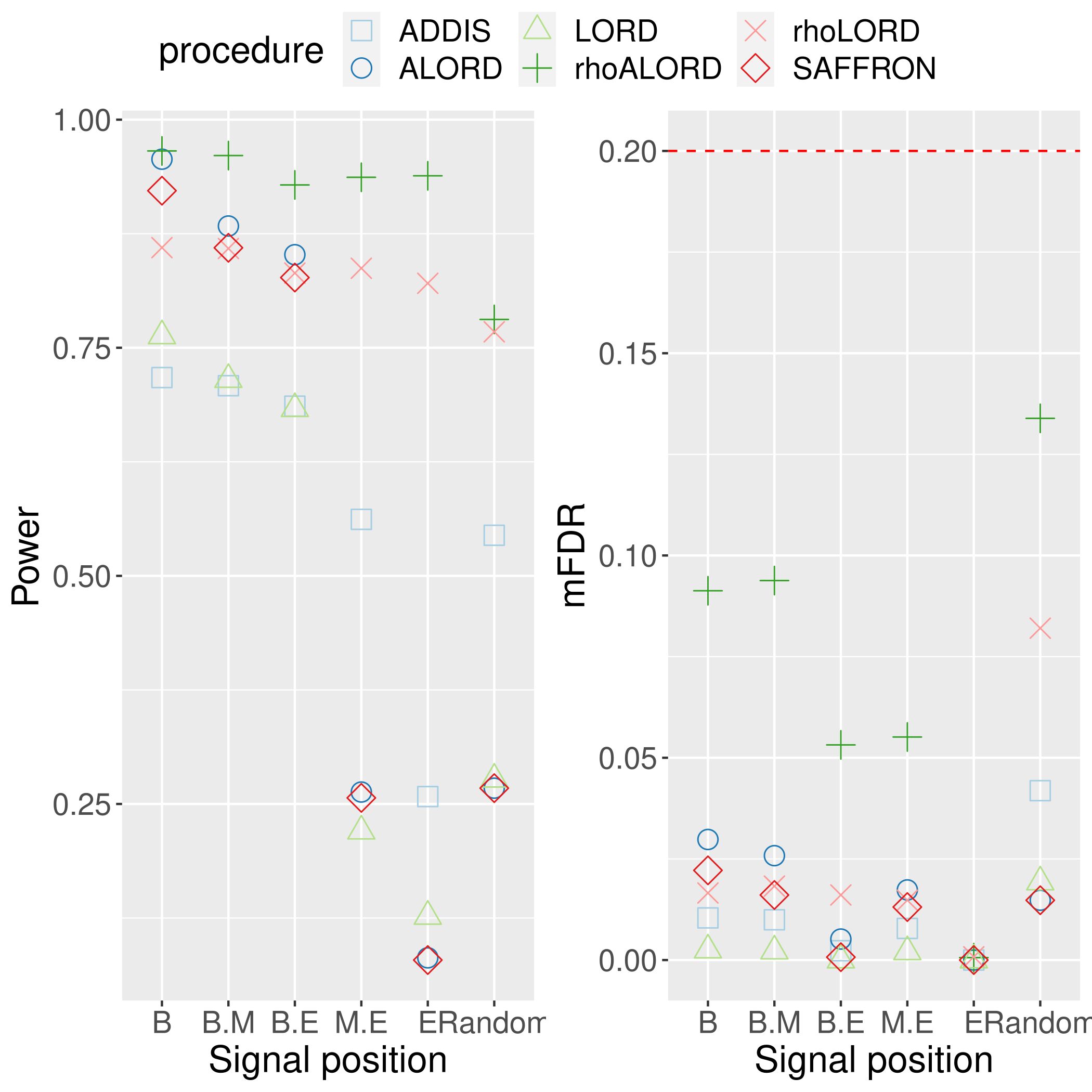}}  
    \end{tabular}
        \vspace{-0.5cm}
    \caption{Power and type I error rates of the different considered OMT procedures versus positions of the signal:
            at the beginning (B), the end (E),
            half at the beginning and half in the middle of the stream (BM), 
            half at the beginning and half at the end of the stream (BE), 
            half in the middle and half at the end of the stream (ME),
            and taken uniformly at random (Random).}
    \label{fig:varypos}
\end{figure}
Consistently with our theoretical results, Figure \ref{fig:varypos} shows that all procedures 
control the type I error rate at level $\alpha = 0.2$. 
In terms of power, we can see that the rewarded procedures have greater power than the associated base procedures.
More specifically, $\AALORDSURE$ uniformly dominates the other procedures for mFDR control and $\AAOBSURE$ for FWER control. 
The gain in power is most noticeable when the signal is not localized at the beginning of the stream 
(i.e. positions ME, E, and Random) for which the online testing problem is more difficult. 
These first results indicate that the rewarded procedures may protect against '$\alpha$-death'.

\subsubsection{Proportion of signal}\label{sec:varypiA}
Figure \ref{fig:varypiA} displays the results for $\pi_A$ varying in $\{ 0.1, \dots ,1\}$.
It shows that the aforementioned superiority of the rewarded procedures holds in this whole range.
Also note that the SUR reward can affect the monotonicity of the power curves: while most curves are increasing 
with $\pi_A$, the power of the rewarded procedure $\AOBSURE$ decreases. 
An explanation could be that when $\pi_A$ increases, the marginal counts increase, and thus the degree of 
discreteness decreases providing a smaller super-uniformity reward. 
However, using adaptivity seems to compensate for this effect, thus providing better results. \\
\begin{figure}[h!]
    \centering
    \begin{tabular}{cc}
    \makebox{\includegraphics[width=.49\linewidth]{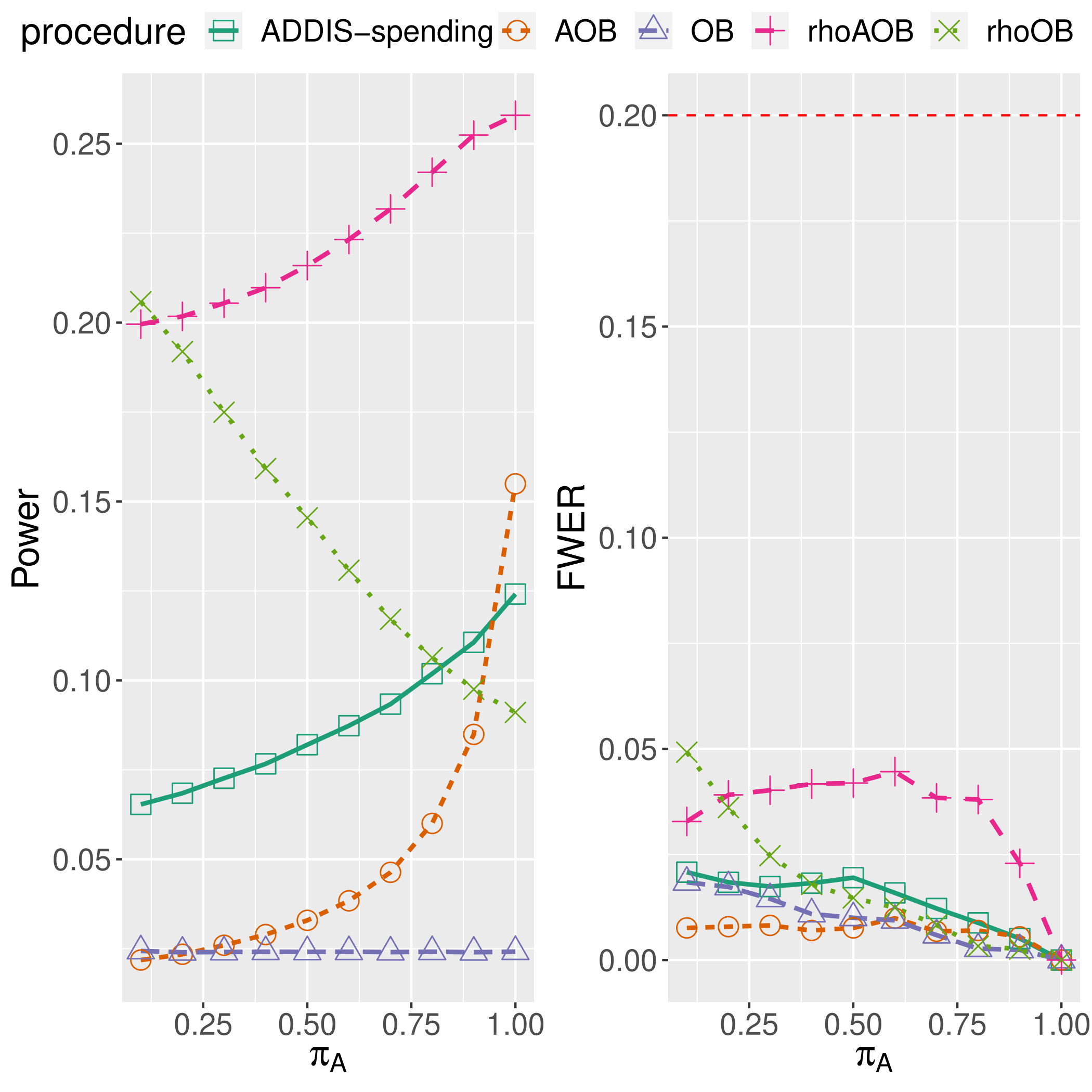}} & 
    \makebox{\includegraphics[width=.49\linewidth]{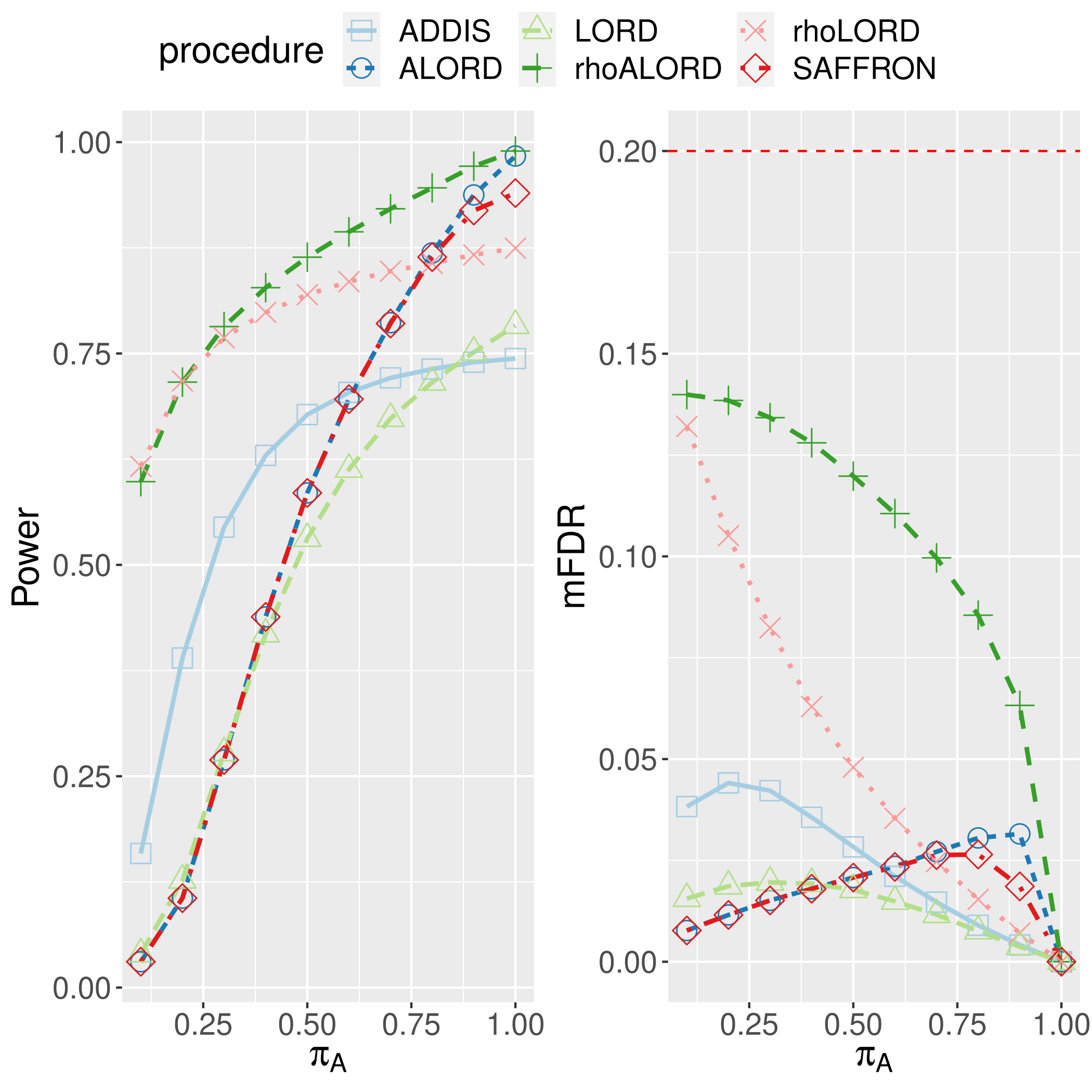}}
    \end{tabular}
        \vspace{-0.5cm}
    \caption{Power and type I error rates of the considered procedures for $\pi_A \in \{0.1, 0.2, \ldots, 0.9, 1 \}$.}
    \label{fig:varypiA}
\end{figure}

Finally, let us mention that the additional numerical results in Section~\ref{sec:addnum} provide qualitatively similar conclusions
for all other explored parameter configurations:
the SUR procedures $\AAOBSURE$ and $\AALORDSURE$ always improve, often substantially, the existing OMT procedures.

\subsection{Application to IMPC data}\label{sec:real_data_appli}
In this section we analyse data from the International Mouse Phenotyping Consortium (IMPC), 
which coordinates studies on the genotype influence on mouse phenotype. 
More precisely, scientists test the hypotheses that the knock-out of certain genes will not change certain phenotypic traits 
(e.g., the coat or eye color). Since the data set  is constantly evolving 
as new genes are studied for new phenotypic traits of interest, 
online multiple testing is a natural approach for analysing such data, see also \cite{tian_onlinefwer_2020,xuramdas2021dynamic}. 
We use the data set  provided by \cite{karp2017prevalence} which includes,
for each studied gene, the count of normal and abnormal 
phenotype for female and male mice (separately), thus providing two by two contingency tables, 
which can be analysed using Fisher exact tests. 
In this section, we  investigate the genotype effect on the phenotype 
separately for male and female.
The data set  originally contains nearly $270\,000$ genes studies, but we focus on the first $30\,000$ genes for simplicity.
We set the global  level $\alpha$ to $0.2$ and $0.05$, respectively for FWER and mFDR procedures. 
For the procedure parameters, we follow the choice made in Section~\ref{sec:discreteproc}. Table~\ref{tabmice} presents the number of discoveries for the FWER controlling procedures OB, AOB, $\rho$OB, $\rho$AOB (left) 
and for the mFDR controlling procedures LORD, ALORD, $\rho$LORD, $\rho$ALORD (right). 
The results show that ignoring the discreteness of the tests causes the scientist to miss (potentially many) discoveries. 
Hence, using the SUR methods helps to reduce this risk. 

\begin{table}
\caption{\label{tabmice} Number of discoveries     for FWER controlling OMT procedures (left) and mFDR controlling OMT procedures (right).     These numbers are obtained by running the procedures on the first $30\,000$ genes for male (second row) and female (third row) mice in the IMPC data.}
\centering
\fbox{%
\begin{tabular}{*{9}{c}} 
 \hline
 Procedures & OB & $\rho$OB & AOB & $\rho$AOB & LORD& $\rho$LORD  &ALORD & $\rho$ALORD \\ [0.5ex] 
 \hline
 $\#$ discoveries (male) & 229 & 377 & 281 & 697 & 882 &972  &  972 & 1041\\
 $\#$ discoveries (female) & 267 & 481 & 764 & 811 & 839  & 946 & 966 & 1046\\
 \hline
\end{tabular}}
\end{table}

Figure~\ref{fig:mfdrimpcbasevsrewarded} (FWER procedures) and Figure~\ref{fig:fwerimpcbasevsrewarded} (mFDR procedures) 
illustrate in more detail how the super-uniformity reward leads to more discoveries, in the case of male mice 
(similar findings hold for the female mice for which the corresponding figures can be found in Section~\ref{apenIPMCfemale}). 
First, note that the smallest $p$-values occur at the beginning of the stream 
(see Figure \ref{fig:impc:male:pvalues} in Section~\ref{apenIPMCLocalization}), 
so that we limit the visual analysis to the first $1500$ $p$-values for clarity of exposition. 
For the $\rho$OB procedure, the benefit  of incorporating the super-uniformity reward is visible 
in the left panel of Figure~\ref{fig:mfdrimpcbasevsrewarded}. As expected from Figure~\ref{fig:smoothedcvs}, 
applying a rectangular kernel to these rewards yields a smooth curve. 
For the $\rho$AOB procedure, presented in the right panel of Figure~\ref{fig:mfdrimpcbasevsrewarded}, the 
improvement is even stronger, but the resulting critical value curve is less smooth. 
This is due to the 'adaptive' reward, that is, the  $\varepsilon_{T-1}$-component of our improvement, 
recall \eqref{eqn:AOBSURE}. 
More precisely, an explanation of this 'saw-tooth' shape is that during a period with $p$-values smaller than $\lambda$, 
we have $\alphaAOBSURE_T	- \alphaAbonf_T\geq  \alphaAOBSURE_{T-1}	- \alphaAbonf_{T-1}$ so the gain increases. 
Also, if this period lasts for a while (as for $500\lesssim t \lesssim 1240$ here), 
the $\rho$-part of the reward vanishes and we end up with a constant gain 
$\alphaAOBSURE_T	- \alphaAbonf_T\approx  \alphaAOBSURE_{T-1}	- \alphaAbonf_{T-1}$, explaining the flat part of the curve, 
until the next $p_T \ge \lambda$ occurs. 
After this point, we switch from the $\varepsilon$-regime back to the $\rho$-regime, 
i.e., $\alphaAOBSURE_{T+1}=\alphaAbonf_{T+1}+ \gamma'_{1}\rho_{T}$. 
Since typically $\gamma'_{1}\rho_{T} \ll \alphaAOBSURE_{T-1}-\alphaAbonf_{T-1}$, this causes the downward jump in the green curve.
For the mFDR procedures presented in Figure~\ref{fig:fwerimpcbasevsrewarded}, 
there is an additional 'rejection' reward as described in Section~\ref{sec:mFDRcontrolSU}. 
 Note that this makes some critical values exceed $1$ (both for ALORD and $\rho$ALORD), which thus cannot be displayed in the $Y$-axis scale considered in that figure. 
However, these values are still used in $\rho$ALORD algorithm to compute the future critical values (see Remark~\ref{rem:exceed1}).
The obtained results are qualitatively similar to the FWER setting: our proposed reward makes the green curves run above the orange ones, 
uniformly over the considered time, hence inducing significantly more discoveries. 

\begin{figure}[h!]
    \centering
    \begin{tabular}{cc}
    \makebox{\includegraphics[width=.49\linewidth]{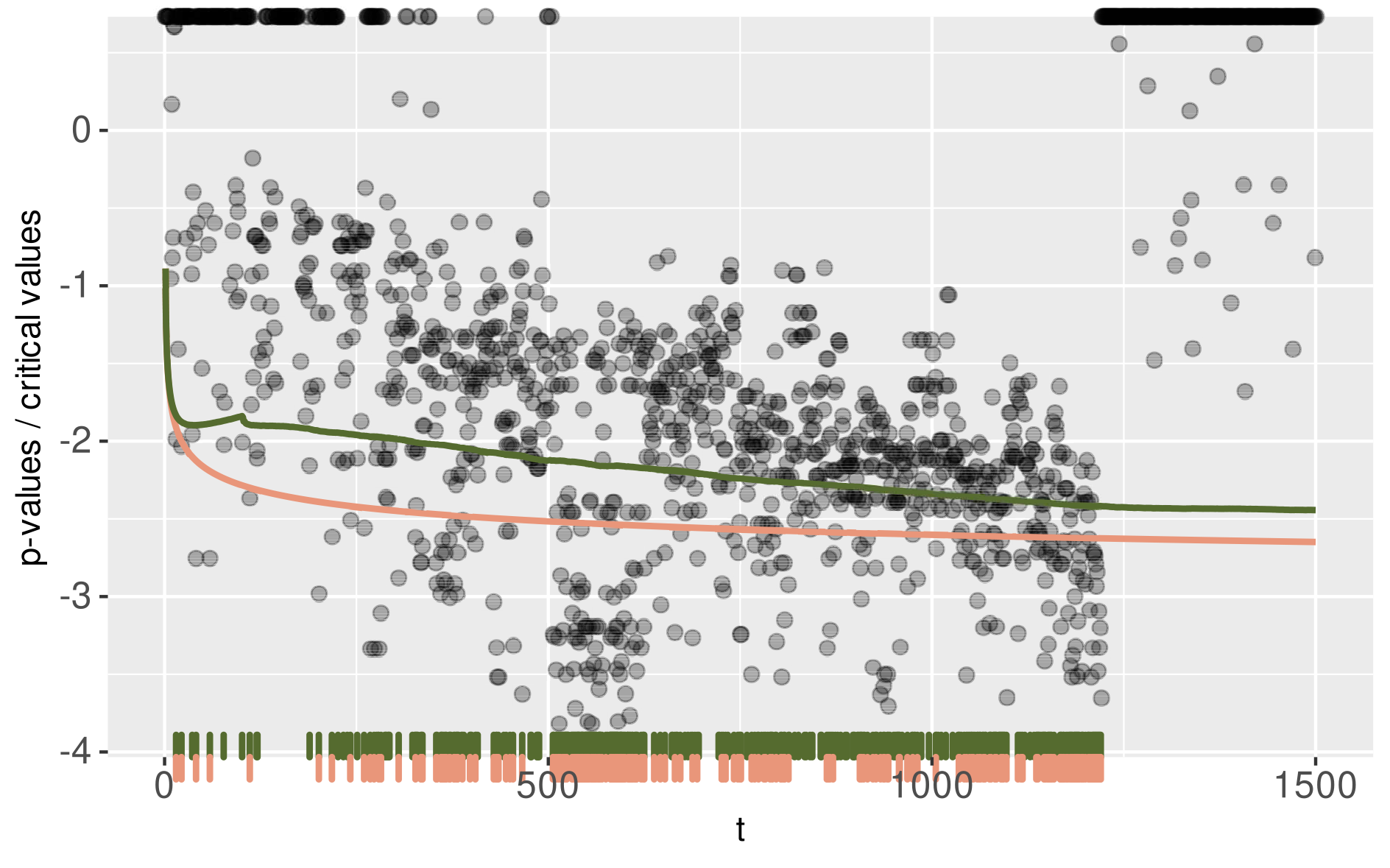}} & 
    \makebox{\includegraphics[width=.49\linewidth]{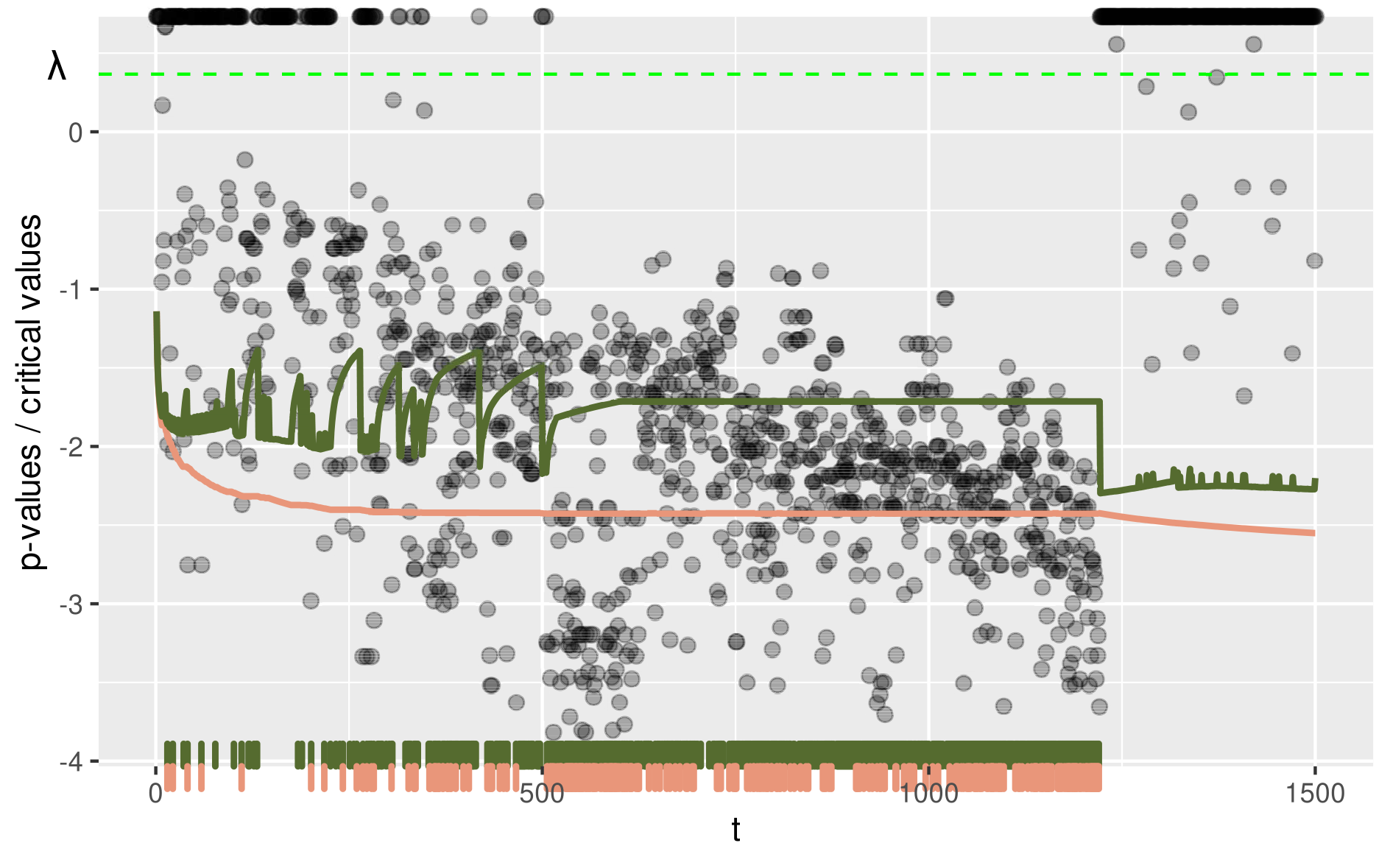} }
    \end{tabular}
        \vspace{-0.5cm}
    \caption{{Applying online FWER controlling procedures to the male mice IMPC data set.} {Left panel: $p$-values and critical values for OB (orange curve) and $\rho$OB (green curve). Right panel:  AOB (orange curve) and $\rho$AOB (green curve).} Representation similar to Figure~\ref{fig:smoothedcvs} {($Y$-axis transformed by $y \mapsto -\log(-\log(y))$; $p$-values equal to $1$ displayed at the top of the picture)}. }
    \label{fig:mfdrimpcbasevsrewarded}
\end{figure}
\begin{figure}[h!]
    \centering
    \begin{tabular}{cc}
    \makebox{\includegraphics[width=.49\linewidth]{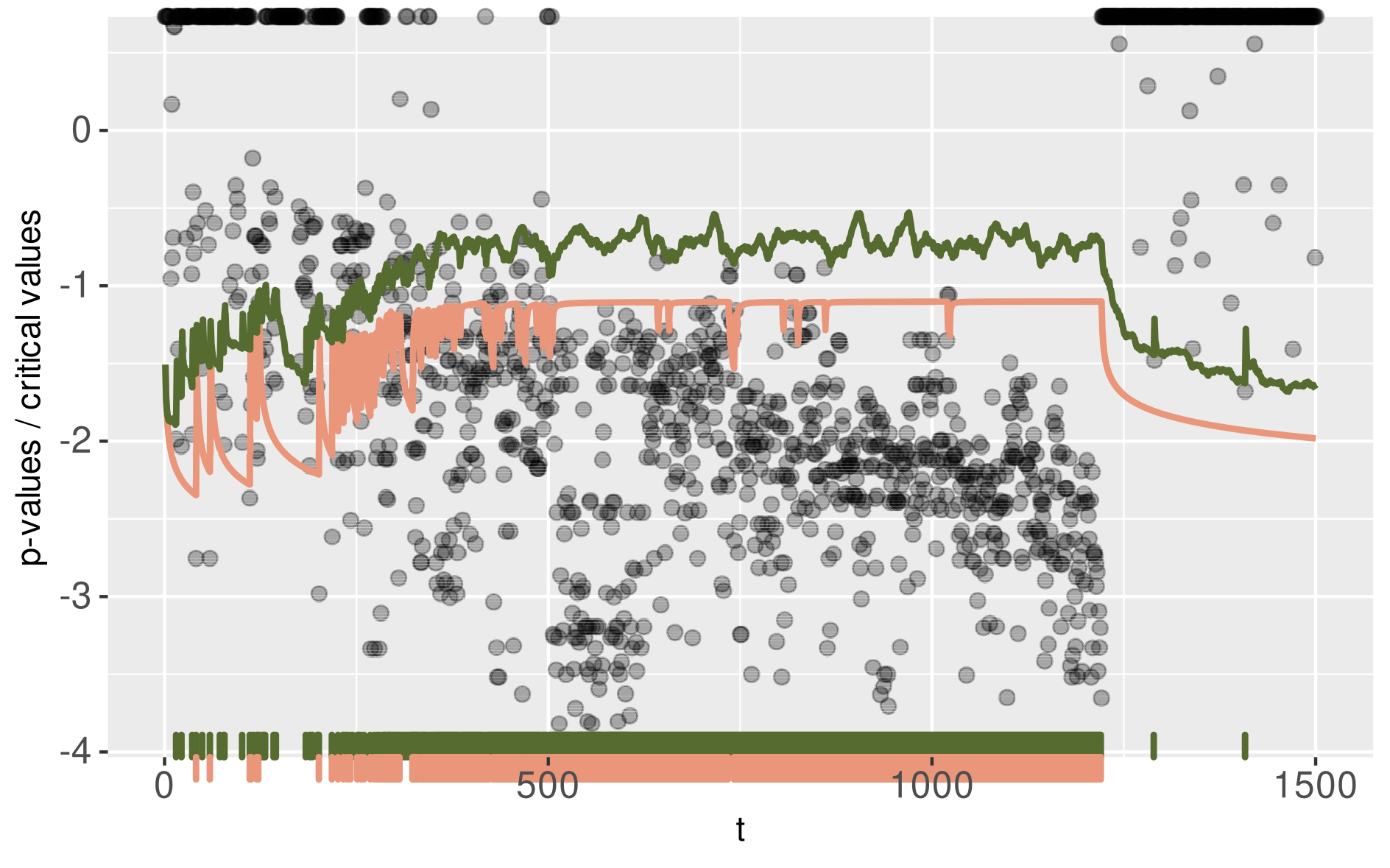}} & 
    \makebox{\includegraphics[width=.49\linewidth]{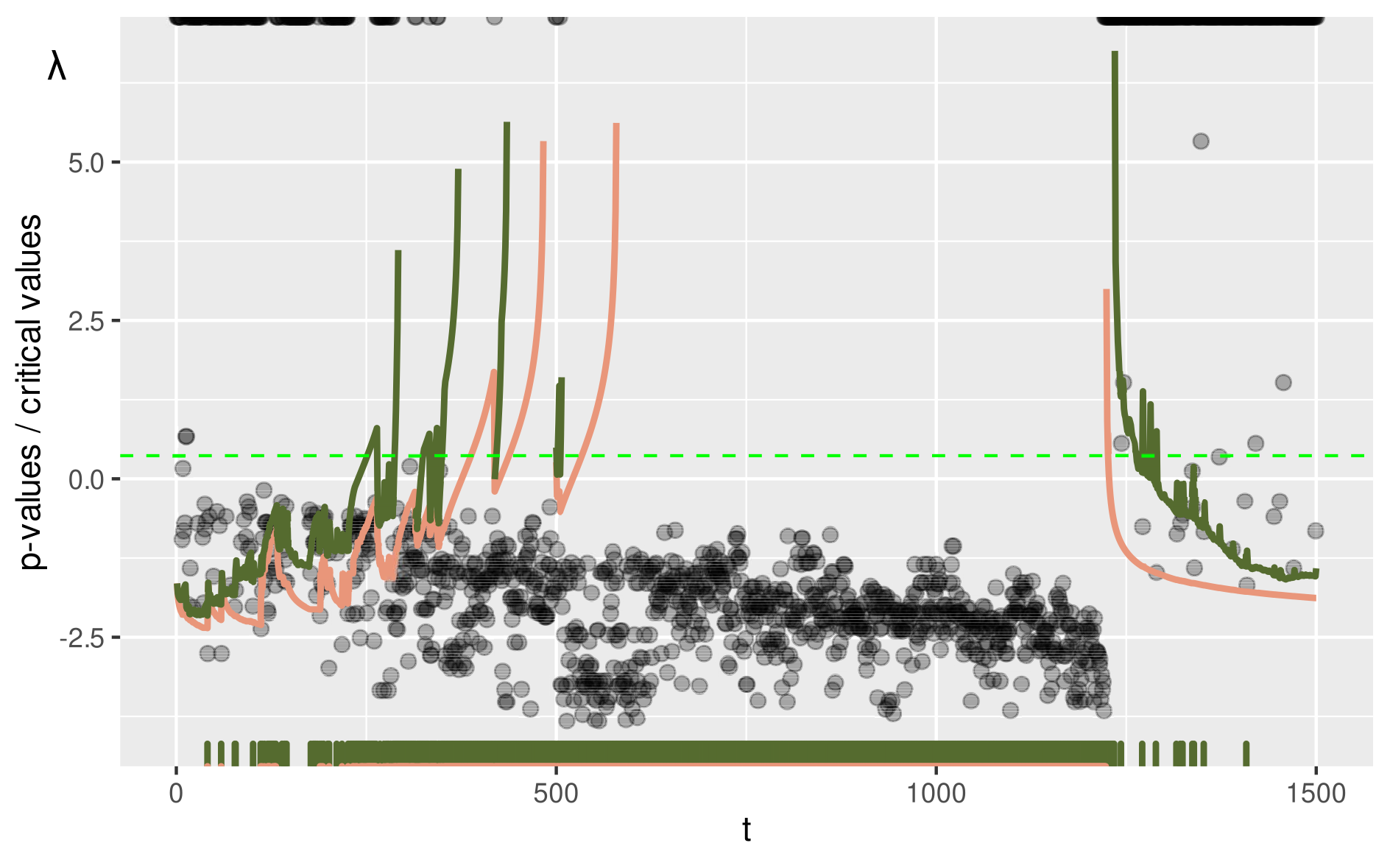}}
    \end{tabular}
        \vspace{-0.5cm}
    \caption{{Applying online mFDR controlling procedures to the male mice IMPC data set.} {Left panel:  $p$-values and critical values for LORD (orange curve) and $\rho$LORD (green curve). Right panel:  ALORD (orange curve) and $\rho$ALORD (green curve).} 
    Representation similar to Figure~\ref{fig:smoothedcvsmfdr} {($Y$-axis transformed by $y \mapsto -\log(-\log(y))$; $p$-values equal to $1$ are displayed at the top of the picture)}. }
    \label{fig:fwerimpcbasevsrewarded}
\end{figure}

\section{SUR procedures for weighted $p$-values}\label{sec:weighting}

In this section, we show how our SUR approach can be easily used to construct valid online $p$-value weighting procedures. 

\subsection{Setting and benchmark procedure}

Consider a standard continuous online multiple testing setting where each $p$-value is super-uniformly distributed under the null, that is, \eqref{eqn:superunif0} holds. 
Assume in addition that, at each time $t$, the $p$-value $p_t$ is associated with a quantity $r_t\geq 0$, called  the \emph{raw weight} (as opposed to  the \emph{rescaled weight} defined further on), which is assumed to be measurable w.r.t. $\mathcal{F}_{t-1}$.  
The magnitude of $r_t$ is interpreted as the level of belief  in a potential true discovery at time $t$: a large weight indicates a strong belief that the corresponding null hypothesis is false. 
Throughout the section, the weights $r_t$ are assumed to be available a priori and we will not discuss how to derive them (for this task, we refer to \cite{WR2006,RDV2006,RW2009,HZZ2010,ZZ2014,Ign2016,chen2020contextual} among others). 

While $p$-value weighting is a  classical tool for improving the performance of multiple testing methods in the offline setting (see references in Section~\ref{sec:superunif}), the incorporation of weights has received little attention in the online case. 
The only relevant work to our knowledge is \cite{ramdas2017online} (Section~5 therein),  which presents sufficient criteria for weighting procedures controlling the (m)FDR based on so-called GAI++ procedures and also discusses the technical challenges associated with weighted online multiple testing. 
An explicit algorithm which satisfies these criteria is used in \cite{ramdas2017online}\footnote{An implementation of this procedure can be found on the website \url{https://github.com/fanny-yang/OnlineFDRCode}}, which is detailed in Appendix~\ref{sec:ramdasweighting} for completeness. 
This method, which will be our benchmark procedure, works by weighting the $p$-values and adjusting for this weighting in the rejection reward. 

\subsection{New weighting approach}

The main idea of our new approach is as follows: consider weighted $p$-values  $\tilde{p}_t = p_t/w_t$ for some rescaled weight $w_t\in [0,1]$ which gives rise to the null bounding family $\mathcal{F} = \{F_t: u\in [0,1]\mapsto u w_t, t \geq 1\}$. Since the weights are constrained to take their values in $[0,1]$, the functions of $\mathcal{F}$ are super-uniform, that is, \eqref{eqn:superunif} holds. {Hence, one can apply our SUR approach with respect to that family $\mathcal{F}$.}

More specifically, our approach takes into account the null bounding family $\mathcal{F}$ in a simple two-step process, which proceeds as follows: for each time $t$,
\begin{enumerate}
	\item enforce super-uniformity by computing the \emph{rescaled weight}
	 $
	 w_t=\xi_t(r_t|r_1,\dots,r_{t-1}),
	 $
	 $t\geq 1$, 
	 for some given \emph{rescaling function}  $\xi_t$ valued in $[0,1]$ (see below for more details and an explicit choice);
	 \item apply any one of the SUR methods from Section~\ref{sec:FWERcontrolSU}  or Section~\ref{sec:mFDRcontrolSU}, depending on whether FWER or mFDR control is desired.
\end{enumerate}

We denote these new procedures by $w{X}$, where $X$ stands for the name of the base procedure (either OB \eqref{eqn:OB}, AOB \eqref{eqn:AOB}, LORD \eqref{eqn:LORD} or ALORD \eqref{eqn:ALORD}).
These procedures all come with the corresponding FWER or mFDR control (by additionally assuming \eqref{indep} if needed). {In particular, to the best of our knowledge, this also provides the first method for weighted online FWER control.}

{At first sight, these SUR weighting approaches may seem to be ineffective due to the conservatism induced by the rescaling step. 
However, this is countered in the second step by using SUR procedures that provide larger values $\alpha_t$, due to the super-uniform rewards accumulated in the past. The hope is that these two effects balance out in such a way as to favor rejection of hypotheses associated with larger values of (raw) weights.}

Finally, let us mention that a simple choice for $\xi_t$  is given by $\xi_t(x|r_1,\dots,r_{t-1})=\hat{F}_{t-1}(x)\ind{x> 0}$, where $\hat{F}_{t-1}(x)=(t-1)^{-1}\sum_{i=1}^{t-1} \ind{r_i\leq x}$ is the empirical c.d.f. of the sample $r_1,\dots,r_{t-1}$ (and by convention $\hat{F}_{0}(x)=1$).
{This particular choice is easy to compute in a sequential manner, and it satisfies the following intuitive and desirable properties: $\xi_t(x)\in [0,1]$ (ensures super-uniformity of $\mathcal{F}$), $\xi_t(x)$ is nondecreasing in $x$ (a larger raw weight leads to a larger rescaled weight), $\xi_t(0)=0$ (raw zero weights rescaled to zero), $\xi_t(\lambda r_t|\lambda r_1,\dots,\lambda r_{t-1})=\xi_t( r_t| r_1,\dots, r_{t-1})$ for all $\lambda>0$ (scale invariance) and if all raw weights are equal then all rescaled weights are equal to $1$.}

\subsection{Analysis of RNA-Seq data}

We revisit an analysis of the RNA-Seq data set `airway'  using results from the Independent Hypothesis Weighting (IHW) approach (for details, see \cite{Ign2016} and the vignette accompanying its software implementation).   While the original data was not collected in an online fashion, we  use it here nevertheless to provide  a proof of concept for weighted SUR procedures. The  `airway' data set contains data from $64102$ genes and the corresponding (offline) weights  are taken from the output of the ihw function from the bioconductor package `IHW'. {These `raw' weights are then transformed into rescaled weights by using the function $\xi_t$ described in the previous section.}
For the procedure parameters, we use the same choices as for the analysis of the IMPC data, see Section \ref{sec:real_data_appli}.

Table~\ref{weighting_results} (left part) presents the result for the FWER controlling procedures OB, AOB (non-weighted), and $w$OB, $w$AOB (SUR weighted approaches). It is clear that incorporating the weights leads to more rejections, which corroborates the fact that the weights coming from \cite{Ign2016} are indeed informative.

\begin{table}
\caption{\label{weighting_results} Number of discoveries for weighted controlling OMT procedures  for the 'airway' data set , with the weights taken from \cite{Ign2016}.}
	    \centering
	\begin{tabular}{|c||cccc||cccc|} 
\hline
Procedures & OB & $w$OB (new) & AOB & $w$AOB (new) & LORD & $w$GAI$_1$ & $w$GAI$_2$ &  $w$LORD (new) \\ [0.5ex] 
\hline
$\#$ discoveries  & 1092 & 1195  & 1188 & 1273  & 3550& 1308  & 3631 & 4445 \\
\hline
\end{tabular}

\end{table}

As for mFDR control, the (non-weighted) LORD is compared to our weighted version $w$LORD in Table~\ref{weighting_results} (right part). 
As additional competitors, we also added the weighted GAI++ procedure proposed in \cite{ramdas2017online} (see Section~\ref{sec:ramdasweighting} for a detailed description), that we use either with the raw weights (denoted by  $w$GAI$_1$) or with the rescaled weights (denoted by $w$GAI$_2$). As one can see, the effect of rescaling the weights is highly beneficial, and the new $w$LORD proposal is the one that incorporates these weights in the most efficient way.

\section{Discussion} \label{sec:conclu}

\subsection{Conclusion}

Existing OMT procedures often suffer from a lack of power due to conservativeness of the $p$-values. This occurs typically for discrete test statistics, which is a common situation in data sets where testing is based upon
counts.
To fill the gap, we introduced new SUR versions of some existing classical procedures, 
that 'reward' the base procedures by spending more efficiently 
the $\alpha$-wealth according to known bounds on the null cumulative distribution functions. 
We showed that our new SUR procedures provide rigorous control of online error criteria (FWER or mFDR) under classical assumptions 
while offering a systematic power enhancement. When using discrete Fisher exact test statistics, the improvement is substantial, both for simulated and real data. 

In addition, even in the standard case of uniformly distributed $p$-values, our approach allowed us to derive new weighted procedures that incorporate  external covariates. This provides improvements w.r.t. existing online weighting strategies.

\subsection{Another viewpoint}

{In the discrete setting, let us consider the following constrained spending problem:
at each step $t$, choose the critical value $\alpha_t$ to be in the support $S_t$ (including $0$) so that the following contraint holds
\begin{equation}\label{constraint}
	\sum_{t \geq 1}\alpha_t  \leq \alpha.
\end{equation}
It solves the super-uniformity problem, because $F_t(\alpha_t)=\alpha_t$ for all $t$, while it controls the online FWER.
This general principle, that we refer to as 'constrained spending strategies', can be implemented in many ways. }

Markedly, the {\it SUR approach is a way to achieve this}, by additionally following some reference critical values --- here the online Bonferroni critical values $ \alphabonf_t$    \eqref{eqn:OB}. Indeed, the rejection decision $p_t \leq \alphaOBSURE_t$ and $p_t \leq \alpha_t = F_t(\alphaOBSURE_t)$ are almost surely identical and we have calibrated $ \alphaOBSURE_t$ such that \eqref{constraint} holds, see \eqref{eqn:FWERcontrol3}. 
In other words, even if our critical values are not constrained to be in the support initially, the effective critical values $\alpha_t = F_t(\alphaOBSURE_t)$ that are actually used in the decision rule will automatically belong to the support. Thus, our approach can be equivalently seen as a way of implementing the constrained spending strategies delineated above.

Obviously, there are other ways to implement the constrained spending strategy. One instance is the delayed spending (DS) approach, that we describe in detail in Appendix~\ref{sec:delay}.


\subsection{Future directions}

While our results address several issues, they also raise new questions. 
First, the bandwidth of the kernel-based SUR spending sequence 
$\gamma'$ given by \eqref{eqn:kernel} has been chosen in a loose way here, but
tuning the bandwidth is certainly interesting from a power enhancement perspective (see Section~\ref{apenband}).
Also, in applications, the user would possibly like to select the bandwidth in a data dependent fashion 
without losing control over type I error rate. These two issues are interesting extensions for future developments.
Second, while our work focuses on marginal FDR, it would be desirable to build rewarded OMT procedures that control the (non-marginal) FDR. 
However, usual proofs rely on a monotonicity property of the critical value sequence (\citealp{ramdas2017online}) 
that is difficult to satisfy here, because the super-uniformity reward naturally varies over time. 
Hence, deriving rewarded FDR controlling procedures is a challenging issue that is left for future investigations. 
Third, most of our results rely on an independence assumption, see \eqref{indep}. 
While this can be considered as a mild restriction in an online framework, 
relaxing it or incorporating a known dependence structure in OMT is an interesting avenue.

\section*{Acknowledgements}
This work has been supported by ANR-16-CE40-0019 (SansSouci), ANR-17-CE40-0001 (BASICS) and by the GDR ISIS through 
the 'projets exploratoires' program (project TASTY). 
It is part of project DO 2463/1-1, funded by the Deutsche Forschungsgemeinschaft. 
The authors thank Florian Junge for his help regarding technical issues when running the simulations, Natasha Karp for explanations on the IMPC data and Aaditya Ramdas for very constructive discussions.

\bibliographystyle{apalike}
\bibliography{biblio}
\appendix
\section{Proofs}\label{sec:proofs}
\subsection{Proofs for online FWER control}\label{sec:proofFWER}
We start by proving Theorem~\ref{th:genFWER} and then deduce Theorems~\ref{th:OBSURE}~and~\ref{th:AOBSURE}.
\begin{proof}[Proof of Theorem \ref{th:genFWER}]
	First, let us show that for any critical values $(\alpha_t,t\geq 1)$, a sufficient condition for FWER control under \eqref{eqn:superunif} is given by 
			\begin{align}
			\alpha_T + \sum_{t=1}^{T-1} \ind{p_t(X)\geq  \lambda }  F_t(\alpha_t) 
			\leq (1-\lambda) \alpha \qquad \text{(a.s.)}\label{sufficientFWERcontrol}
			\end{align}
	if either \eqref{indep} or if $(\alpha_t,t\geq 1)$ are deterministic for all $T \geq 1$. This comes from 
	Markov's inequality combined with Lemma~\ref{lemmacontroladapt}:
\begin{align*}
    \FWER(T,\mathcal{A},P)
    &\leq \E_{X\sim P}\Big(\sum_{t=1}^T \ind{t\in \cH_0(P), p_t \leq \alpha_t}\Big)\\
    &\leq  (1-\lambda)^{-1} \: \E\left(\sum_{t=1}^T \ind{p_t(X) \geq \lambda } F_t(\alpha_t)\right)\\
    &\leq (1-\lambda)^{-1} \: \E\left(\alpha_T+ \sum_{t=1}^{T-1} \ind{p_t(X) \geq \lambda } F_t(\alpha_t)\right),
\end{align*}
which gives the announced sufficient condition.
Now, we obtain statement (i) of the theorem by verifying the above criterion  \eqref{sufficientFWERcontrol} for $\alpha^0_t$ using the (crude)  bound $F_t(x) \le x$ and assumption \eqref{eqn:conditionalpha0}.
Next, we obtain statement (ii) of the theorem by verifying the above criterion  \eqref{sufficientFWERcontrol} for $\alpha_t$. This is done by reducing this to a statement on $\alpha^0_t$ via Lemma~\ref{lem:gencomput}. More precisely, {with $a_T = \sum_{t=1}^T \gamma'_t,$ we have}
\begin{align*}
    \alpha_T + \sum_{t=1}^{T-1} \ind{p_t(X) \geq  \lambda }  F_t(\alpha_t)
    &\leq \alpha_T + \sum_{t=1}^{T-1} \ind{p_t \geq \lambda} \left[(1 - a_{T-t}) \alpha_t + a_{T-t}F_{t}(\alpha_t) \right]\\
    &= \alpha^0_T + \sum_{t=1}^{T-1} \ind{p_t \geq \lambda} \alpha^0_t \leq \alpha,
\end{align*}
where the equality above is true provided that the following recursion holds for all $T\geq 1$,
\begin{align*}
    \alpha_T = 
    \alpha^0_T 
    + \sum_{t=1}^{T-1} \ind{p_t \geq \lambda} \alpha^0_t
    - \sum_{t=1}^{T-1} \ind{p_t \geq \lambda} \left[(1 - a_{T-t}) \alpha_t 
    + a_{T-t}F_{t}(\alpha_t) \right].
\end{align*}
This is true by Lemma~\ref{lem:gencomput} because of the expression \eqref{eqn:generalreward} of $\alpha_t$. This concludes the proof.
\end{proof}
\begin{proof}[Proof of Theorems~\ref{th:OBSURE}~and~\ref{th:AOBSURE}]
Theorems~\ref{th:OBSURE}~and~\ref{th:AOBSURE} are corollaries of Theorem \ref{th:genFWER}, by considering 
$\mathcal{A}^0=\Abonf$ ($\lambda=0$) and $\mathcal{A}^0=\AAbonf$, respectively. 
Indeed, checking \eqref{eqn:conditionalpha0} is straightforward for $\Abonf$ from the spending sequence definition or comes from Lemma~\ref{lem:Tronde} for $\AAbonf$.
\end{proof}

\subsection{Proofs for online mFDR control}\label{sec:proofmFDR}\label{proof:ALORD}

The global proof strategy is similar to the one used for FWER: we start by proving Theorem~\ref{th:genmFDR} and then deduce Theorem~\ref{th:LORDSURE} and Theorem~\ref{th:ALORDSURE}. 

\begin{proof}[Proof of Theorem~\ref{th:genmFDR}]
    First, we establish that mFDR control is provided under \eqref{eqn:superunif} and  \eqref{indep} for any procedure $\mathcal{A}=(\alpha_t,t\geq 1)$ if
    \begin{align}
        \alpha_T + \sum_{1\leq t\leq T-1, \atop p_t\geq\lambda} F_t(\alpha_t) \leq (1-\lambda) \alpha\:( 1\vee R(T)),\quad \mbox{(a.s.)}.
        \label{conditionproofmFDR}
    \end{align}
    Indeed, by Lemma~\ref{lemmacontroladapt}, we have
    \begin{align*}
        \E_{X\sim P}\Big(\sum_{t=1}^T \ind{t \in \cH_0, p_t \leq \alpha_t}\Big) 
        &\leq (1-\lambda)^{-1} \: \E\left(\sum_{t=1}^T \ind{p_t(X)\geq \lambda } F_t( \alpha_t) \right)\\
        &\leq \alpha\:\E( 1\vee R(T)),
    \end{align*}
    by using \eqref{conditionproofmFDR}, which is exactly the desired mFDR control.
    Now, statement (i) holds because \eqref{conditionproofmFDR} holds for $(\alpha^0_t,t\geq 1)$ from \eqref{conditionalpha0mFDR} and \eqref{eqn:superunif}. 
    Finally, we establish statement (ii).  By \eqref{conditionalpha0mFDR} and \eqref{eqn:superunif}, condition \eqref{conditionproofmFDR} holds for $(\alpha_t,t\geq 1)$ if for all $T\geq 1$,
    \begin{align*}
        \alpha_T + \sum_{t=1}^{T-1} \ind{p_t(X) \geq \lambda} \left[(1 - a_{T-t}) \alpha_t 
        + a_{T-t}F_{t}(\alpha_t) \right]
        = \alpha^0_T + \sum_{p_t \geq\lambda, 1 \leq t \leq T-1} \alpha^0_t,
    \end{align*}
   {where $a_T = \sum_{t=1}^T \gamma'_t$. Now the last display holds true}    by Lemma~\ref{lem:gencomput} because of \eqref{eqn:generalreward}, which concludes the proof.    
\end{proof}

\begin{proof}[Proof of Theorems~\ref{th:LORDSURE}~and~\ref{th:ALORDSURE}]
Theorem~\ref{th:LORDSURE} and Theorem~\ref{th:ALORDSURE} can be derived from Theorem~\ref{th:genmFDR} for  $\mathcal{A}^0=\ALORD$ (using $\lambda=0$) and $\mathcal{A}^0=\AALORD$, respectively, by checking \eqref{conditionalpha0mFDR} in both cases.
First, for $\ALORD$, we have
\begin{align}
    \sum_{t=1}^T \alphaLORD_t&=\sum_{t=1}^T \left(W_0 \gamma_t + (\alpha-W_0) \gamma_{t-\tau_1} + \alpha\sum_{j\geq 2} \gamma_{t-\tau_j}\right) \nonumber\\
    &= W_0 \sum_{t=1}^T  \gamma_t +(\alpha-W_0) \sum_{t=1}^T\gamma_{t-\tau_1} + \alpha\sum_{j\geq 2}\ind{T-\tau_j\geq 1}\sum_{t=1}^T \gamma_{t-\tau_j}\nonumber\\
    &\leq \alpha (1+0\vee (R(T-1)-1))\leq \alpha (1\vee R(T)),\label{equreasoningmFDR}
\end{align}
because $\tau_j \leq T-1$ is equivalent to $R(T-1) \geq j$ by definition.
Second, for $\AALORD$, we proceed similarly with the help of Lemma~\ref{lem:Tronde}:
        by definition \eqref{eqn:ALORD}, we have  
\begin{align*}
    &(1-\lambda)^{-1}\left( \alphaALORD_T + \sum_{1 \leq t \leq T-1, \atop p_t \geq \lambda} \alphaALORD_t\right)\\
    &= W_0 \left(\gamma_{\mathcal{T}_0(T)} 
    + \sum_{1 \leq t \leq T-1, \atop p_t \geq \lambda}\gamma_{\mathcal{T}_0(t)}\right) 
    + (\alpha-W_0) \left(\gamma_{\mathcal{T}_1(T)} 
    + \sum_{1 \leq t \leq T-1, \atop p_t \geq \lambda}\gamma_{\mathcal{T}_1(t)} \right)\\
    &\:\:\: + \alpha \sum_{j \geq 2} \ind{T \geq \tau_j+1} \left(\gamma_{\mathcal{T}_j(T)} 
    + \sum_{1 \leq t \leq T-1, \atop p_t \geq \lambda}\gamma_{\mathcal{T}_j(t)}\right).
\end{align*}
Finally, by using  \eqref{eqn:interm} and \eqref{eqn:intermj}, the latter is equal to  
\begin{align*}
    &W_0 \sum_{t=1}^{\mathcal{T}_0(T)} \gamma_{t} 
    + (\alpha-W_0) \sum_{t=1}^{\mathcal{T}_1(T)} \gamma_{t} 
    + \alpha \sum_{j\geq 2} \ind{T \geq \tau_j+1} \sum_{t=1}^{\mathcal{T}_j(T)} \gamma_{t}\\
    &\leq W_0 + \alpha-W_0 + \alpha  \sum_{j \geq 2} \ind{T \geq \tau_j+1} 
    = \alpha (1 + 0\vee (R(T-1)-1)) \leq \alpha\:( 1 \vee R(T)),
\end{align*}
because $T \geq \tau_j+1$ if and only if $R(T-1) \geq j$. 
\end{proof}

\subsection{Auxiliary lemmas}

The following lemma provides a tool for controlling both online $\FWER$ and $\mFDR$.

\begin{lemma}
    \label{lemmacontroladapt}
    For any procedure $\mathcal{A} = (\alpha_t, t \geq 1)$, we have for all $\lambda\in [0,1)$,
    \begin{align}
        \E_{X\sim P}\Big(\sum_{t=1}^T \ind{t \in \cH_0, p_t \leq \alpha_t}\Big) 
        \leq (1-\lambda)^{-1} \: \E\left(\sum_{t=1}^T \ind{p_t(X)\geq \lambda } F_t( \alpha_t) \right),
        \label{conditioncontroladapt}
    \end{align}
    provided that \eqref{eqn:superunif} holds and  if either \eqref{indep} holds or if the critical values $(\alpha_t, t \geq 1)$ are deterministic.
\end{lemma}

\begin{proof}
Recall $\alpha_t$ is either deterministic or $\mathcal{F}_{t-1}$-measurable (in which case it is independent of $p_t(X)$ under \eqref{indep}). Therefore, under the conditions of the lemma, we have in any case: for all $t \in \cH_0$, both
\begin{align*}
    \E\left(\frac{\ind{p_t(X) > \lambda}}{1-\lambda} \bigg| \alpha_t\right) \geq 1 , \:\:\: \P\left(p_t(X) \leq \alpha_t \:|\: \alpha_t\right)\leq F_t(\alpha_t).
\end{align*}
This entails
\begin{align*}
    \E_{X \sim P}\Big(\sum_{t=1}^T \ind{t \in \cH_0, p_t \leq \alpha_t}\Big)
    &= \sum_{t=1}^T \ind{t \in \cH_0} \:\E \left(\P(p_t(X) \leq \alpha_t \:|\: \alpha_t)\right) \\
    &\leq \sum_{t=1}^T \ind{t \in \cH_0} \:\E \left(F_t(\alpha_t) \right) \\
    &\leq \sum_{t=1}^T \ind{t \in \cH_0}\: \E \left(F_t(\alpha_t)  \E \left(\frac{\ind{p_t(X) \geq \lambda}}{1-\lambda} \:|\: \alpha_t\right)\right)\\
    &\leq (1-\lambda)^{-1} \:\E\left(\sum_{t=1}^T \ind{p_t(X) \geq \lambda} F_t(\alpha_t) \right).
\end{align*}
\end{proof}
The following representation lemma is the key tool for building the new rewarded critical values.

\begin{lemma}
    \label{lem:gencomput}
    Let $(\alpha^0_t, t \geq 1)$ be any nonnegative sequence. 
    Let $(\tilde{\alpha}_t, t \geq 1)$ be the sequence defined by the recursive relation
    \begin{equation}
        \label{alphatildegen}
        \tilde{\alpha}_T = \alpha^0_T 
        + \sum_{t=1}^{T-1} \ind{p_t \geq \lambda} \alpha^0_t 
        - \sum_{t=1}^{T-1} \ind{p_t \geq \lambda} \left[ (1 - a_{T-t}) \tilde{\alpha}_t 
        + a_{T-t}F_{t}(\tilde{\alpha}_t) \right],\:\: \quad T\geq 1,
    \end{equation}
    where  $a_T = \sum_{t=1}^T \gamma'_t,$ $T \geq 1$ for any real values $\gamma'_t$, $p_t$, $\lambda$ and functions $F_t$.
    Let $(\bar{\alpha}_t, t \geq 1)$ be the sequence defined by the recursive relation
    \begin{align*}
        \bar{\alpha}_T = \alpha^0_T
            + \sum_{1\leq t \leq T-1 \atop p_t \geq \lambda} \gamma'_{T-t} (\bar{\alpha}_t - F_t(\bar{\alpha}_t)) 
            +  \ind{ p_{T-1} < \lambda }(\bar{\alpha}_{T-1} - \alpha^0_{T-1}) , \quad T\geq 1.
    \end{align*}
Then we have $\tilde{\alpha}_t = \bar{\alpha}_t$ for all $t \geq 1$. {Moreover, $\bar{\alpha}_t \geq \bar{\alpha}^0_t$ for all $t \geq 1$ under \eqref{eqn:superunif}. In particular, these critical values are nonnegative.}
\end{lemma}

\begin{proof}
Clearly, $\tilde{\alpha}_1 = \alpha^0_1 = \bar{\alpha}_1$  so the result is satisfied for $T=1$.
For $T\geq 2$, by using \eqref{alphatildegen} for $\tilde{\alpha}_T $ and $\tilde{\alpha}_{T-1}$, we have 
\begin{align*} 
    \tilde{\alpha}_T -\tilde{\alpha}_{T-1} =&\:\alpha^0_T-\alpha^0_{T-1}+\ind{p_{T-1}\geq \lambda} \alpha^0_{T-1}\\
    &-\ind{p_{T-1}\geq\lambda} \left[ (1 - a_1)  \tilde{\alpha}_{T-1} + a_1F_{T-1}(\tilde{\alpha}_{T-1}) \right]\\
    &+ \sum_{t =1}^{ T-2}  \ind{p_t\geq\lambda}\left[ (a_{T-t} - a_{T-t-1})  \tilde{\alpha}_t -(a_{T-t}- a_{T-t-1})F_{t}(\tilde{\alpha}_t) \right].
\end{align*}
Hence, by using 
$\tilde{\alpha}_{T-1} = \tilde{\alpha}_{T-1} \ind{p_{T-1} < \lambda} +\tilde{\alpha}_{T-1} \ind{p_{T-1} \geq \lambda}$, we obtain
\begin{align*} 
    \tilde{\alpha}_T  =&\:\alpha^0_{T}- \ind{p_{T-1}< \lambda} \alpha^0_{T-1}+\tilde{\alpha}_{T-1}\ind{p_{T-1}< \lambda} \\
    &+\ind{p_{T-1}\geq\lambda} \left[   \gamma'_1  \tilde{\alpha}_{T-1} - \gamma'_1F_{T-1}(\tilde{\alpha}_{T-1}) \right]\\
    &+ \sum_{t =1}^{ T-2}  \ind{p_t\geq\lambda}\left[ \gamma'_{T-t}  \tilde{\alpha}_t - \gamma'_{T-t}F_{t}(\tilde{\alpha}_t) \right], 
    \end{align*}
because $\gamma'_1=a_1$, and we recognize the expression given in the lemma.

{Let us finally prove that $\bar{\alpha}_T\geq \bar{\alpha}^0_T$ for all $T\geq 1$. This is true for $\bar{\alpha}_1$ because $\bar{\alpha}_1=\alpha^0_1$. Now, if $\bar{\alpha}_{1}\geq \alpha^0_1,\dots,\bar{\alpha}_{T-1}\geq\alpha^0_{T-1} $ then we also have 
$$
 \bar{\alpha}_T = \alpha^0_T
            + \sum_{1\leq t \leq T-1 \atop p_t \geq \lambda} \gamma'_{T-t} (\bar{\alpha}_t - F_t(\bar{\alpha}_t)) 
            +  \ind{ p_{T-1} < \lambda }(\bar{\alpha}_{T-1} - \alpha^0_{T-1})\geq \alpha^0_T,
$$
because $\bar{\alpha}_t \geq  F_t(\bar{\alpha}_t)$ by \eqref{eqn:superunif}. This finishes the proof.
}
\end{proof}

We now establish a result for the functionals $\mathcal{T}(\cdot)$ and $\mathcal{T}_j(\cdot)$, $j\geq 1$, 
which are used by the adaptive procedures $\AAbonf$ and $\AALORD$, respectively.

\begin{lemma}\label{lem:Tronde}
Consider the functional $\mathcal{T}(\cdot)$ defined by \eqref{Tronde_def} for some realization of the $p$-values and some $\lambda\in [0,1)$. Then for any sequence $(\gamma_t)_{t\geq 1}$ and for any $T\geq 1$, we have
\begin{align}
    \label{eqn:interm}
    \sum_{t=1}^T \ind{p_t \geq \lambda}\gamma_{\mathcal{T}(t)} = \sum_{t=1}^{\mathcal{T}(T+1)-1}\gamma_t.
\end{align}
In addition, for any $j \geq 1$, consider the $\tau_j$ defined by \eqref{eqn:tauj} and the functional $\mathcal{T}_j(\cdot)$ defined by \eqref{Tronde_def_gen}. Then for all $T \geq \tau_j+1$,
\begin{equation}
    \label{eqn:intermj}
    \sum_{1 \leq t \leq T \atop p_t \geq \lambda} \gamma_{\mathcal{T}_j(t)} = \sum_{t=1}^{\mathcal{T}_j(T+1)-1} \gamma_{t}. 
\end{equation}
\end{lemma}

\begin{proof}
Let us first prove \eqref{eqn:interm}. Since $\mathcal{T}(t+1) = \mathcal{T}(t)+1$ when $p_t \geq \lambda$ from definition \eqref{Tronde_def}, 
we can write
\begin{align*}
    \sum_{t=1}^T \ind{p_t\geq\lambda}\gamma_{\mathcal{T}(t)}
    = \sum_{t=1}^T \ind{p_t\geq\lambda}\gamma_{\mathcal{T}(t+1)-1} 
    = \sum_{t=2}^{T+1} \ind{p_{t-1}\geq\lambda}\gamma_{\mathcal{T}(t)-1}.
\end{align*}
Additionally, it is clear that $\mathcal{T}(\cdot)$ is a bijection 
mapping $\lbrace 1, \, 2 \leq t \leq T+1 \:: p_{t-1} \geq \lambda \rbrace$ 
into $\lbrace 1, 2, \dots, \mathcal{T}(T+1) \rbrace$. 
Hence, the latter sum can be rewritten as 
$\sum_{t=2}^{\mathcal{T}(T+1)} \gamma_{t-1}=\sum_{t=1}^{\mathcal{T}(T+1)-1} \gamma_{t}$ which provides \eqref{eqn:interm}.

Second, for proving \eqref{eqn:intermj}, 
the crucial point is that according to the definition of $\mathcal{T}_j(T)$ \eqref{Tronde_def_gen}, 
the functional $\mathcal{T}_j: \{\tau_j+1, \dots\} \to \{1,\dots\}$ 
is a bijection from $\{\tau_j+1\}\cup \{t \in \{\tau_j+2, \dots,T+1\} \::\: p_{t-1} \geq \lambda\}$ to $\{1, \dots,\mathcal{T}_j(T+1)\}$,
for any $j \geq 1$ and $T \geq \tau_j+1$.
In particular, this entails
\begin{align*}
    \sum_{p_t \geq \lambda, 1 \leq t \leq T} \gamma_{\mathcal{T}_j(t)} 
    &= \sum_{p_t \geq \lambda, \tau_j+1 \leq t \leq T} \gamma_{\mathcal{T}_j(t)} 
    = \sum_{p_t \geq \lambda, \tau_j+1 \leq t \leq T} \gamma_{\mathcal{T}_j(t+1)-1} \\
    &= \sum_{p_{t-1} \geq \lambda, \tau_j+2 \leq t \leq T+1} \gamma_{\mathcal{T}_j(t)-1} 
    = \sum_{t=2}^{\mathcal{T}_j(T+1)} \gamma_{t-1} = \sum_{t=1}^{\mathcal{T}_j(T+1)-1} \gamma_{t} .
\end{align*}
This proves \eqref{eqn:intermj}.
\end{proof}

\section{Delayed spending approach}\label{sec:delay}

In this section we present another way of incorporating super-uniformity into OMT {which we refer to as \emph{delayed spending}} (in the sequel abbreviated as {DS}). We are grateful to Aaditya Ramdas for this suggestion. 

The new procedure is introduced in Section~\ref{sec:mainideadelaying}, while 
 	we highlight some mathematical and practical differences with our approach 
 	in Sections~\ref{sec:comparisondelaying} and~\ref{sec:comparSURE}. 
	In order to make the new procedure more efficient we also present a hybrid version in Section~\ref{sec:hybrid}. For simplicity, we restrict ourselves to  FWER controlling procedures {for discrete data} throughout this section.

 \subsection{Definition}\label{sec:mainideadelaying}

Let us start with the critical value $\alpha_1=\alpha \gamma_1$. While the OB procedure would choose $\alpha_2=\alpha \gamma_2$, the idea is that if the super-uniformity is strong enough to ensure $F_1(\alpha\gamma_1)+F_2(\alpha\gamma_1)\leq \alpha \gamma_1$, we can still use $\alpha_2=\alpha \gamma_1$ in the second round. This process can be continued until $F_1(\alpha\gamma_1)+\dots+F_{b_1+1}(\alpha\gamma_1)> \alpha \gamma_1$, in which case we switch to $\alpha_{b_1}=\alpha \gamma_2$, and so on. This way, we can incorporate the super-uniformity directly by 'delaying' the $\gamma$ sequence.

More formally, consider the setting of Section~\ref{sec:setting}, 
  where a null bounding family $\mathcal{F} = \{F_t, t \geq 1\}$ satisfying \eqref{eqn:superunif} is at hand. 
 The above strategy reads: 
\begin{align}
	\alphaDelay_t &= \alpha \gamma_{\delay(t)}, \mbox{ where } \delay(t)= \min\{j\geq 1\::\: b_j\geq t\},\:\:\:t\geq 1;\label{equdelay}\\
	b_j &= \max\Big\{T \geq b_{j-1}+1\::\: \sum_{t= b_{j-1}+1}^{T} F_t( \alpha \gamma_j)\leq \alpha \gamma_j\Big\},\:\:\:j\geq 1,\label{equbj}
\end{align}
(with the convention $b_0=0$ and $b_j=+\infty$ if the set in \eqref{equbj} is empty), so that $j=\delay(t)$ for $b_{j-1}+1\leq t\leq b_j$.
Thus, the {DS} method processes each sub-budget $\alpha\gamma_j$ one at a time, until the stopping rule in \eqref{equbj} is met and the transition to the next sub-budget $\alpha\gamma_{j+1}$ is made. {Since $\delay(t) \le t$ we can interpret $\alphaDelay_t=\alpha \gamma_{\delay(t)}$ as a  'slowed-down' variant of the original OB procedure. }

The procedure \eqref{equdelay} controls the online FWER under \eqref{eqn:superunif} because by \eqref{eqn:FWERcontrol3}, a sufficient condition is given by $\sum_{t=1}^T F_t(\alpha_t)\leq \alpha$, $T\geq 1$, and we indeed have 
$$
\sum_{t\geq 1} F_t(\alphaDelay_t)  \leq  \sum_{j\geq 1} \sum_{t= b_{j-1}+1}^{b_j}  F_t(\alpha \gamma_j) \leq \sum_{j\geq 1} \alpha \gamma_j\leq \alpha,
$$
by definition of the $b_j$'s.  
Note that the $b_j$'s are based on local averages (in time) of the $F_t(x)$'s at certain points $x$. This shares similarity to the approach of \cite{WestWolf1997} for offline FWER control.

\subsection{Comparison to SUR for real data}\label{sec:comparisondelaying}

{Both the {DS} and the SUR approaches use super-uniform rewarding. In a nutshell, the {DS} approach slows down the clock whereas the SUR approach augments the critical values of existing OMT procedures in an additive way.}
While a more detailed comparison can be found in the following Section~\ref{sec:comparSURE}, we may say that no method dominates the other one uniformly. 
The examples given in Section~\ref{sec:comparSURE} (delayed start, long/infinite delay, ineffective delay) suggest that the {DS} method could be more efficient at the very start of the stream but may suffer from conservativeness afterwards.

To assess the behaviour of the procedures in a practical setting, we reanalyse the IMPC data from Section~\ref{sec:real_data_appli} using the {DS} procedure defined by \eqref{equdelay} and \eqref{equbj} and compare it with the OB and $\rho$OB from Section~\ref{sec:FWERcontrolSU}. The results for FWER control at level $\alpha=0.2$ are displayed in Table~\ref{tabmicedelayed} and Figure~\ref{fig:compardelayed}. As Figure~\ref{fig:compardelayed} (right panel) shows,  the rejection process $\{R(T), T\geq 1\}$, is almost identical at the very start. However, for larger $T$, 
the delayed approach makes less discoveries than the $\rho$OB procedure and this, uniformly in time for this data set.  
This conservative behaviour is probably caused by under-utilization of wealth as described in Section~\ref{sec:comparSURE}. More specifically, the non-utilized component of $\alpha=0.2$ accumulates up to time $T=1500$ approximately to $0.077$, so that approximately 
$ 38.5\%$ of $\alpha=0.2$ are effectively neglected. Accordingly, the wealth plot displayed in Figure~\ref{fig:compardelayed} 
 shows that the delayed approach manages to spend more wealth than the OB procedure, but still deviates strongly from the nominal wealth curve. 
 {Figure~\ref{fig:compardelayedcriticalvalues} illustrates the same phenomenon for the critical values. (This replaces the old section D.2)}
 \begin{table}
\caption{\label{tabmicedelayed} Number of discoveries
    for SURE online Bonferroni \eqref{eqn:OBSURE} (bandwidth $h=10$) and the {DS} approach \eqref{equdelay}. Here $\delay(30\,000)=5083$ as defined in \eqref{equdelay}.
    These numbers are obtained by running the procedures on the first $30\,000$ genes for male (second row) and female (third row) mice in the IMPC data.}
\centering
\fbox{%
\begin{tabular}{*{4}{c}} 
 \hline
 Procedures & OB & $\rho$OB & Delayed  \\ 
 \hline
 $\#$ discoveries (male) & 229 & 377 & 293 \\
 $\#$ discoveries (female) & 267 & 481 & 355 \\
 \hline
\end{tabular}}
\end{table}

\begin{figure}[h!]
    \centering
    \begin{tabular}{cc}
    \makebox{\includegraphics[width=.49\linewidth]{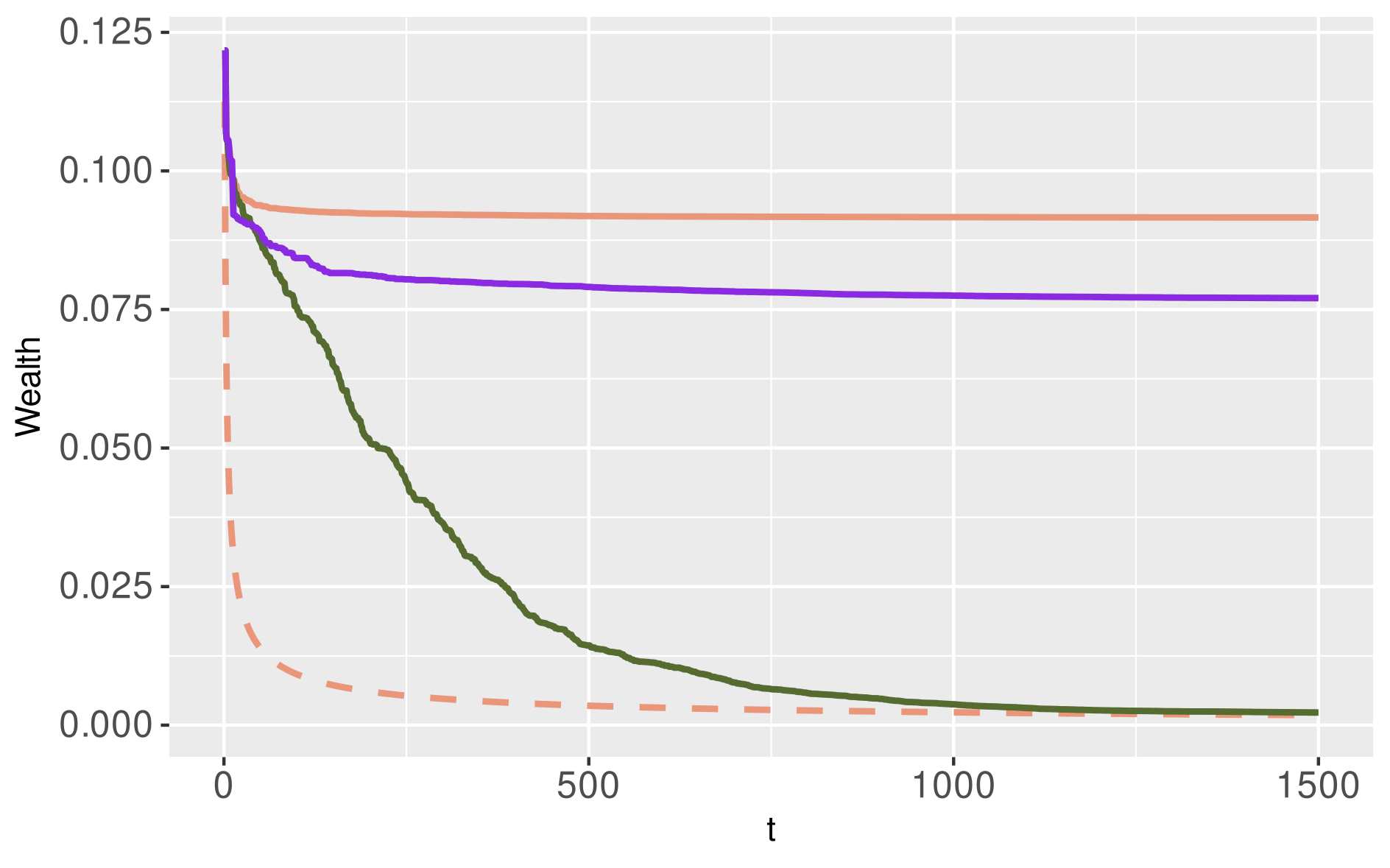}} & 
    \makebox{\includegraphics[width=.49\linewidth]{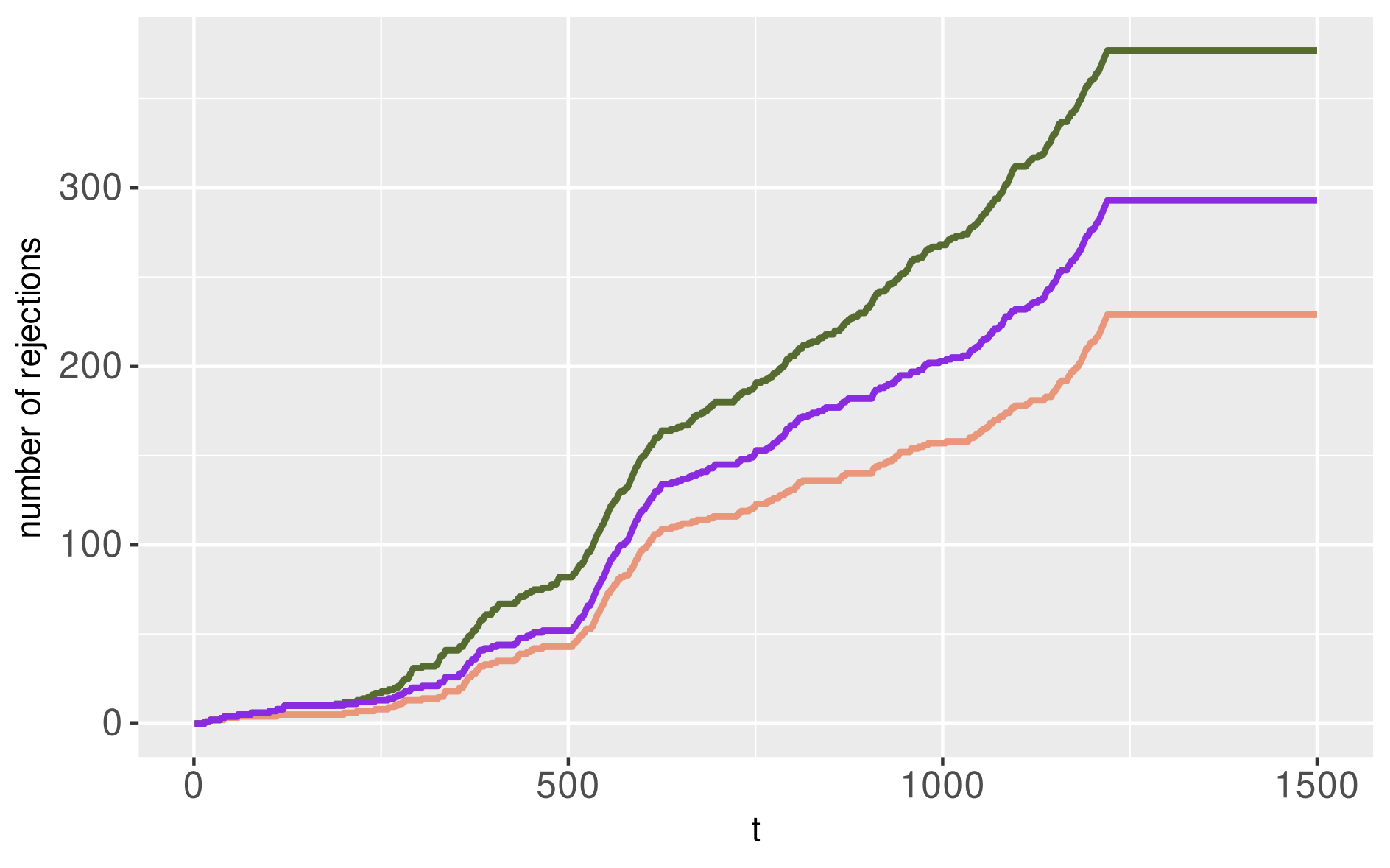} }
    \end{tabular}
        \vspace{-0.5cm}
    \caption{Comparison with {DS}. Left: nominal wealth for OB (dashed orange curve), effective wealth for OB (solid orange curve), effective wealth for $\rho$OB (solid green curve) and effective wealth for {DS} (solid purple curve), plot similar to  Figure~\ref{fig:nom_vs_eff_wealth_fwerproc}. Right: rejection numbers, cumulated over time, for the same procedures (same color code).   
     Both plots are computed from the male IMPC data.  \label{fig:compardelayed}
}
  \end{figure}

\begin{figure}[h!]
	\centering
	\begin{tabular}{cc}
		\makebox{\includegraphics[width=.49\linewidth]{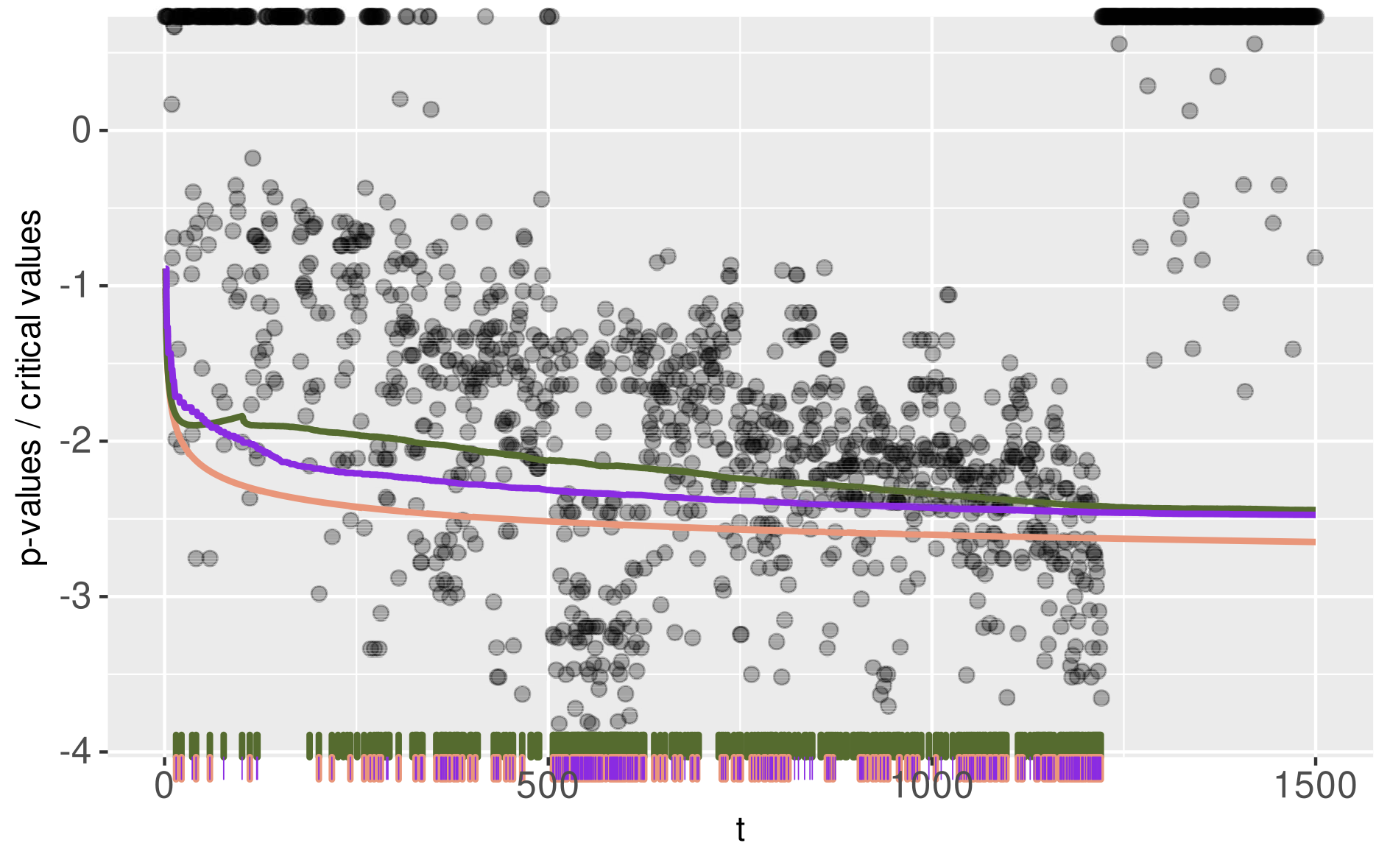}} & 
		\makebox{\includegraphics[width=.49\linewidth]{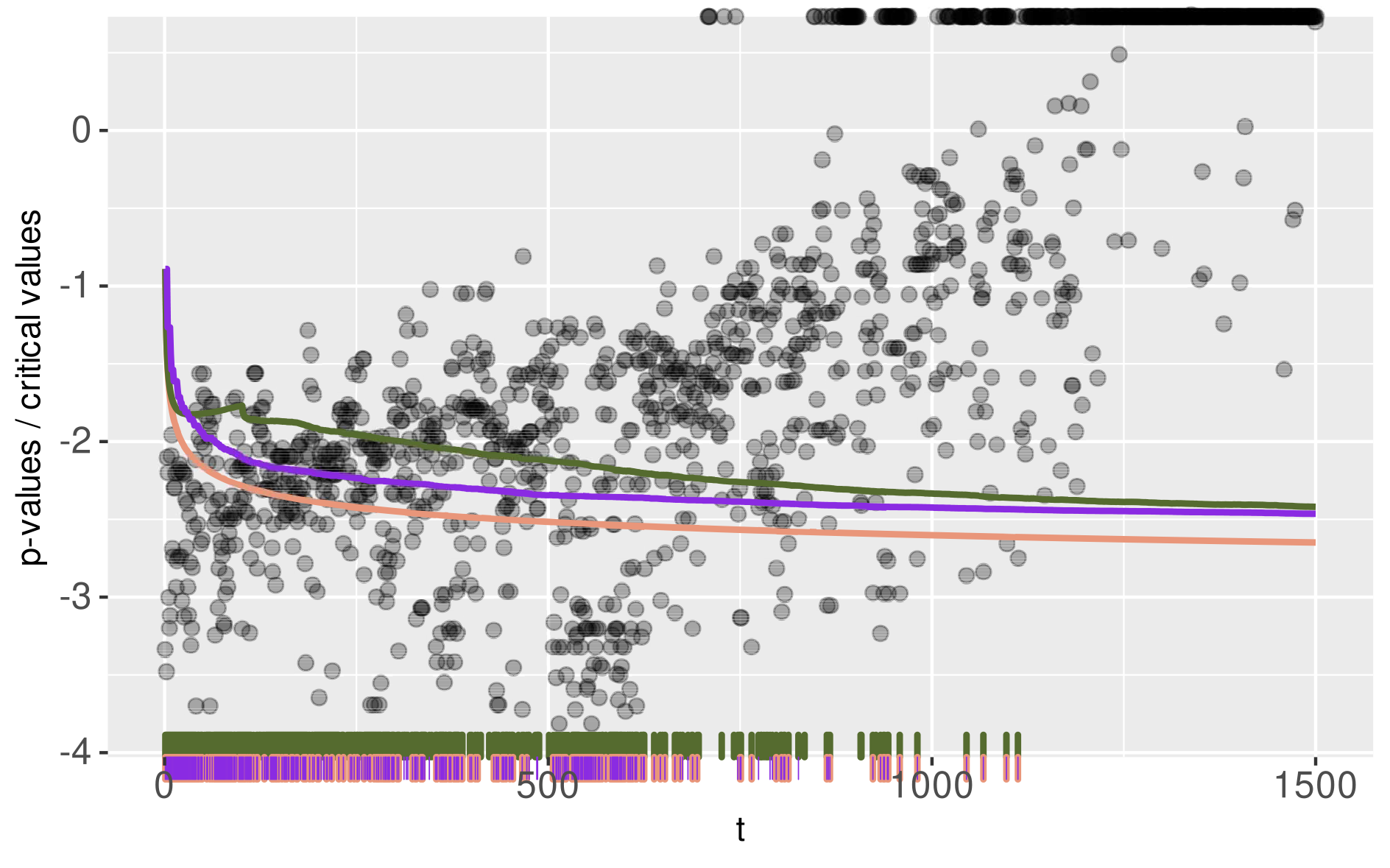} }
	\end{tabular}
	\vspace{-0.5cm}
	\caption{Critical values of OB (orange), $\rho$OB (green) and DS (purple) for the IMPC data (left panel is for male, right panel is for female). 
		\label{fig:compardelayedcriticalvalues}
	}
\end{figure}

 \subsection{Formal properties}\label{sec:comparSURE}
 
 {From the definition of the {DS}  approach we obtain the following comparison to OB and  $\rho$OB:}
 \begin{itemize}
 	\item the DS approach improves OB {uniformly} when $\gamma_t$ is nonincreasing: indeed $\delay(t)\leq t$, so that $\alphaDelay_t = \alpha \gamma_{\delay(t)}\geq \alpha \gamma_{t} = \alphabonf_t$.
 	
 	\item the DS approach does not depend on any other tuning parameter such as the bandwidth. By contrast, choosing this parameter badly in the $\rho$OB procedure may adversely affect its performance.
 	
 	\item the DS approach is another way of using the super uniformity reward. For instance, if there is no super uniformity reward,  that is, $F_t(\alpha \gamma_t)=\alpha \gamma_t$ for all $t$, then $b_t=t$ and the DS procedure reduces to OB.
 	
 \end{itemize}
 In addition, we have the following observations: 
 \begin{itemize}
 	\item Delayed start: If $F_t(x)=0$ for all $x<1$ and $t\leq T_0$ and $F_t(x)=x$ for $t\geq T_0+1$, the DS procedure is much more intuitive: it yields 
 	$b_1=T_0+1$ by \eqref{equbj} and $\alphaDelay_t=\alpha \gamma_{t-T_0}$ for $t\geq T_0+1$ which is the most natural way to proceed (just start the testing process at time $T_0+1$). By contrast, $\rho$OB (with rectangular kernel of bandwidth $r$) 
 	collects some reward in $\alphaOBSURE_{t}$, $1\leq t\leq T_0$, spends the reward in the following $r$ time points, but continues with
 	$\alphaOBSURE_{t}=\alpha \gamma_t$ for $t\geq T_0+r+1$. Hence, delaying spends the super-uniformity more intuitively than $\rho$OB in that situation.
 	More generally, in practice, we may therefore expect DS to be more efficient in the beginning of the stream. 
 	
 	\item Long/infinite delay: Conversely,  if there exists $T_0\geq 1$ such that for all $t\geq b_{T_0}+1$, $F_t(\alpha\gamma_{\mathcal{C}(T_0)+1})=0$, then we have $b_{T_0+1}=+\infty$ from \eqref{equbj}, 
 	which in turn implies $\delay(t)\leq T_0+1$. But for $t\geq b_{T_0}+1$, we have  $\delay(t)\geq T_0+1$ by \eqref{equdelay}. Hence, for $t\geq b_{T_0}+1$, $\delay(t)= T_0+1$ and 
 	the 'spending clock' freezes. On the one hand, we have $\alphaDelay_{t}=\alpha \gamma_{T_0+1}$ so the delaying works perfectly to effectively improve the OB critical values. On the other hand, this effectively stops the spending of any further budget and thus a large part of the wealth is left unspent.
 	This is in contrast to the SUR approach which uses a reward of an additive nature and thus always has a chance to spend the budget.
 	\item Under-utilization of wealth. The DS method processes each sub-budget $\alpha\gamma_j$ one at a time, until the transition to the next sub-budget $\alpha\gamma_{j+1}$ is made. In most cases, however,  the inequality \eqref{equbj} defining the transition time $b_j$ will be a strict inequality, meaning that when we move on to the next sub-budget we will have used $\sum_{t= b_{j-1}+1}^{b_j} F_t( \alpha \gamma_j) < \alpha \gamma_j$. Thus, this method does not exhaust the available sub-budgets. Moreover, since it neglects these 'alpha-gaps', they accumulate over time. 
 	This under-utilized wealth leads to unnecessary conservatism. Removing such  gaps was precisely the primary motivation for introducing our SUR method{, see Section~\ref{sec:sure}}.
 	
 	The most disadvantageous scenario occurs when $b_t=t$ for all $t\leq T$, so that the DS procedure reduces to the original OB procedure up to time $T$. As an example consider $\epsilon\in (0,\alpha \gamma_{T})$ for some large $T\geq 1$ and assume that the support of each $p_t$ is given by $S_t=\{\epsilon,A_t,\alpha \gamma_{t-1}\}\cup\{1\}$ (convention $\alpha \gamma_{0}=1$), where $A_t$ is a finite subset of $(\alpha \gamma_{t},\alpha \gamma_{t-1})$. Then we have 
 	$F_1(\alpha \gamma_1)+F_2(\alpha \gamma_1) = \alpha \gamma_1+\epsilon$ hence $b_1=1$, and more generally $F_{t}(\alpha \gamma_t)+F_{t+1}(\alpha \gamma_t) = \alpha \gamma_t+\epsilon$ for all $t\leq T$, which implies $b_t=t$ for all $t\leq T$. 
 	However, we know that OB does not allow to spend all the budget in such a discrete situation, see Figure~\ref{fig:nom_vs_eff_wealth_fwerproc}.
 	
 	A potential remedy for the conservatism of the DS method could be to  combine it with our SUR method. We describe such a hybrid approach in more detail in Section \ref{sec:hybrid}.
 	
 \end{itemize}
 
{In summary, it may be said that the delaying method is particularly appealing in terms of simplicity and elegance, while the primary aim of the SUR approach is on efficiency.}

%
%
%
%
 
 \subsection{Hybrid approach}\label{sec:hybrid}
 
 In this section, we describe a hybrid approach, combining the ideas underlying DS and SUR, in order to
 improve the utilization of wealth of DS.
 
 The method starts as follows: first let
 $
 \alphaHyb_1=\alpha \gamma_1,$  $\dots,$ $\alphaHyb_{b_1}=\alpha \gamma_1
 $
 as long as $F_1(\alpha \gamma_1)+\dots+F_{b_1}(\alpha \gamma_1)\leq \alpha \gamma_1$. Then consider the reward $\rho_1=\alpha \gamma_1-(F_1(\alpha \gamma_1)+\dots+F_{b_1}(\alpha \gamma_1))$ and let
 $
 \alphaHyb_{b_1+1}=\alpha \gamma_2+\rho_1,$  $\dots,$ $\alphaHyb_{b_2}=\alpha \gamma_2+\rho_1
 $
 as long as
 $F_{b_1+1}(\alpha \gamma_2+\rho_1)+\dots+F_{b_2}(\alpha \gamma_2+\rho_1)\leq \alpha \gamma_2 + \rho_1$. 
 More generally, let $b_0=0$, $\rho_0=0$, and for all $j\geq 1$, 
 \begin{align*}
 \alphaHyb_{b_{j-1}+1}&=\alpha \gamma_j + \rho_{j-1}, \dots, \alphaHyb_{b_{j}}=\alpha \gamma_j + \rho_{j-1}\\
 b_j &=\max\left\{T\geq 1\::\: \sum_{t= b_{j-1}+1}^{T} F_t( \alpha \gamma_j + \rho_{j-1})\leq \alpha \gamma_j + \rho_{j-1}\right\}\\
 \rho_j &= \alpha \gamma_j + \rho_{j-1} - \left(\sum_{t= b_{j-1}+1}^{b_j} F_t( \alpha \gamma_j + \rho_{j-1})\right)
 .
 \end{align*}
 Then the online FWER control holds because for all $j_0\geq 1$, we have
 \begin{align*}
 \sum_{t \geq 1} F_t(\alphaHyb_t) &= \sum_{j=1}^{j_0} \left(\sum_{t= b_{j-1}+1}^{b_j} F_t( \alpha \gamma_j + \rho_{j-1})\right) + F_{b_{j_0}+1}( \alpha \gamma_{j_0+1} + \rho_{j_0})\\
 &\leq \sum_{j= 1}^{j_0}  (\alpha \gamma_j + \rho_{j-1} -\rho_j) + \alpha \gamma_{j_0+1} + \rho_{j_0} = \sum_{j= 1}^{j_0+1}\alpha \gamma_j\leq \alpha,
 \end{align*}
 because $ \sum_{j= 1}^{j_0} (\rho_{j-1} -\rho_j) = -\rho_{j_0}$ (telescopic sum). 
 When $\rho_t=0$ for all $t\geq 1$, the hybrid approach reduces to the DS approach. When  $b_j=j$, the hybrid approach reduces to the greedy SUR procedure. 
 
 We can also combine the  DS with smoothed SUR rewarding, which 
 gives us the following, slightly more involved, procedure. For some { SUR spending sequence}  $\gamma' = (\gamma'_t)_{t \geq 1}$ (nonnegative and such that $\sum_{t \geq 1} \gamma'_t \leq 1$),  let $b_0=0$, $\rho_0=0$ and for all $j\geq 1$, 
 \begin{align*}
 \alphaHyb_{b_{j-1}+1}&=\alpha \gamma_j + \sum_{i=1}^{j-1}\gamma'_{j-i}\rho_{i}, \:\:\:\dots,\:\:\: \alphaHyb_{b_{j}}=\alpha \gamma_j + \sum_{i=1}^{j-1}\gamma'_{j-i}\rho_{i}\\
 b_j &=\max\left\{T\geq 1\::\: \sum_{t= b_{j-1}+1}^{T} F_t\left( \alpha \gamma_j + \sum_{i=1}^{j-1}\gamma'_{j-i}\rho_{i}\right)\leq \alpha \gamma_j + \sum_{i=1}^{j-1}\gamma'_{j-i}\rho_{i}\right\}\\
 \rho_j &= \alpha \gamma_j + \sum_{i=1}^{j-1}\gamma'_{j-i}\rho_{i} - \left(\sum_{t= b_{j-1}+1}^{b_j} F_t\left( \alpha \gamma_j + \sum_{i=1}^{j-1}\gamma'_{j-i}\rho_{i}\right)\right).
 \end{align*}
 The online FWER control holds because for all $j_0\geq 1$, we have
 \begin{align*}
 \sum_{t \geq 1} F_t(\alphaHyb_t) &\leq \sum_{j=1}^{j_0} \left(\sum_{t= b_{j-1}+1}^{b_j} F_t\left( \alpha \gamma_j +  \sum_{i=1}^{j-1}\gamma'_{j-i}\rho_{i}\right)\right) + \alpha \gamma_{j_0+1} +  \sum_{i=1}^{j_0}\gamma'_{j_0+1-i}\rho_{i}.
 \end{align*}
 Now  letting $a_T=\sum_{t=1}^T\gamma'_t$, we obtain
 \begin{align*}
 &\sum_{j=1}^{j_0} \left(\sum_{t= b_{j-1}+1}^{b_j} F_t\left( \alpha \gamma_j +  \sum_{i=1}^{j-1}\gamma'_{j-i}\rho_{i}\right)\right)\\
 &\leq 
 \sum_{j=1}^{j_0} a_{j_0-j+1}\left(\alpha \gamma_j +  \sum_{i=1}^{j-1}\gamma'_{j-i}\rho_{i} - \rho_j\right)+\sum_{j=1}^{j_0} (1-a_{j_0-j+1})\left(\alpha \gamma_j +  \sum_{i=1}^{j-1}\gamma'_{j-i}\rho_{i}\right)\\
 &= \sum_{j=1}^{j_0} \alpha \gamma_j + \sum_{j=1}^{j_0}   \sum_{i=1}^{j-1}\gamma'_{j-i}\rho_{i} -  \sum_{j=1}^{j_0} a_{j_0-j+1}  \rho_j\\
 &=\sum_{j=1}^{j_0} \alpha \gamma_j + \sum_{i=1}^{j_0-1} a_{j_0-i} \rho_i -  \sum_{j=1}^{j_0} a_{j_0-j+1}  \rho_j,
 \end{align*}
 and the latter is equal to
 $
 \sum_{j=1}^{j_0} \alpha \gamma_j - \sum_{j=1}^{j_0-1} \gamma'_{j_0-j+1} \rho_j  -  a_{1}  \rho_{j_0}
 = 
 \sum_{j=1}^{j_0} \alpha \gamma_j - \sum_{j=1}^{j_0} \gamma'_{j_0-j+1} \rho_j  .
 $
 {Combining this with the above bound for the FWER concludes the proof.}
 
 To compare the performance of the hybrid approach with the SUR and DS approaches, we use the simulation setting from  Section~\ref{sec:numerical_results} in the case where the signal is positioned at the beginning of the stream for each simulation run, which is the most favorable position of the signal for any procedure (see Section~\ref{sec:varyH1} for more details). {We consider both procedures based on the uniform kernel (bandwidth $h=100$) and  those based on the greedy spending sequence (denoted by `greedy').}
 
 {Figure~\ref{fig:hybrid} shows that taking super-uniformity into account is always beneficial, regardless of the specific approach used. 
 	The base DS method performs similarly to {the greedy $\rho$OB and the greedy hybrid}. 
 	{In contrast, the hybrid approach based on a uniform kernel improves DS, with performance close to  $\rho$OB. Hence, we conclude that 
 		closing the alpha-gaps by smoothing with an adequate kernel can make the hybrid approach as powerful as the smoothed $\rho$OB method. However, given the added complexity of the hybrid approach, we prefer to stick with the smoothed SUR.}

 	\begin{figure}[h!]
 		\centering
 		\includegraphics[width=.49\linewidth]{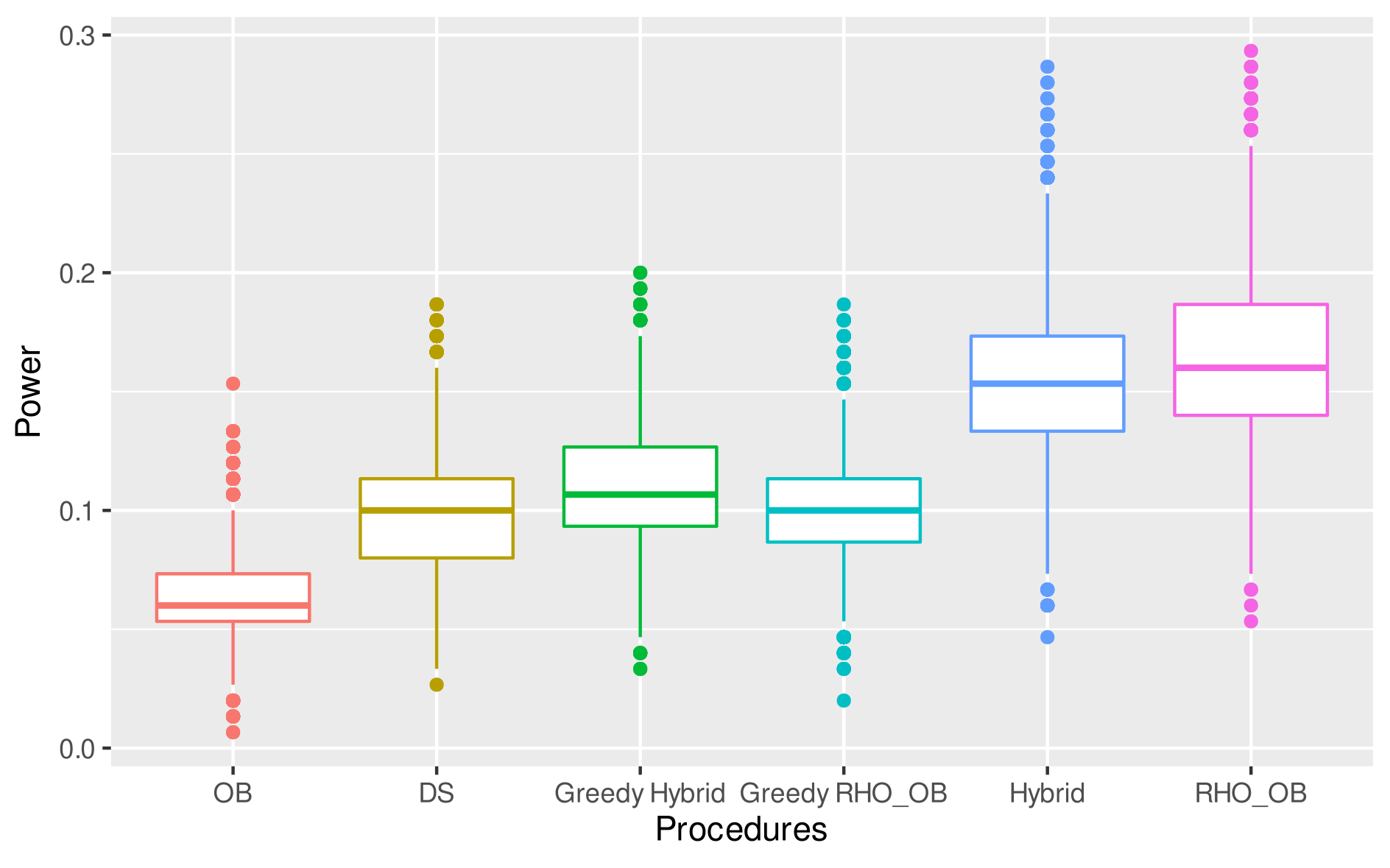}
 		\vspace{-0.5cm}
 		\caption{
 			Power 
 			of several online FWER controlling approaches for simulated data (see text): online Bonferroni (OB),  {Delayed spending} (DS), greedy hybrid,  greedy $\rho$OB, hybrid,  $\rho$OB.  \label{fig:hybrid}
 		}
 	\end{figure}
 	 
\section{Complements on generalized $\alpha$-investing rules }\label{sec:GAI++}

\subsection{SUR-GAI++ rules}\label{sec:defSURGAI}

GAI++ rules have been introduced in \cite{ramdas2017online} to control the (m)FDR. Here, we can extend them to our super-uniform setting as follows. Let us consider the following recursive constraints: for $t\geq 1$,
\begin{align*}
R_t&=\ind{p_t\leq \alpha_t}\\
W(t)&=W(t-1)-\phi_t + R_t \psi_t \:\:\:\mbox{‘wealth available at time $t+1$’}\\
\phi_t&\in [0, W(t-1)] \:\:\:\mbox{‘spent at time $t$’}\\
\psi_t&\leq b_t+ \min\left(\phi_t ,\phi_t/F_t(\alpha_t) -1\right)\:\:\:\mbox{‘reward at time $t$’}\\
\psi_t&\geq 0\\
b_t&= \alpha-W_0\ind{t\leq \tau_1},
\end{align*}
where $W(0)=W_0\in [0,\alpha]$. 
Any choice of $W_0$ and $\alpha_t,\phi_t,\psi_t$ that are $\mathcal{F}_{t-1}$ measurable and satisfying the above constraints defines a SUR-GAI++ procedure.
Here, the only difference with the original GAI++ rule is the presence of $F_t(\alpha_t)$ instead of $\alpha_t$ in the definition of $\Psi_t$.

\begin{proposition}
Consider the setting of Section~\ref{sec:setting} where a null bounding family $\mathcal{F}=\{F_t,t\geq 1\}$ 
    satisfying \eqref{eqn:condsuperunif} is at hand. Then any SUR-GAI++ procedure controls the mFDR at level $\alpha$. 
\end{proposition}

The proof is totally analogous to the one of Theorem~1 in \cite{ramdas2017online} (adapted to the mFDR, so without using any monotonicity). 

\subsection{GAI++ weighting}\label{sec:ramdasweighting}

Consider (continuous) $p$-values satisfying \eqref{eqn:superunif0}-\eqref{indep} and weights $w_t\geq 0$ that are $\mathcal{F}_{t-1}$ measurable for all $t$.
In Section~5 of \cite{ramdas2017online}, the following (implicit) GAI++ weighting scheme has been proposed: 
\begin{align*}
R_t&=\ind{p_t\leq w_t\alpha_t}\\
W(t)&=W(t-1)-\phi_t + R_t \psi_t\\
\phi_t&\in [0, W(t-1)]\\
\psi_t&\leq b_{t } + \min\left(\phi_t ,\phi_t/(w_t \alpha_t)  -1\right)\\
\psi_t&\geq 0\\
b_t&= \alpha-W_0\ind{t\leq \tau_1}.
\end{align*}
Note that the latter constraints are similar to the constraints given in Section~\ref{sec:defSURGAI} for $F_t(x)=(w_t x)\wedge 1$ (up to the ‘$\wedge 1$’ which makes the constraints here slightly more stringent) so that this weighting case is a particular SUR-GAI++ procedure.

For given raw weights $r_t\geq 0$ ($\mathcal{F}_{t-1}$ measurable), an explicit procedure which is used in  \cite{ramdas2017online}\footnote{This procedure is available at \url{https://github.com/fanny-yang/OnlineFDRCode}}, is obtained by choosing $\alpha_t$, $w_t$, $\phi_t$, $\psi_t$ as follows:
\begin{align*}
w_t&=r_t \wedge \frac{1}{1-b_t}\\
\phi_t&=\alpha_t = W_0 \gamma_{t} + \sum_{j\geq 1} \gamma_{t-\tau_j} \psi_{\tau_j}\\
\psi_t&= b_{t} + \min\left(\phi_t,1/w_t  -1\right).
\end{align*}
This choice is valid because $\alpha_t \leq W(t-1)$ for all $t$. Indeed,
\begin{align*}
W(t-1)=W_0+\sum_{i=1}^{t-1} (-\alpha_i + R_i \psi_i),
\end{align*}
so $\alpha_t \leq W(t-1)$ if and only if
$
\sum_{i=1}^t \alpha_i\leq W_0+\sum_{i=1}^{t-1} R_i \psi_i,
$
which is true.

\subsection{Our $\rho$-LORD is a SUR-GAI++ rule}\label{sec:ALORSisGAI}

We claim here that the procedure $\rho$-LORD corresponds to a SUR-GAI++ rule with the choice 
 $\phi_t=F_t(\alpha_t)$, $\psi_t=b_t$, and 
\begin{equation}\label{equ-alphat-GAI}
 \alpha_t = W_0 \gamma_t + (\alpha-W_0) \gamma_{t-\tau_1} + \alpha\sum_{j\geq 2} \gamma_{t-\tau_j} +  \sum_{i=1}^{t-1} \gamma'_{t-i} \rho_{i} \quad t\geq 1.
\end{equation}
To establish this, we check that all constraints given in Section~\ref{sec:defSURGAI} are satisfied. The only non-trivial one is $\phi_t=F_t(\alpha_t)\leq W(t-1)$. Let us now prove it. Recall that $W(t)=W(t-1)-\phi_t+R_t b_t$ and $W(0)=W_0$. Hence $\alpha_1=W_0 \gamma_1\leq W_0$. Moreover, for $t\geq 2$, 
$$
W(t-1)=W_0 +  ( \alpha-W_0) \ind{t -1\geq \tau_1} + \alpha \sum_{j\geq 2} \ind{t -1\geq \tau_j}-\sum_{i=1}^{t-1} F_i(\alpha_i).
$$  
So we have $\bar{\alpha}_t\leq W(t-1)$ for the critical value
\begin{align*}
\bar{\alpha}_t=& \left(\sum_{i=1}^t \gamma_i \right)W_0 +\sum_{i=1}^{t-1}\left( ( \alpha-W_0) \gamma_{i-\tau_1+1}\ind{i\geq \tau_1} + \alpha \sum_{j\geq 2} \gamma_{i-\tau_j+1}\ind{i\geq \tau_j}\right) \\
&-\sum_{i=1}^{t-1}
\left[ a_{t-i} F_{i}(\bar{\alpha}_i )
        + (1-a_{t-i}) \bar{\alpha}_i \right],
\end{align*}
by letting $a_t=\sum_{i=1}^t \gamma'_i$.
But now, we have that $\bar{\alpha}_t= \alpha_t $ for all $t$, for $\alpha_t$ defined by \eqref{equ-alphat-GAI}. {Indeed, this can be seen from Lemma~\ref{lem:gencomput}, applied with $\lambda=0$ and $\alpha^0_T $ being the LORD critical values.}

\section{Additional numerical experiments}\label{sec:addnum}

\subsection{Sample size}\label{sec:varyN}

Figure \ref{fig:varyN} illustrates results when the sample size $N$, 
i.e., the subjects number per group, takes values in the set $\{25, 50, \ldots, 150\}$. 
As expected, the power plots show that the detection problem becomes easier when $N$ increases. In fact, for large $N$ the power of all procedures converge to 1.}
We see that our rewarded procedures do well on the whole range of $N$ values and improve substantially on existing OMT procedures for small and moderate values of $N$, including our default value $N = 25$. 

\begin{figure}[h!]
    \begin{center}
    \begin{tabular}{cc}
    \includegraphics[width=.5\linewidth]{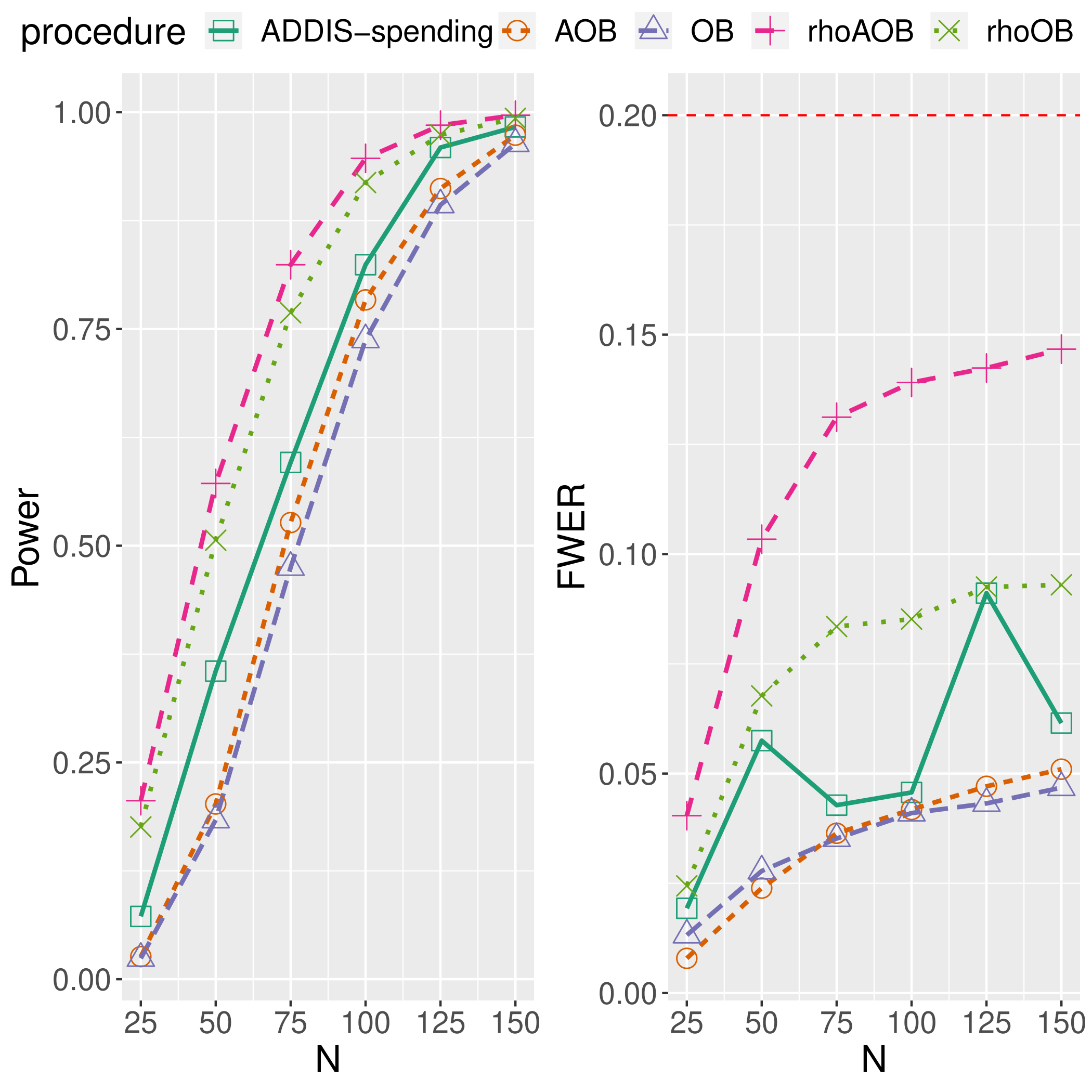}  & \includegraphics[width=.5\linewidth]{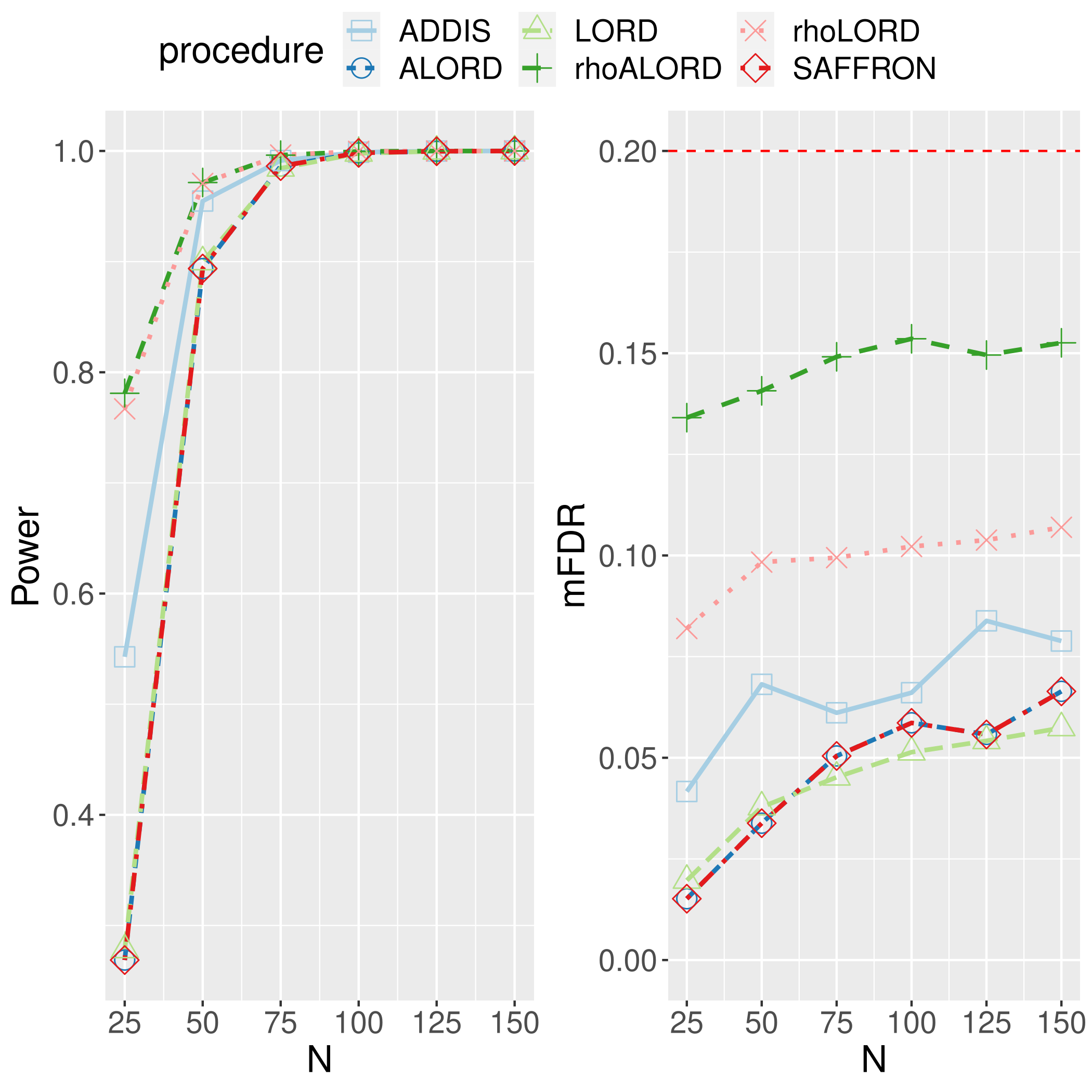}  
    \end{tabular}
    \end{center}
    \caption{Power and type I error rates of the considered procedures
     versus $N \in \{25, 50, \ldots, 150\}$, the number of subjects in the groups. }
    \label{fig:varyN}
\end{figure}

\subsection{Signal strength}\label{sec:varyp3}
Here, we vary the strength of the signal $p_3$ in the set $\{0.1, 0.2, \ldots, 1 \}$.
We see that the SUR procedures dominate their base counterparts, as expected. 
In addition, depending on the signal strength, the gain in power can be considerable. 
Also note that, perhaps surprisingly, all curves exhibit a decrease in power for $p_3$ near $1$. 
Since this happens even for the original OB procedure, this is not due to the super-uniformity reward, 
but could perhaps be caused by the behavior of the power function of multiple Fisher exact tests taken at different levels. 
\begin{figure}[h!]
    \begin{center}
    \begin{tabular}{cc}
    \includegraphics[width=.5\linewidth]{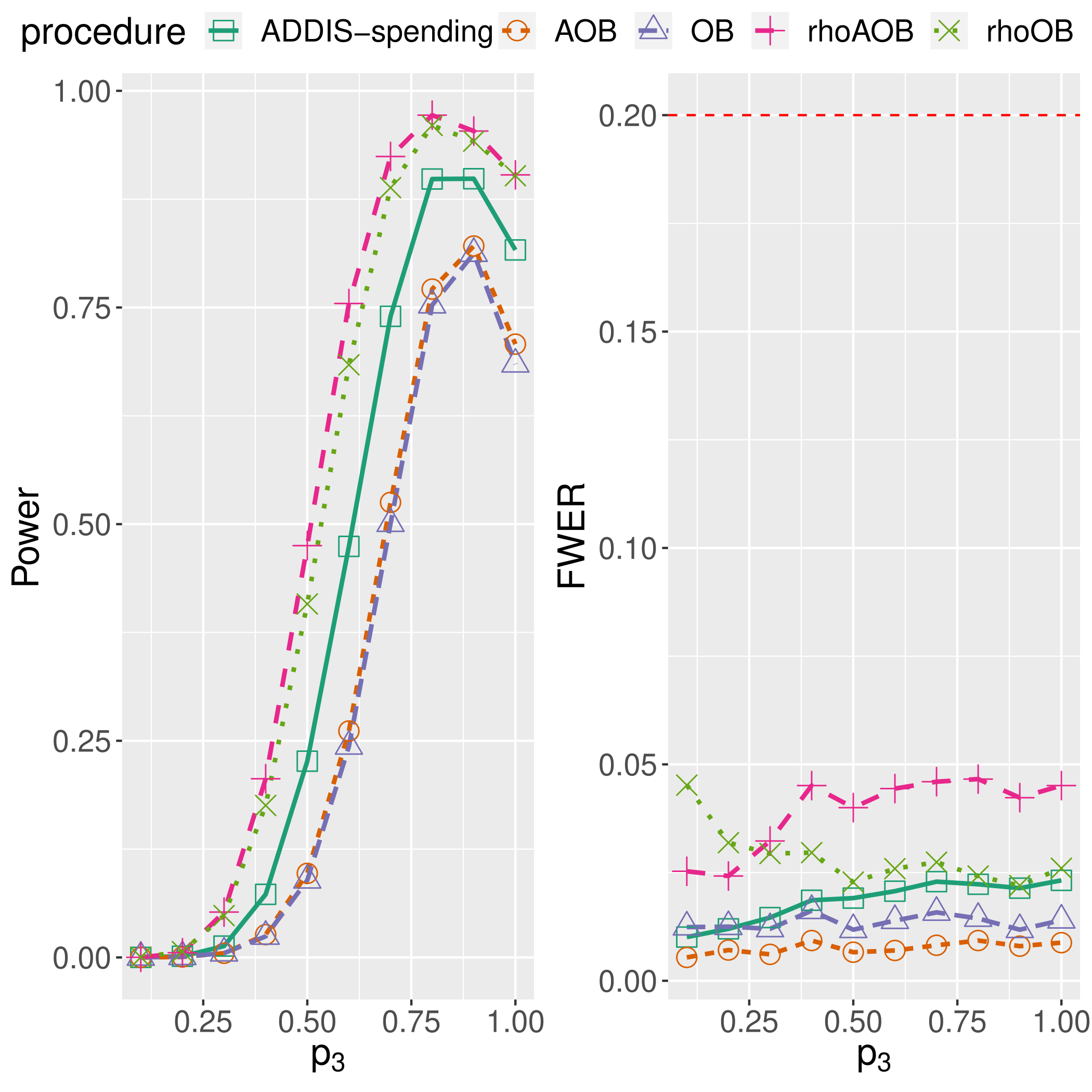}  & \includegraphics[width=.5\linewidth]{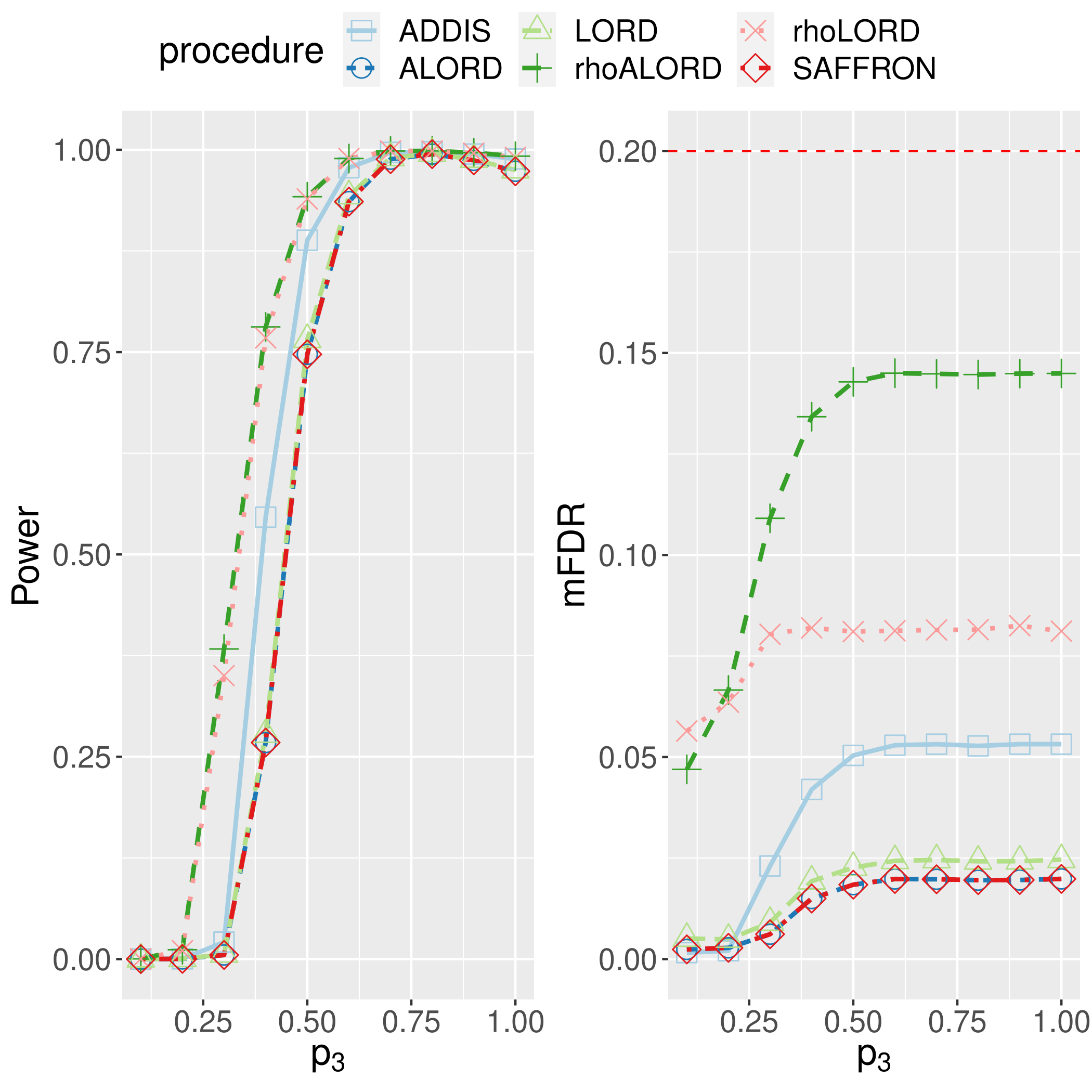}  
    \end{tabular}
    \end{center}
    \caption{Power and type I error rates of the considered procedures
     versus the strength of the signal $p_3 \in \{0.1, 0.2, \ldots, 0.9, 1 \}$.}
    \label{fig:varyp3}
\end{figure}

\subsection{Local alternatives}\label{sec:refinement}

As Figure \ref{fig:varyN} demonstrates, for a fixed value of the signal strength $p_3$, the detection problem becomes easier as $N$ increases, so that all procedures attain a power of 1. In this section we are interested in obtaining a more refined analysis of the various power curves  when $N$ is large. To this end, we introduce local alternatives, i.e. we now model $p_3$ as a function of the sample size $N$. To be more specific, we take $N \in \{5,10,\ldots,30\} \times 1000$ and set $p_3 = p_1 + \frac{1}{\sqrt{N}} $ for mFDR procedures and, $p_3 = p_1 + \frac{1.5}{\sqrt{N}} $ for FWER procedures, we fix $p_1 = p_2 = 0.1$, and generate simulated data as in Section~5.2. 
Figure \ref{fig:varyN_localalter} displays power and error rates for this data. Taking $N$ as a (crude) proxy for discreteness, 
{we observe that even with a low discreteness (say $N\leq 30000$) the SUR methods still provide some degree of improvement. }
{Finally, for FWER procedures, ADDIS-spending provides the best power performance over the whole range of the experiment. 
This might be explained by the setting causing very conservative nulls $p$-values (\emph{i.e.} very close to 1), thus allowing the discarding scheme to redistribute and spend a large part of the wealth on testing alternative hypotheses.	Using the SUR method along with the discarding scheme \citep{tian_onlinefdr_2019, tian_onlinefwer_2020} might provide an interesting avenue for further improvement, but this would define yet another class of procedures, which is outside of the scope of this paper.}

\begin{figure}[h!]
    \begin{center}
    \begin{tabular}{cc}
    \includegraphics[width=.5\linewidth]{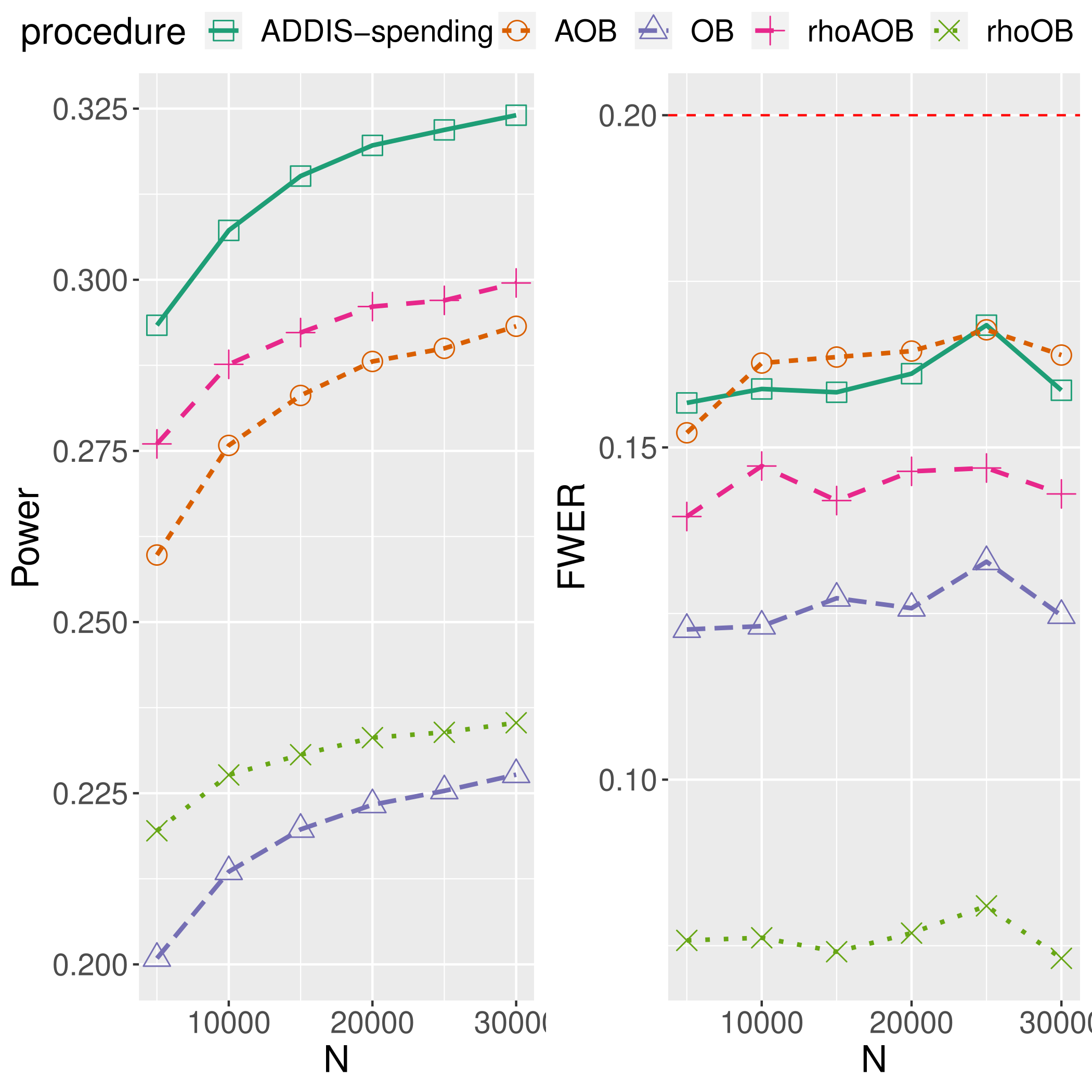}  & \includegraphics[width=.5\linewidth]{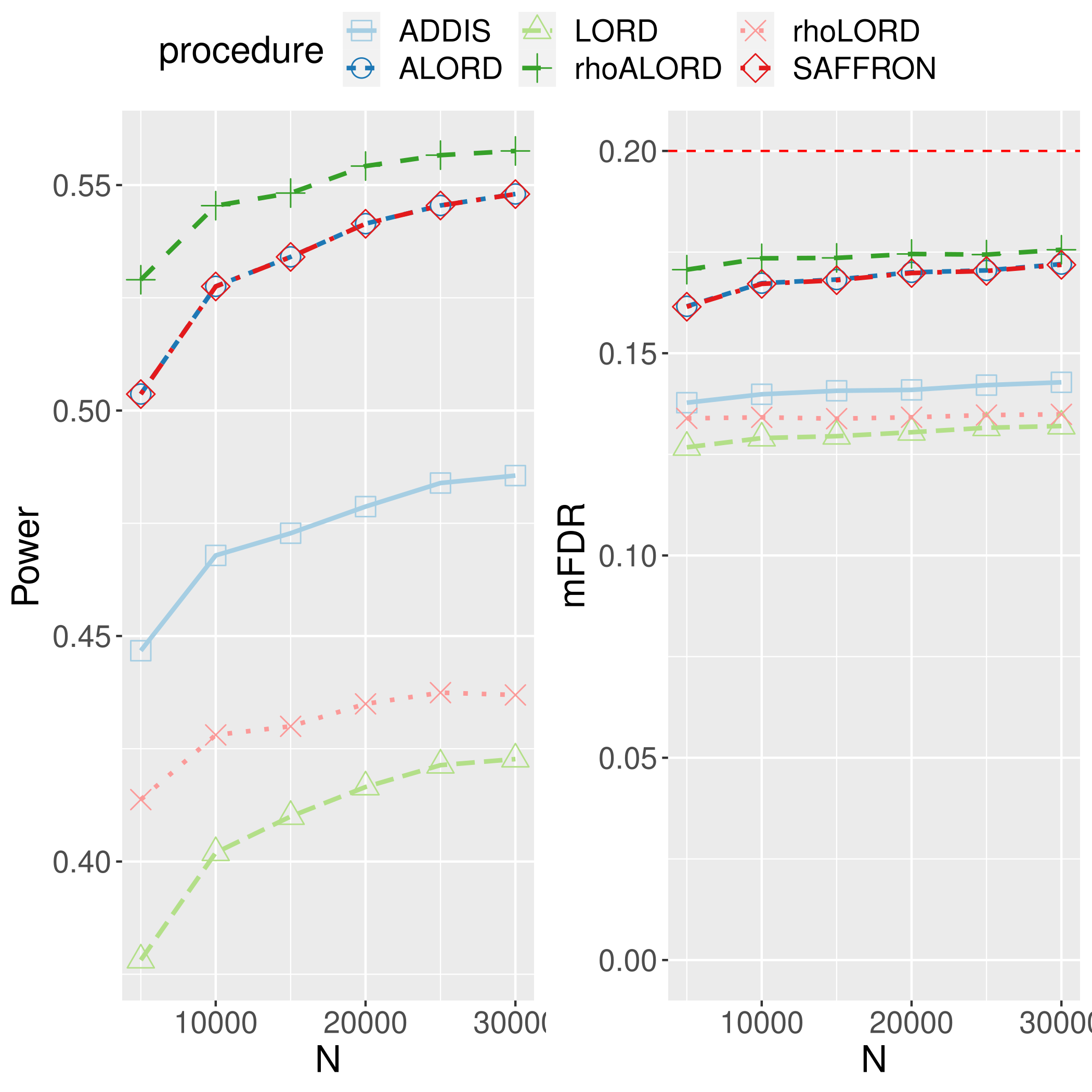}  
    \end{tabular}
    \end{center}
    \caption{Power and type I error rates of the considered procedures
     versus $N \in \{5,10,\ldots,30\} \times 1000$, with local alternatives. }   \label{fig:varyN_localalter}
\end{figure}

\subsection{Adaptivity parameter}\label{apenadapt}

We study the choice of $\lambda$ for the procedures using adaptivity. 
It seems that $\lambda = 0.5$ is a reasonable choice for the adaptive procedures.

\begin{figure}[h!]
    \begin{center}
    \begin{tabular}{cc}
        \includegraphics[width=.5\linewidth]{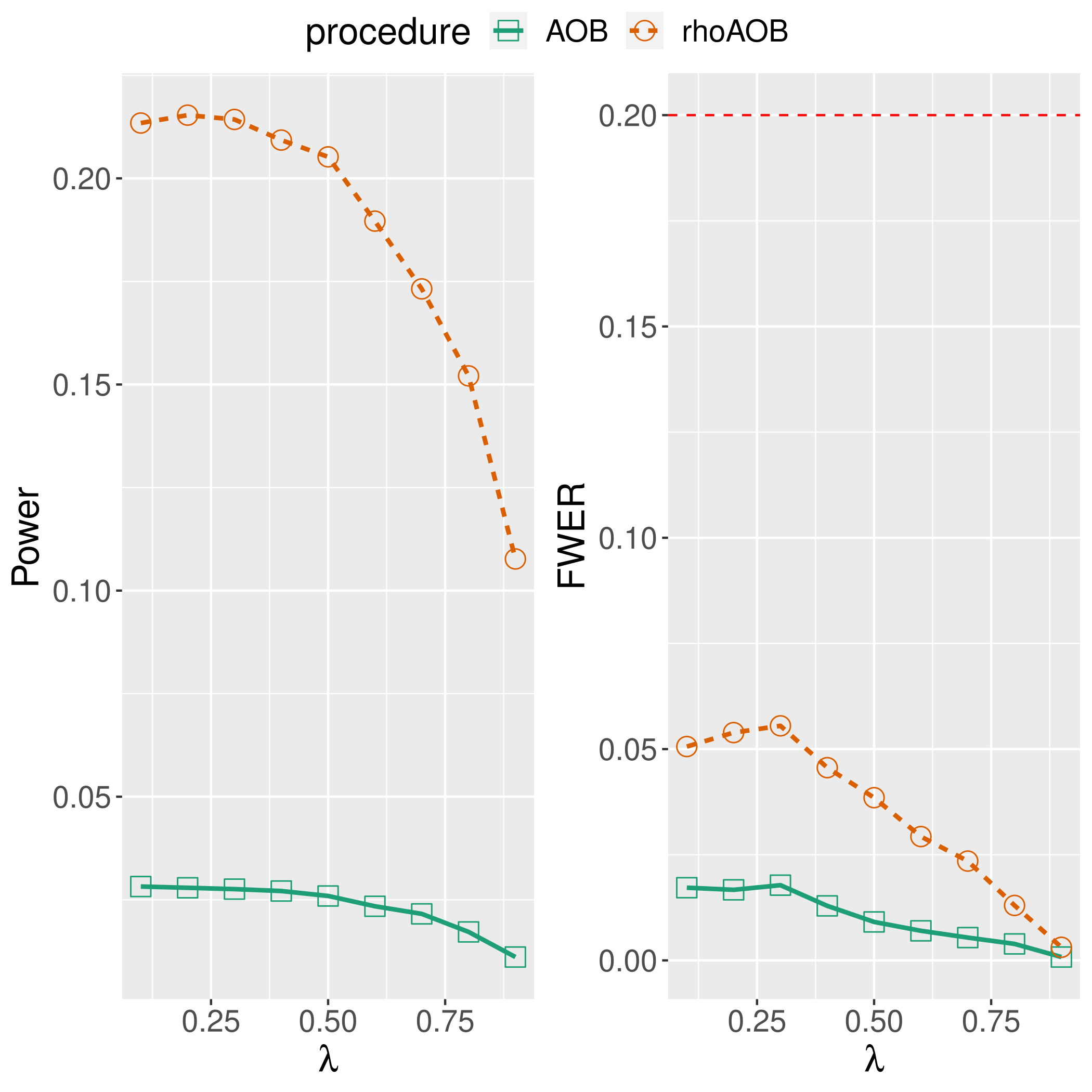} & \includegraphics[width=.5\linewidth]{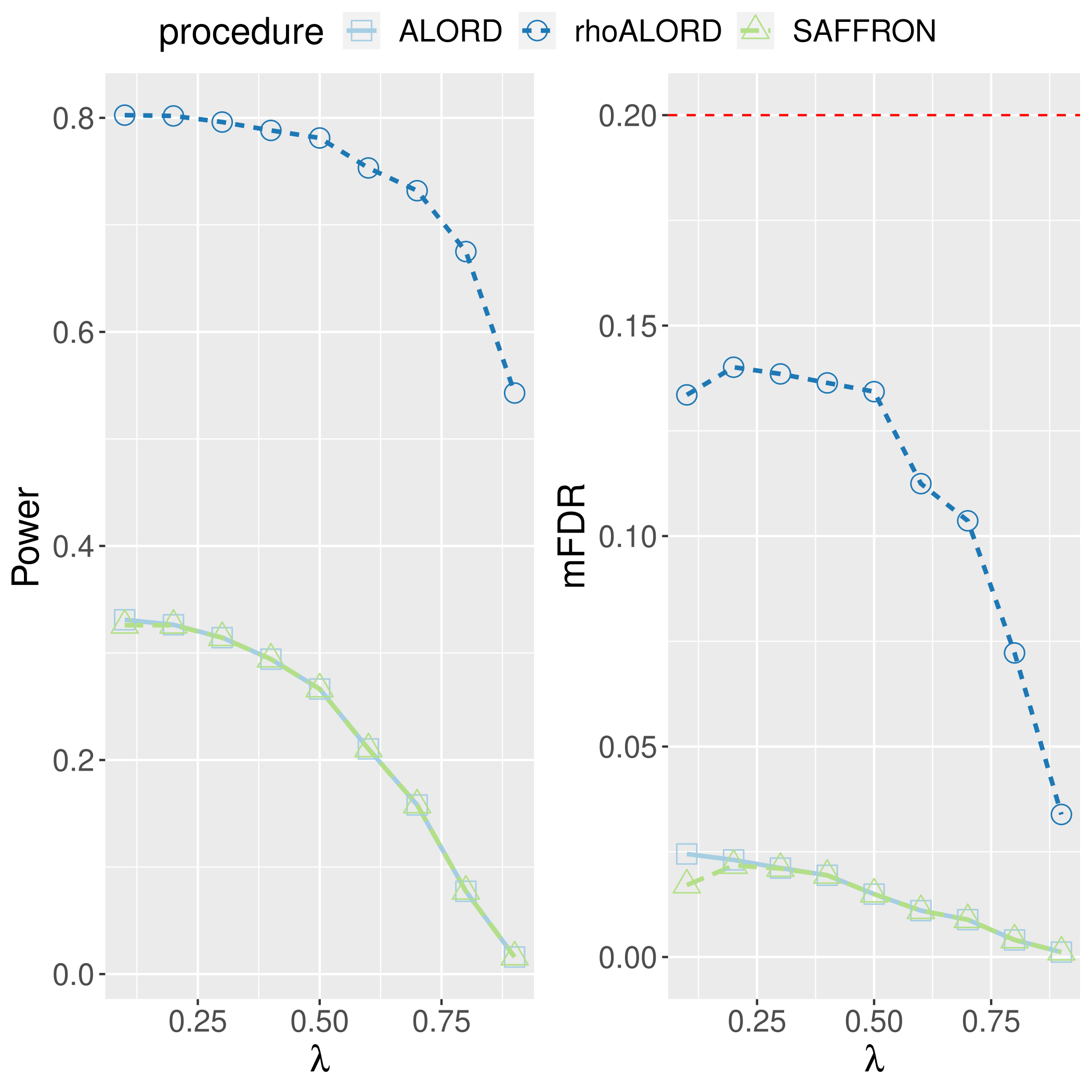}  
    \end{tabular}
    \end{center}
    \caption{Power and type I error rates, for the considered procedures, versus the adaptivity parameter $\lambda$.}
    \label{fig:varylambda}
\end{figure}

\subsection{Rectangular kernel bandwidth}\label{apenband}

Finally, we study the choice of the bandwidth parameter for the rectangular kernel used for the rewarded procedures.
As we can see, using a smaller bandwidth provides the best performance for the mFDR controlling rewarded procedures, 
whereas FWER controlling procedures require a larger bandwidth. 
The choices $h=100$ for FWER controlling procedures, and $h=10$ for mFDR controlling procedures seem 
reasonable although not necessarily optimal.
\begin{figure}[h!]
    \begin{center}
    \begin{tabular}{cc}
        \includegraphics[width=.5\linewidth]{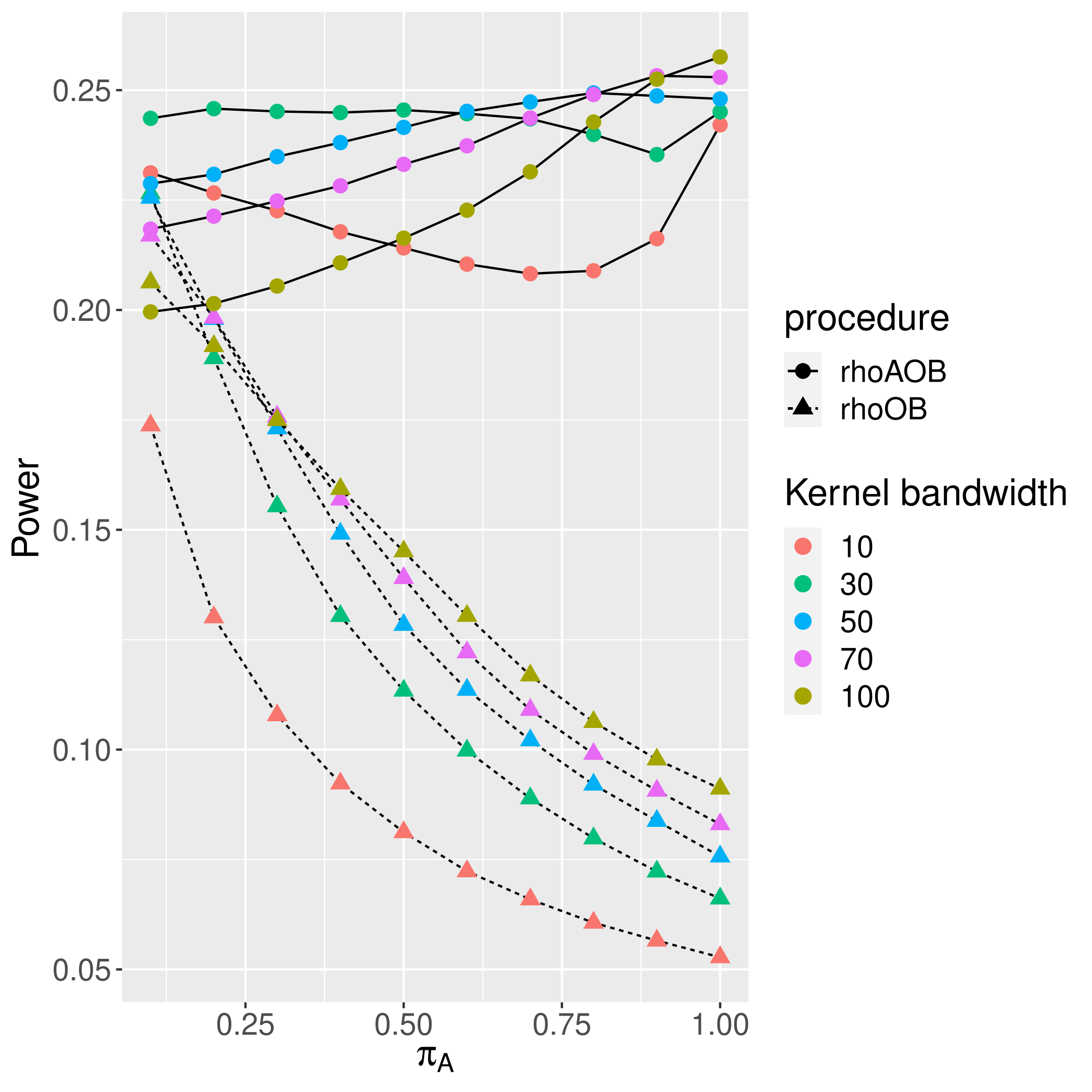} & \includegraphics[width=0.5\linewidth]{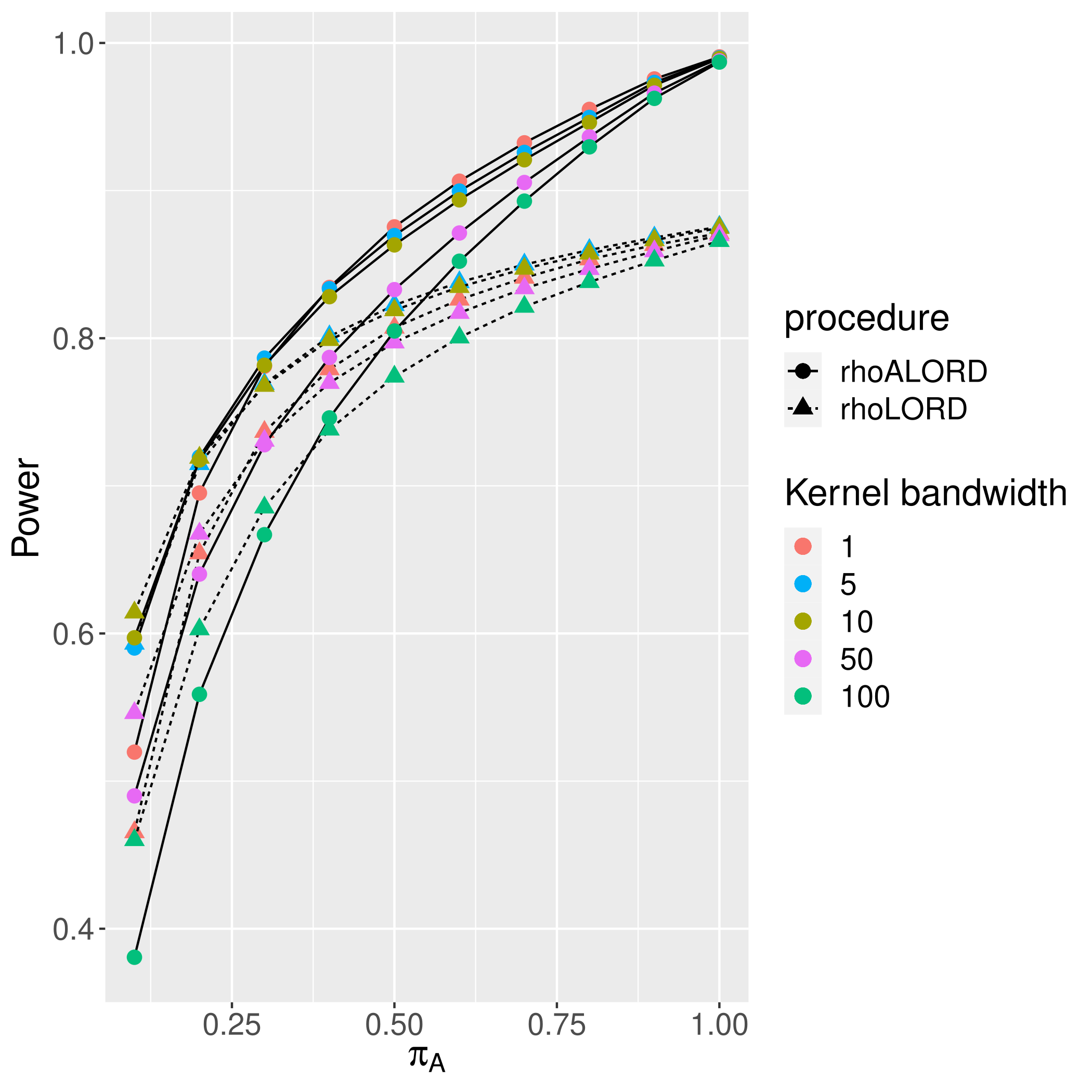} 
    \end{tabular}
    \end{center}
    \caption{Power for FWER (left) and mFDR (right) rewarded procedures 
    versus the proportion of signal $\pi_A$, for different kernel bandwidths.}
    \label{fig:varykernelbandwidth}
\end{figure}

\section{Additional figures for the analysis of IMPC data}\label{apenIPMCAdditional}
\subsection{Localization of small $p$-values}\label{apenIPMCLocalization}
Figures~\ref{fig:impc:male:pvalues} and~\ref{fig:impc:female:pvalues} show that small $p$-values 
mostly occur at the beginning of the data set, both for male and female mice.
\begin{figure}[h!]
	\begin{center}
		\begin{tabular}{cc}
			\includegraphics[width=1\linewidth]{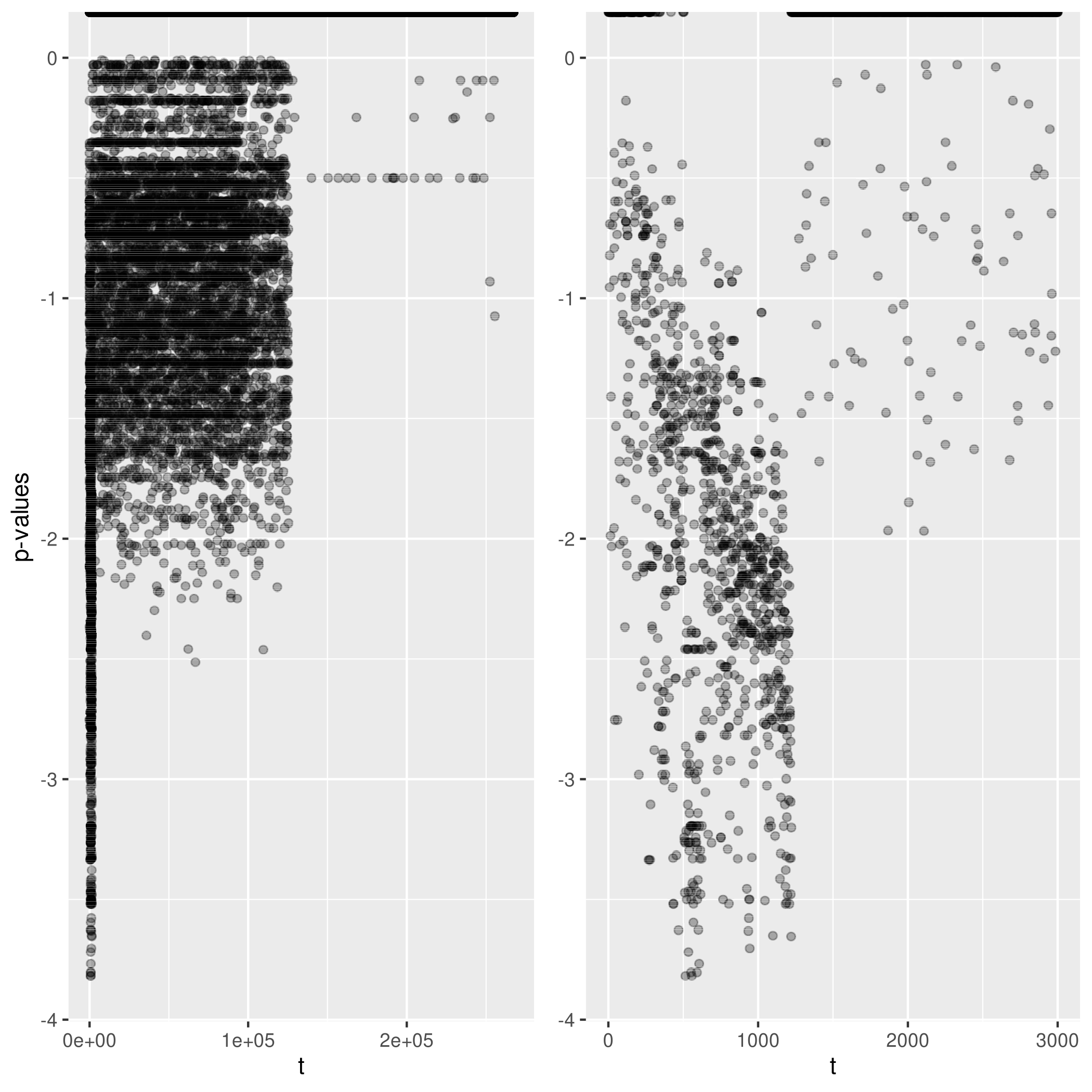}  
		\end{tabular}
	\end{center}
	\caption{$p$-values for male mice in the IMPC data of Section~\ref{sec:real_data_appli}. 
    The left panel presents all $p$-values, the right panel the first 3000 $p$-values. 
    The $p$-values have been transformed as in  Figure~\ref{fig:smoothedcvs}.}
	\label{fig:impc:male:pvalues}
\end{figure}
\begin{figure}[h!]
	\begin{center}
		\begin{tabular}{cc}
			\includegraphics[width=1\linewidth]{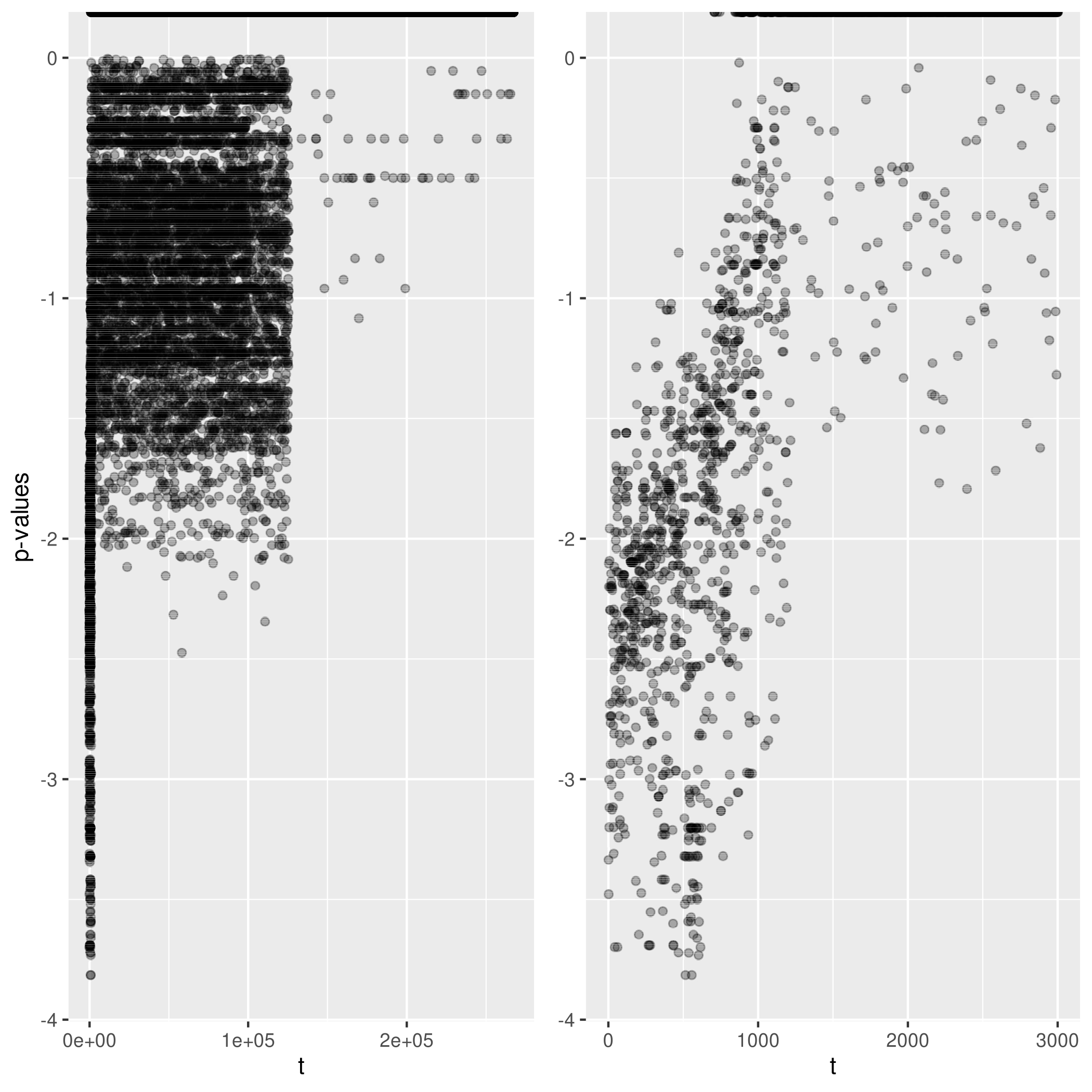}  
		\end{tabular}
	\end{center}
	\caption{$p$-values for female mice in the IMPC data of Section~\ref{sec:real_data_appli}. 
    The left panel presents all $p$-values, the right panel the first 3000 $p$-values. 
    The $p$-values have been transformed as in  Figure~\ref{fig:smoothedcvs}.}
	\label{fig:impc:female:pvalues}
\end{figure}

\subsection{Figures for female mice in the IMPC data}\label{apenIPMCfemale}
Figures~\ref{fig:mfdrimpcbasevsrewarded:fem} and~\ref{fig:fwerimpcbasevsrewarded:fem} 
display the critical values of the studied online procedures when applied to the IMPC data in the case of female mice.
\begin{figure}[h!]
    \begin{center}
    \begin{tabular}{cc}
    \includegraphics[width=.5\linewidth]{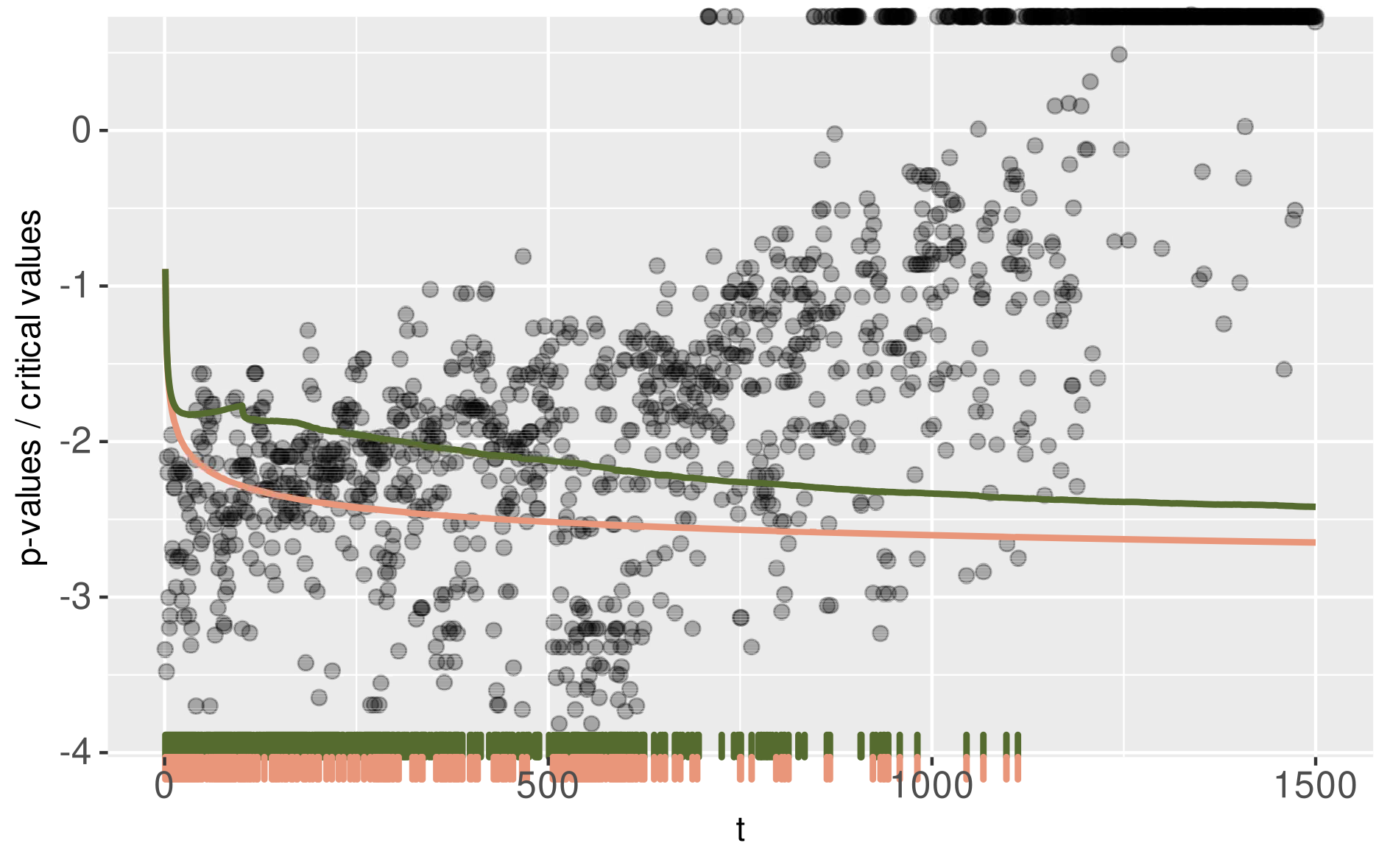} & \includegraphics[width=.5\linewidth]{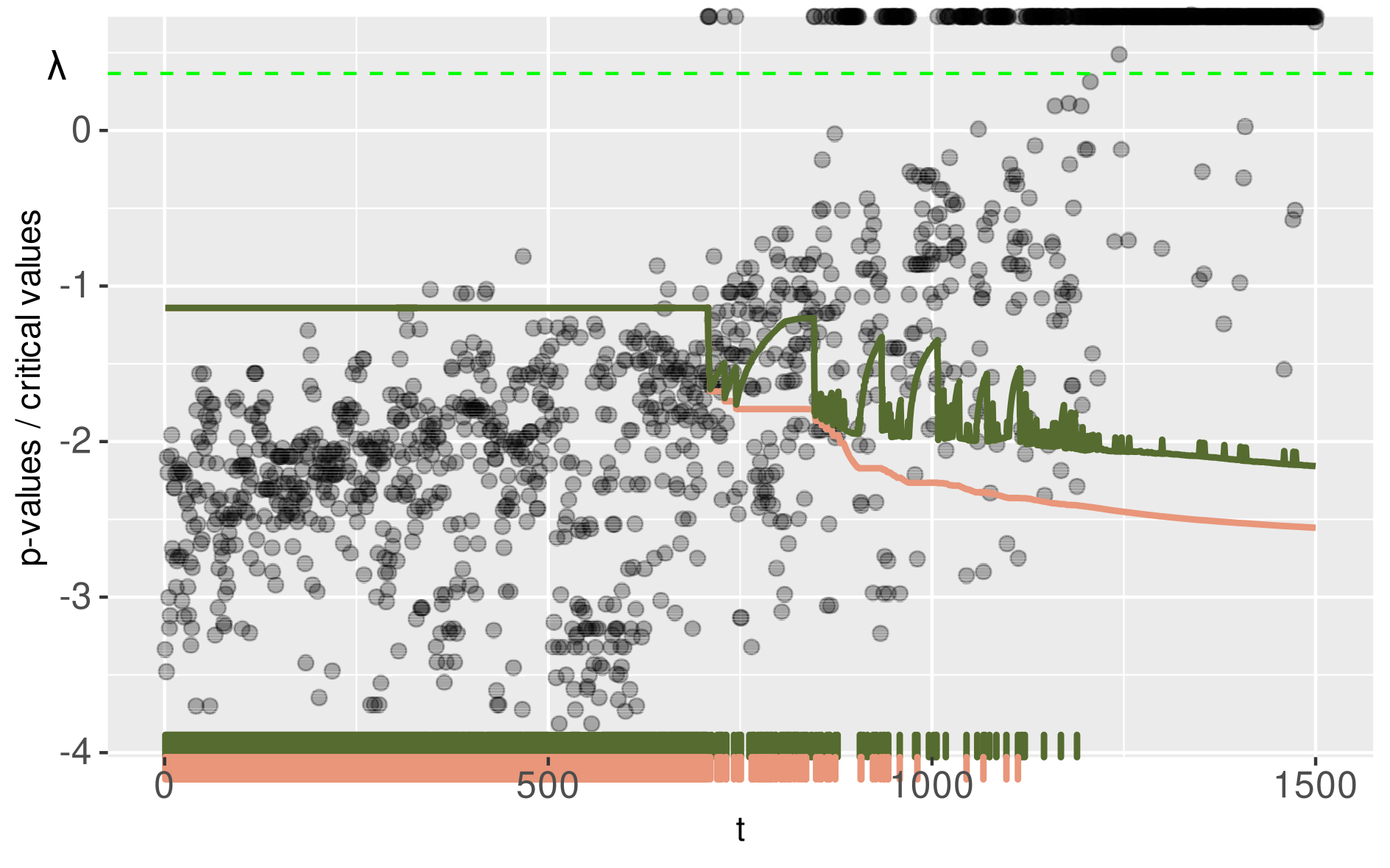}  
    \end{tabular}
    \end{center}
    \caption{Same as Figure~\ref{fig:mfdrimpcbasevsrewarded} but for female mice of IMPC data (see Section~\ref{sec:real_data_appli}). 
    }
    \label{fig:mfdrimpcbasevsrewarded:fem}
\end{figure}
\begin{figure}[h!]
    \begin{center}
    \begin{tabular}{cc}
    \includegraphics[width=.5\linewidth]{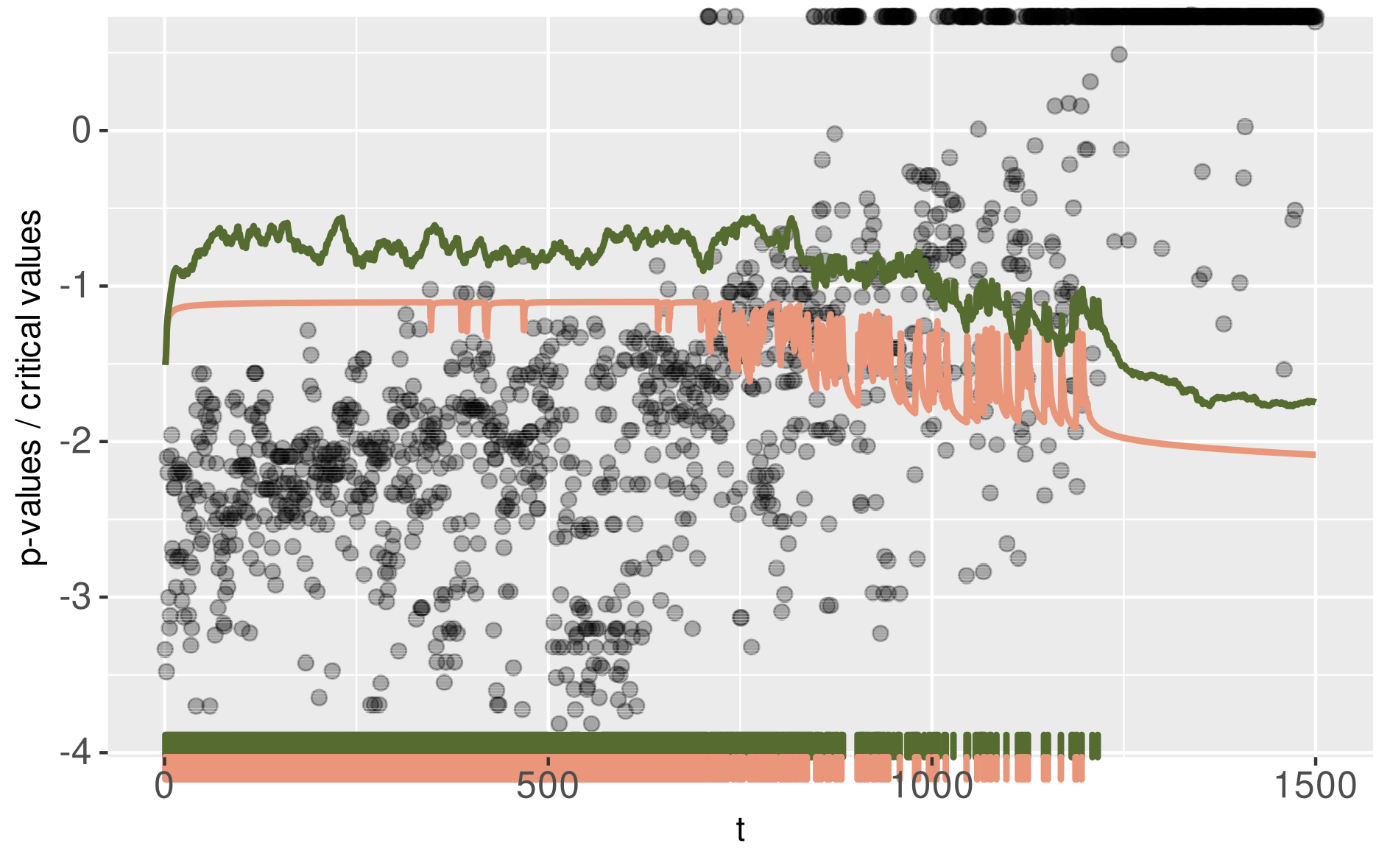} & \includegraphics[width=.5\linewidth]{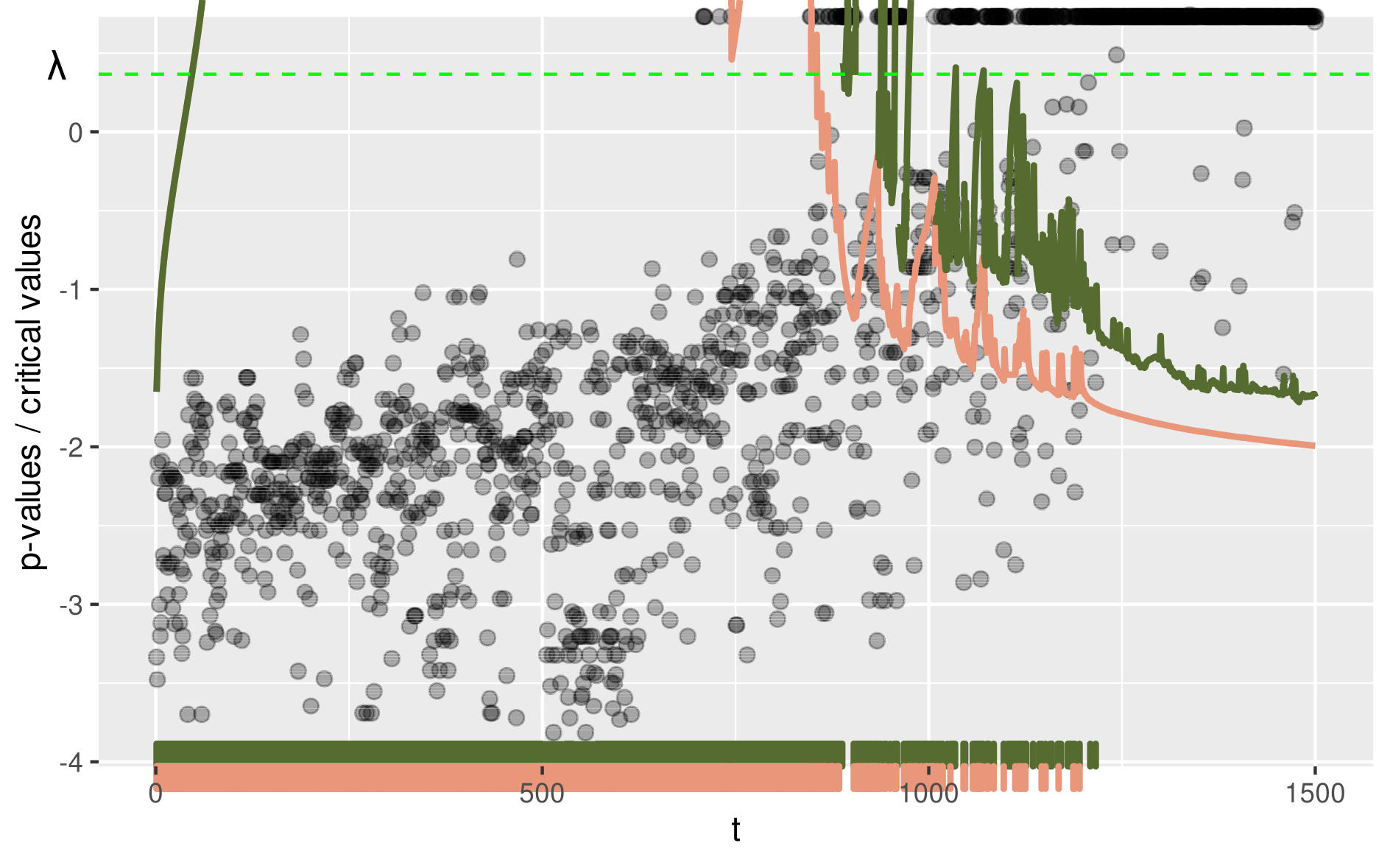}  
    \end{tabular}
    \end{center}
    \caption{Same as Figure~\ref{fig:fwerimpcbasevsrewarded} for female mice of IMPC data (see Section~\ref{sec:real_data_appli}). 
    } 
    \label{fig:fwerimpcbasevsrewarded:fem}
\end{figure}

\end{document}